\documentclass[11pt]{article}

\usepackage[a4paper,margin=1in]{geometry}
\usepackage{amsmath,amssymb,amsfonts,amsthm,mathtools,bm}
\usepackage{authblk}
\usepackage{booktabs,longtable,array,multirow}
\usepackage{enumitem}
\usepackage{algorithm}
\usepackage{algpseudocode}
\usepackage{hyperref}
\usepackage{xcolor}
\usepackage{caption}
\usepackage{bbm}
\usepackage{subcaption}
\usepackage{graphicx}
\usepackage{graphicx}
\usepackage{booktabs}
\usepackage{longtable}
\usepackage{array}
\usepackage{adjustbox}
\usepackage{multirow}
\usepackage{xcolor}
\usepackage{caption}
\usepackage{subcaption}
\hypersetup{
    colorlinks=true
}

\newtheorem{assumption}{Assumption}
\newtheorem{definition}{Definition}
\newtheorem{problem}{Problem}

\newtheorem{lemma}{Lemma}
\newtheorem{theorem}{Theorem}

\newcommand{\E}{\mathbb{E}}
\newcommand{\Prob}{\mathbb{P}}
\newcommand{\Var}{\mathrm{Var}}
\newcommand{\CVaR}{\mathrm{CVaR}}

\newcommand{\argmin}{\operatorname*{arg\,min}}

\newcommand{\one}{\mathbf{1}}
\newcommand{\cD}{\mathcal{D}}
\newcommand{\cG}{\mathcal{G}}
\newcommand{\cX}{\mathcal{X}}

\newcommand{\R}{\mathbb{R}}

\title{ Bayesian Multi-Topology Express Transportation Network Design
under Posterior Predictive Demand, Sorting-Efficiency and Delivery-Time Uncertainty}

\author[1,2]{Debashis Chatterjee\thanks{Corresponding author: \href{mailto:debashis.chatterjee@visva-bharati.ac.in}{debashis.chatterjee@visva-bharati.ac.in}}}
\affil[1]{Department of Statistics, Siksha-Bhavana (Institute of Science), Visva-Bharati University, Santiniketan, West Bengal, India}
\date{}

\begin{document}
\maketitle

\begin{abstract}
Express transportation network design is intrinsically uncertain because origin--destination demand, travel time, transportation cost, hub congestion, and realized sorting productivity fluctuate over time. Existing bi-objective express transportation network design models typically optimize operating cost and maximum arrival time under fixed input parameters, and then compare a finite set of admissible topological structures. Such deterministic designs may be cost-effective under nominal data but fragile under demand surges, route disruption, and hub-level productivity variation. This paper develops a Bayesian multi-topology express transportation network design framework in which origin--destination demand, travel time, cost multipliers, and effective hub sorting reliability are learned from historical or benchmark-calibrated data and propagated through posterior predictive scenarios. For each topology---fully connected, single-allocation hub-and-spoke, multiple-allocation hub-and-spoke, $R$-allocation hub-and-spoke, and their direct-link hybrid analogues---posterior scenarios are embedded into a mixed-integer design evaluation. The proposed design criterion combines posterior expected operating cost, posterior conditional value-at-risk of maximum arrival time, service-level reliability, hub hold-time reliability, and emission-aware penalties. A Bayesian multi-structure parallel design methodology is developed by integrating posterior simulation, sample-average approximation, topology-specific network optimization, Pareto comparison, and Bayes-risk selection. Theoretical results establish existence of a Bayes-optimal topology-design pair, almost-sure convergence of posterior scenario risks, consistency of scenario-optimal designs under a strict risk gap, and stability of topology selection under posterior concentration. Simulation experiments and a real CAB hub-location benchmark case study show that the proposed Bayesian method can deliberately trade a modest cost increase for substantial reductions in future arrival-time tail risk and major improvements in hold-time reliability. The framework therefore provides a statistically principled and operationally interpretable alternative to deterministic multi-topology express transportation network design.
\end{abstract}



\section{Introduction}

The rapid growth of courier and express-delivery systems has made transportation network design a central problem in logistics, operations research and data-driven industrial engineering. An express transportation network must decide which locations should operate as hubs, how non-hub nodes should be allocated to hubs, which direct links should be opened, how flows should be routed, and what sorting capacities should be installed at hubs. These decisions jointly determine operating cost, service time, resilience to congestion, and customer-perceived delivery reliability.

A recent deterministic formulation by Zhong, Wang, Li and Garc\'ia \cite{Zhong2023} considered a bi-objective express transportation network design problem under multiple topological structures. Their model compares seven commonly used network topologies: fully connected (FC), single-allocation hub-and-spoke (SAHS), multiple-allocation hub-and-spoke (MAHS), $R$-allocation hub-and-spoke (RAHS), direct-connected single-allocation hub-and-spoke (DSAHS), direct-connected multiple-allocation hub-and-spoke (DMAHS), and direct-connected $R$-allocation hub-and-spoke (DRAHS). Their main methodological contribution is a multi-structure parallel design methodology in which each topology is optimized separately, and then compared using a preference-based multi-objective selection rule. The model also treats hub sorting efficiency as a decision variable, which is highly relevant to the courier industry.

However, deterministic network design has a fundamental limitation: the values inserted into the optimization model are rarely known with certainty. Origin--destination parcel flows vary by weekday, season, festival, region, weather, macroeconomic activity and platform campaign. Travel times fluctuate because of traffic, road disruption, rainfall, and vehicle availability. Hub sorting productivity may degrade under labor shortage, machine downtime, congestion, or peak-season overload. Transportation cost and effective vehicle capacity can also vary. Thus, a topology that is optimal for a single nominal data table may be statistically fragile. A more defensible design should optimize with respect to the posterior predictive distribution of future logistics conditions.

The purpose of this paper is to formulate a Bayesian version of multi-topology express transportation network design. The central idea is to preserve the operational richness of the multi-topology mixed-integer design framework, but replace deterministic inputs by posterior predictive distributions. The proposed Bayesian model learns demand, travel time, cost and sorting productivity from historical data; generates posterior predictive scenarios; solves topology-specific scenario-based mixed-integer optimization problems; and selects a final topology using posterior Bayes risk, posterior Pareto dominance and service-level reliability. This converts deterministic network comparison into a Bayesian decision problem.

The main contributions of this paper are as follows.
\begin{enumerate}[label=(C\arabic*)]
    \item We propose a Bayesian multi-topology express transportation network design problem in which the topology class itself is a decision dimension.
    \item We formulate hierarchical Bayesian models for origin--destination demand, travel time, transportation cost and realized hub sorting capacity.
    \item We replace deterministic cost and maximum-arrival-time objectives by posterior expected cost and posterior tail risk of maximum arrival time, measured by conditional value-at-risk.
    \item We introduce Bayesian hold-time and service-level chance constraints, giving probabilistic guarantees for hub processing and late-delivery control.
    \item We develop a Bayesian multi-structure parallel design methodology based on posterior simulation, sample-average approximation and mixed-integer optimization.
    \item We prove existence, scenario convergence and topology-selection stability results under transparent assumptions.
\end{enumerate}

The paper is methodological. It is intended as a rigorous Bayesian extension of deterministic multi-topology express transportation network design rather than a mere parametric perturbation of an existing model.

\section{Related Literature}

The proposed model lies at the intersection of hub location, stochastic network design, Bayesian decision theory, and risk-averse stochastic optimization.

Hub location models originate from the recognition that flow consolidation through hub facilities can generate economies of scale. O'Kelly \cite{OKelly1987} formulated a quadratic integer programming model for interacting hub facilities. Campbell \cite{Campbell1992} studied location-allocation with transshipments and transportation economies of scale. Surveys and perspective papers by Alumur and Kara \cite{AlumurKara2008}, Farahani et al. \cite{Farahani2013}, Gelareh and Nickel \cite{GelarehNickel2011}, and Alumur et al. \cite{Alumur2021} provide broad accounts of hub location models, classifications and applications. Allocation strategies in hub networks, including single, multiple and restricted allocation, were systematically discussed by Yaman \cite{Yaman2011}. Capacitated and vehicle-type-aware hub network designs are treated by Serper and Alumur \cite{SerperAlumur2016}. The design of courier-service networks and frequency-delay models is considered by Lin, Zhao and Lin \cite{Lin2020}. Flexible multimodal hub network design is studied by Real et al. \cite{Real2021}. Stochastic single-allocation hub location is investigated by Rostami et al. \cite{Rostami2021}, while hazardous-material hub network design under uncertainty is considered by Zahiri and Suresh \cite{Zahiri2021}.

The deterministic baseline most directly related to this paper is Zhong et al. \cite{Zhong2023}, who developed a bi-objective mixed-integer nonlinear optimization formulation for express transportation networks under multiple topological structures. Their paper is important because it recognizes that topology choice is not a secondary modeling detail but a primary design decision. It also makes sorting efficiency a decision variable. The present paper builds on that insight, but moves from deterministic parameter optimization to posterior predictive decision optimization.

Bayesian decision theory provides the philosophical and mathematical basis for selecting actions under uncertainty by minimizing posterior expected loss \cite{Berger1985,Gelman2013}. Posterior simulation and Markov chain Monte Carlo are standard tools for evaluating posterior summaries and predictive distributions \cite{GelfandSmith1990,Gelman2013}. Stochastic programming provides the optimization counterpart, especially when uncertain data enter the objective or constraints \cite{Prekopa1995,Shapiro2009,RuszczynskiShapiro2003}. The sample-average approximation method justifies replacing an expectation over random scenarios by an empirical average over simulated scenarios \cite{Kleywegt2002}. Chance-constrained programming goes back to Charnes and Cooper \cite{CharnesCooper1959}; finite-scenario mixed-integer formulations for probabilistic constraints are developed, among others, by Luedtke, Ahmed and Nemhauser \cite{Luedtke2010}. Risk-averse optimization using conditional value-at-risk follows Rockafellar and Uryasev \cite{RockafellarUryasev2000}. Robust optimization gives an alternative non-Bayesian uncertainty paradigm \cite{BenTal2009}; the present work is Bayesian rather than worst-case robust because uncertainty is learned and updated through a posterior distribution.

\section{Deterministic Multi-Topology Network Design: Baseline}

Let $N=\{1,\ldots,n\}$ denote the set of demand nodes and let $H\subseteq N$ denote the set of candidate hubs. Let $M=\{1,\ldots,m\}$ denote vehicle types. For each ordered origin--destination pair $(i,j)\in N\times N$, let $w_{ij}$ denote parcel flow and $d_{ij}$ denote distance. A network topology class is denoted by
\[
\cG=\{G_1,\ldots,G_7\},
\]
where the seven elements correspond to FC, SAHS, MAHS, RAHS, DSAHS, DMAHS and DRAHS.

For a topology $G_g$, a deterministic design vector is written as
\[
x_g=(X,Y,Z,S,f,e),
\]
where:
\begin{itemize}
    \item $X_{ik}\in\{0,1\}$ equals one if node $i$ is assigned to hub $k$;
    \item $Y_{ijkl}\in\{0,1\}$ equals one if flow from $i$ to $j$ passes through hubs $k$ and $l$;
    \item $Z_{kl}\in\{0,1\}$ equals one if the inter-hub link $(k,l)$ is opened;
    \item $S_{ij}\in\{0,1\}$ equals one if the direct non-hub link $(i,j)$ is opened;
    \item $f_{ikl}\ge 0$ is the flow from origin $i$ passing through hub arc $(k,l)$;
    \item $e_k\ge 0$ is the installed sorting capacity at hub $k$.
\end{itemize}

The deterministic model minimizes two objectives:
\[
\min_{x_g\in\cX_g}\{C_g(x_g), A_g(x_g)\},
\]
where $C_g(x_g)$ is operating cost and $A_g(x_g)$ is maximum arrival time. The feasible region $\cX_g$ depends on topology-specific allocation, routing, flow-conservation and direct-link constraints. For example, under SAHS each non-hub node is allocated to exactly one hub:
\[
\sum_{k\in H}X_{ik}=1,\qquad i\in N,
\]
whereas under MAHS this equality is relaxed to permit multiple allocations. Under RAHS, one imposes
\[
\sum_{k\in H}X_{ik}\le R,\qquad i\in N.
\]
Hybrid structures allow selected direct links through $S_{ij}$, whereas pure hub-and-spoke structures disallow such direct non-hub connections.

The deterministic model is useful, but it assumes fixed future values of $w_{ij}$, travel time, cost and sorting performance. The Bayesian model below replaces these fixed quantities by posterior predictive random variables.

\section{Bayesian Data-Generating Model}

Suppose historical data are observed over periods $t=1,\ldots,T$. The data may include parcel counts, realized travel times, link costs, hub throughput, hub processing times, weather indicators and calendar variables. Denote the full dataset by
\[
\cD_T=\{w_{ij,t},\tau_{ij,t},c_{ij,t},Q_{k,t},B_{k,t},z_{ij,t},r_{ij,t}: i,j\in N,\ k\in H,\ t=1,\ldots,T\},
\]
where $Q_{k,t}$ is processed hub load and $B_{k,t}$ is observed processing or sorting time at hub $k$ in period $t$.

Let $\Theta$ denote all unknown statistical parameters. The posterior distribution is
\[
\pi(\Theta\mid \cD_T)
=
\frac{L(\cD_T\mid \Theta)\pi(\Theta)}
{\int L(\cD_T\mid \theta)\pi(\theta)\,d\theta}.
\]
The posterior predictive distribution of future logistics conditions is
\[
p(\Xi^{\mathrm{new}}\mid \cD_T)
=
\int p(\Xi^{\mathrm{new}}\mid \Theta)\pi(\Theta\mid \cD_T)\,d\Theta,
\]
where $\Xi^{\mathrm{new}}$ collects future demand, travel time, cost and realized hub sorting productivity.

\subsection{Origin--destination demand model}

Parcel demand is nonnegative, integer-valued and often overdispersed. We therefore use a negative-binomial hierarchical model:
\[
w_{ij,t}\mid \lambda_{ij,t},\phi_w
\sim
\mathrm{NegBin}(\lambda_{ij,t},\phi_w),
\]
parameterized so that
\[
\E(w_{ij,t}\mid\lambda_{ij,t})=\lambda_{ij,t},
\qquad
\Var(w_{ij,t}\mid\lambda_{ij,t})=
\lambda_{ij,t}+\frac{\lambda_{ij,t}^{2}}{\phi_w}.
\]
The log-intensity is
\[
\log \lambda_{ij,t}
=
\alpha_w
+
a_i^{(o)}
+
a_j^{(d)}
+
b_{r(i),r(j)}
+
\bm\beta_w^\top z_{ij,t}
+
u_t.
\]
Here $a_i^{(o)}$ is an origin effect, $a_j^{(d)}$ is a destination effect, $b_{r(i),r(j)}$ is a region-pair effect, $z_{ij,t}$ is a covariate vector, and $u_t$ is a temporal effect. We use the hierarchical priors
\[
a_i^{(o)}\sim N(0,\sigma_o^2),
\qquad
a_j^{(d)}\sim N(0,\sigma_d^2),
\qquad
b_{rs}\sim N(0,\sigma_b^2),
\]
and
\[
u_t=\rho u_{t-1}+\epsilon_t,\qquad
\epsilon_t\sim N(0,\sigma_u^2),\qquad |\rho|<1.
\]
Weakly informative priors may be used:
\[
\alpha_w\sim N(0,10^2),\quad
\bm\beta_w\sim N(\bm 0,10^2I),\quad
\sigma_o,\sigma_d,\sigma_b,\sigma_u\sim \mathrm{half}\text{-}t_\nu(0,s),
\quad
\phi_w\sim \mathrm{Gamma}(a_\phi,b_\phi).
\]

\subsection{Travel-time model}

Let $\tau_{ij,t}$ be the realized travel time from $i$ to $j$. A lognormal model is used:
\[
\log \tau_{ij,t}\mid \mu^\tau_{ij,t},\sigma_\tau^2
\sim
N(\mu^\tau_{ij,t},\sigma_\tau^2),
\]
where
\[
\mu^\tau_{ij,t}
=
\alpha_\tau
+
\delta_\tau\log d_{ij}
+
\bm\beta_\tau^\top r_{ij,t}
+
\xi_{ij}.
\]
Here $r_{ij,t}$ may contain rainfall, traffic intensity, road-class indicators, holiday effects or fuel-supply disruption indicators. The dyad effect $\xi_{ij}$ accounts for persistent link-level heterogeneity:
\[
\xi_{ij}\sim N(0,\sigma_\xi^2).
\]

\subsection{Transportation-cost model}

Let $c_{ijm,t}$ denote the variable cost per trip or per unit distance on arc $(i,j)$ using vehicle type $m$. A simple Bayesian model is
\[
\log c_{ijm,t}
\sim
N(\mu^c_{ijm,t},\sigma_c^2),
\]
where
\[
\mu^c_{ijm,t}
=
\alpha_c+\alpha_m^{(v)}+\delta_c\log d_{ij}
+\bm\beta_c^\top h_{ij,t}.
\]
The vector $h_{ij,t}$ may contain fuel price, toll index, road disruption, inflation index, and vehicle-type indicators. More complex models may allow $c_{ijm,t}$ to be correlated with travel time.

\subsection{Bayesian hub sorting-efficiency model}

Let $e_k$ denote installed sorting capacity at hub $k$, chosen by the optimization model. The actual realized fraction of installed capacity available in period $t$ is modeled by a random reliability factor $R_{k,t}\in(0,1]$:
\[
\widetilde e_{k,t}=e_k R_{k,t}.
\]
We model
\[
\mathrm{logit}(R_{k,t})
\sim
N(\mu_{R,k},\sigma_R^2),
\qquad
\mu_{R,k}\sim N(\mu_R,\sigma_{\mu R}^2).
\]
Let $L_{k,t}(x_g,w_t)$ denote the random parcel load assigned to hub $k$ under topology $G_g$ and decision $x_g$. A congestion-adjusted sorting time is modeled as
\[
H_{k,t}(x_g,\Xi_t)
=
\frac{L_{k,t}(x_g,w_t)}{\widetilde e_{k,t}}
+
\gamma_k
\left\{
\frac{L_{k,t}(x_g,w_t)}{\widetilde e_{k,t}}
\right\}^{2},
\]
where $\gamma_k\ge 0$ is a hub-specific congestion coefficient. The quadratic term is a parsimonious way to express nonlinear delay under high utilization; if congestion information is not available, one may set $\gamma_k=0$ and recover a linear sorting-time model.

\section{Bayesian Multi-Topology Design Variables and Feasible Sets}

For each topology $G_g$, let $\cX_g$ denote the feasible set of mixed-integer design decisions satisfying the corresponding structural constraints. The global design space is the disjoint union
\[
\cX=\bigcup_{g=1}^{7}\{g\}\times \cX_g.
\]
A complete design is a pair $(g,x_g)$, where $g$ selects the topology class and $x_g$ selects hubs, links, assignments, flows and capacities under that class.

For compactness, we write the general constraints as:
\[
x_g\in\cX_g
\quad\Longleftrightarrow\quad
\begin{cases}
\text{demand satisfaction constraints,}\\
\text{flow conservation constraints,}\\
\text{hub-opening and assignment compatibility constraints,}\\
\text{allocation-strategy constraints,}\\
\text{direct-link admissibility constraints,}\\
\text{binary and nonnegativity restrictions.}
\end{cases}
\]

The essential topology-dependent constraints are summarized in Table~\ref{tab:topology}.

\begin{table}[htbp]
\centering
\caption{Topology-dependent structural restrictions.}
\label{tab:topology}
\begin{tabular}{lll}
\toprule
Topology & Allocation rule & Direct non-hub links \\
\midrule
FC & no hub allocation required & all OD pairs direct \\
SAHS & $\sum_{k\in H}X_{ik}=1$ & not allowed \\
MAHS & multiple hub assignments allowed & not allowed \\
RAHS & $\sum_{k\in H}X_{ik}\le R$ & not allowed \\
DSAHS & $\sum_{k\in H}X_{ik}=1$ & allowed selectively \\
DMAHS & multiple hub assignments allowed & allowed selectively \\
DRAHS & $\sum_{k\in H}X_{ik}\le R$ & allowed selectively \\
\bottomrule
\end{tabular}
\end{table}

\section{Posterior Predictive Objectives}

For a future scenario $\Xi$, define the random operating cost of design $(g,x_g)$ by
\[
C_g(x_g,\Xi)
=
F_g(x_g)
+
V_g(x_g,\Xi)
+
S_g(x_g,\Xi)
+
K_g(x_g),
\]
where $F_g$ is fixed facility and line-opening cost, $V_g$ is transportation cost, $S_g$ is sorting cost, and $K_g$ is the cost of installing sorting capacity.

The posterior expected operating cost is
\[
\mathcal{C}_g(x_g)
=
\E\{C_g(x_g,\Xi)\mid \cD_T\}.
\]

Let $A_{ij,g}(x_g,\Xi)$ denote the arrival time for OD pair $(i,j)$ under topology $G_g$. For a hub-routed shipment using first hub $k$ and second hub $l$, a representative expression is
\[
A_{ij,g}(x_g,\Xi)
=
\tau_{ik}
+
H_{k}(x_g,\Xi)
+
\tau_{kl}
+
H_l(x_g,\Xi)
+
\tau_{lj},
\]
with terms omitted or modified according to topology. For a direct link,
\[
A_{ij,g}^{\mathrm{direct}}(x_g,\Xi)=\tau_{ij}+s_i,
\]
where $s_i$ is local service time. The maximum arrival time is
\[
A_g(x_g,\Xi)
=
\max_{i,j\in N} A_{ij,g}(x_g,\Xi).
\]

The posterior mean maximum arrival time is
\[
\mathcal{A}_g(x_g)
=
\E\{A_g(x_g,\Xi)\mid \cD_T\}.
\]
However, express logistics is more sensitive to late-tail performance than to average performance. 
We therefore define the posterior tail-risk objective as the conditional value-at-risk of the 
posterior predictive maximum arrival time:
\begin{equation}
\label{eq:posterior_tail_objective}
\mathcal{T}_{g,\alpha}(x_g)
=
\CVaR_\alpha\!\left(A_g(x_g,\Xi)\mid \cD_T\right),
\qquad 0<\alpha<1.
\end{equation}
Following the Rockafellar--Uryasev optimization representation of conditional value-at-risk
\cite{RockafellarUryasev2000}, equation~\eqref{eq:posterior_tail_objective} can be written as
\begin{equation}
\label{eq:posterior_cvar_ru}
\mathcal{T}_{g,\alpha}(x_g)
=
\min_{\zeta_g\in\R}
\left[
\zeta_g
+
\frac{1}{1-\alpha}
\E\left\{
\left(A_g(x_g,\Xi)-\zeta_g\right)_+
\mid \cD_T
\right\}
\right],
\end{equation}
where \((u)_+=\max\{u,0\}\).  The scalar \(\zeta_g\) is an auxiliary
VaR-type threshold for topology \(G_g\).  Thus, minimizing
\(\mathcal{T}_{g,\alpha}(x_g)\) penalizes not only the typical maximum
arrival time, but also the expected excess delay in the worst
\(100(1-\alpha)\%\) posterior predictive delivery scenarios.
\section{Bayesian Chance Constraints}

The deterministic hold-time constraint requires hub processing time to be below a fixed threshold. In the Bayesian model the processing time is random. Hence, for every opened hub $k$, impose
\[
\Prob\left\{
H_k(x_g,\Xi)\le d_t
\mid \cD_T
\right\}
\ge 1-\varepsilon_k,
\qquad k\in H.
\]
Similarly, if $T^\star$ is a target maximum delivery time, impose
\[
\Prob\left\{
A_g(x_g,\Xi)\le T^\star
\mid \cD_T
\right\}
\ge 1-\varepsilon_T.
\]
These constraints have direct managerial interpretation: after learning from historical data, the chosen design must satisfy a specified posterior reliability level.

\begin{definition}[Posterior service reliability]
For a topology-design pair $(g,x_g)$ and service target $T^\star$, define
\[
R_T(g,x_g)
=
\Prob\{A_g(x_g,\Xi)\le T^\star\mid \cD_T\}.
\]
A design is $(T^\star,\varepsilon_T)$-reliable if $R_T(g,x_g)\ge 1-\varepsilon_T$.
\end{definition}

\section{Bayesian Bi-Objective and Scalar Bayes-Risk Formulations}

The Bayesian bi-objective problem is
\begin{equation}
\label{eq:biobj}
\min_{g\in\{1,\ldots,7\},\,x_g\in\cX_g}
\left\{
\mathcal{C}_g(x_g),
\mathcal{T}_{g,\alpha}(x_g)
\right\},
\end{equation}
subject to Bayesian chance constraints.

For decision-making, one may scalarize the two objectives by posterior Bayes risk:
\begin{equation}
\label{eq:bayesrisk}
\rho(g,x_g)
=
p_c\,\frac{\mathcal{C}_g(x_g)-C_{\min}}{C_{\max}-C_{\min}}
+
p_t\,\frac{\mathcal{T}_{g,\alpha}(x_g)-T_{\min}}{T_{\max}-T_{\min}}
+
p_r\,\Prob\{A_g(x_g,\Xi)>T^\star\mid \cD_T\},
\end{equation}
where $p_c,p_t,p_r\ge0$ and $p_c+p_t+p_r=1$. The constants $C_{\min},C_{\max},T_{\min},T_{\max}$ are normalization constants computed over a candidate solution set. If normalization is not desired, the unscaled version
\[
\rho_0(g,x_g)
=
\omega_c\mathcal{C}_g(x_g)
+
\omega_t\mathcal{T}_{g,\alpha}(x_g)
+
\omega_r\Prob\{A_g(x_g,\Xi)>T^\star\mid \cD_T\}
\]
may be used.

\begin{problem}[Bayesian multi-topology express transportation network design]
Find
\[
(\widehat g,\widehat x)
\in
\argmin_{g,x_g}
\rho(g,x_g)
\]
subject to $x_g\in\cX_g$ and the posterior chance constraints.
\end{problem}

\section{Posterior Scenario Approximation}

Draw $B$ posterior predictive scenarios
\[
\Xi^{(1)},\ldots,\Xi^{(B)}
\sim p(\Xi\mid \cD_T).
\]
The posterior expected cost is approximated by
\[
\widehat{\mathcal{C}}_{g,B}(x_g)
=
\frac{1}{B}\sum_{b=1}^{B}C_g(x_g,\Xi^{(b)}).
\]
The CVaR term is approximated by introducing variables $\zeta_g$ and $u_b$:
\[
\widehat{\CVaR}_{\alpha,B}
=
\min_{\zeta_g,u_1,\ldots,u_B}
\left[
\zeta_g+\frac{1}{(1-\alpha)B}\sum_{b=1}^{B}u_b
\right],
\]
subject to
\[
u_b\ge A_g(x_g,\Xi^{(b)})-\zeta_g,\qquad
u_b\ge0,\qquad b=1,\ldots,B.
\]
The finite-scenario optimization problem is therefore:
\begin{align}
\min_{x_g,\zeta_g,u}\quad
&
p_c\frac{1}{B}\sum_{b=1}^{B}C_g(x_g,\Xi^{(b)})
+
p_t\left[
\zeta_g+\frac{1}{(1-\alpha)B}\sum_{b=1}^{B}u_b
\right]
+
p_r\frac{1}{B}\sum_{b=1}^{B}v_b
\label{eq:saa}\\
\text{s.t.}\quad
&
u_b\ge A_g(x_g,\Xi^{(b)})-\zeta_g,\quad u_b\ge0,
\label{eq:cvarcon}\\
&
v_b\ge \one\{A_g(x_g,\Xi^{(b)})>T^\star\},
\label{eq:lateindicator}\\
&
x_g\in\cX_g,
\end{align}
with scenario versions of the chance constraints:
\[
\frac{1}{B}\sum_{b=1}^{B}
\one\{H_k(x_g,\Xi^{(b)})\le d_t\}
\ge 1-\varepsilon_k,
\]
and
\[
\frac{1}{B}\sum_{b=1}^{B}
\one\{A_g(x_g,\Xi^{(b)})\le T^\star\}
\ge 1-\varepsilon_T.
\]
In practice, indicator constraints can be implemented using binary variables and big-$M$ formulations.

\section{Bayesian Multi-Structure Parallel Design Methodology}

The proposed Bayesian algorithm is called Bayesian Multi-Structure Parallel Design Methodology (B-MS-PDM).

\begin{algorithm}[htbp]
\caption{Bayesian Multi-Structure Parallel Design Methodology}
\label{alg:bmspdm}
\begin{algorithmic}[1]
\Require Historical logistics data $\cD_T$; topology set $\cG$; posterior scenario size $B$; risk level $\alpha$; reliability tolerances $\varepsilon_k,\varepsilon_T$.
\State Fit Bayesian demand, travel-time, cost and sorting-productivity models.
\State Draw posterior samples $\Theta^{(1)},\ldots,\Theta^{(B)}\sim \pi(\Theta\mid \cD_T)$.
\State Generate posterior predictive scenarios $\Xi^{(1)},\ldots,\Xi^{(B)}$.
\For{$g=1,\ldots,7$}
    \State Construct topology-specific feasible region $\cX_g$.
    \State Solve the scenario-based mixed-integer optimization problem for topology $G_g$.
    \State Store feasible near-optimal solutions and compute posterior cost, posterior CVaR, posterior late probability and service reliability.
    \State Extract the non-dominated posterior Pareto set $\widehat P_g$.
\EndFor
\State Form $\widehat P=\bigcup_{g=1}^{7}\widehat P_g$.
\State Select $(\widehat g,\widehat x)$ by minimum posterior Bayes risk or by maximum posterior probability of being best.
\Ensure Bayesian topology-design pair $(\widehat g,\widehat x)$ and posterior reliability summaries.
\end{algorithmic}
\end{algorithm}

A useful diagnostic is the posterior probability that topology $G_g$ is best:
\[
\widehat\pi_g^{\mathrm{best}}
=
\frac{1}{B}\sum_{b=1}^{B}
\one\left[
g=
\argmin_{\ell\in\{1,\ldots,7\}}
\rho_b(\ell,\widehat x_\ell)
\right],
\]
where $\rho_b$ is the scenario-specific loss. This quantity replaces deterministic topology ranking by Bayesian model-based topology evidence.


\section{Theoretical Properties}
\label{sec:theory}

This section records theoretical properties of the proposed Bayesian multi-topology design problem.  The purpose is not to claim computational tractability for arbitrary large-scale instances, but to establish that the posterior decision problem is mathematically well-defined and that the finite posterior-scenario approximation used in the algorithm is statistically consistent.

Let
\[
\mathfrak X
=
\bigcup_{g=1}^{7}\{g\}\times \mathcal X_g
\]
denote the global topology-design space.  A generic feasible design is written as \(z=(g,x_g)\in\mathfrak X\).  Let \(\Xi\sim P_T(\cdot)=P(\cdot\mid\mathcal D_T)\) denote one posterior predictive logistics scenario, including future OD demand, travel times, cost multipliers, and realized hub reliability.  For a design \(z\), let \(C(z,\Xi)\) be operating cost, \(A(z,\Xi)\) the maximum arrival time, and \(H_k(z,\Xi)\) the hub hold-time load at candidate hub \(k\).  For a fixed risk level \(\alpha\in(0,1)\), define
\[
\mathcal T_\alpha(z)
=
\mathrm{CVaR}_\alpha\{A(z,\Xi)\mid \mathcal D_T\}.
\]
The posterior Bayes risk is
\begin{equation}
\label{eq:formal_bayes_risk}
\rho_T(z)
=
w_C\,\mathbb E_T\{C(z,\Xi)\}
+
w_A\,\mathcal T_\alpha(z)
+
w_L\,P_T\{A(z,\Xi)>T^\star\}
+
w_H\sum_{k\in H}P_T\{H_k(z,\Xi)>d_t\},
\end{equation}
where \(w_C,w_A,w_L,w_H\ge0\), and \(\mathbb E_T\) denotes expectation under the posterior predictive distribution \(P_T\).

\begin{assumption}[Finite engineering design space]
\label{ass:finite_design}
For each topology \(G_g\), the feasible set \(\mathcal X_g\) is nonempty and finite.  This holds, for example, when all binary topology, hub, allocation, routing and direct-link decisions are finite, and when installed sorting capacities are selected from a bounded engineering grid.
\end{assumption}

\begin{assumption}[Posterior integrability]
\label{ass:integrability}
For every feasible design \(z\in\mathfrak X\),
\[
\mathbb E_T|C(z,\Xi)|<\infty,
\qquad
\mathbb E_T|A(z,\Xi)|<\infty,
\]
and, for every candidate hub \(k\in H\),
\[
P_T\{|H_k(z,\Xi)|<\infty\}=1.
\]
\end{assumption}

\begin{assumption}[Positive reliability thresholds]
\label{ass:thresholds}
The service threshold \(T^\star\) and hub hold threshold \(d_t\) are finite constants.  The violation indicators
\[
\mathbbm 1\{A(z,\Xi)>T^\star\},
\qquad
\mathbbm 1\{H_k(z,\Xi)>d_t\}
\]
are measurable for every \(z\in\mathfrak X\) and \(k\in H\).
\end{assumption}

\begin{definition}[Bayes-optimal topology-design pair]
\label{def:bayes_optimal}
A feasible design \(z_T^\star=(g_T^\star,x_{g_T^\star}^\star)\in\mathfrak X\) is called Bayes-optimal if
\[
z_T^\star\in \arg\min_{z\in\mathfrak X}\rho_T(z).
\]
\end{definition}

\begin{theorem}[Existence of a Bayes-optimal topology-design pair]
\label{thm:existence}
Under Assumptions~\ref{ass:finite_design}--\ref{ass:thresholds}, the Bayesian multi-topology express transportation network design problem admits at least one Bayes-optimal topology-design pair.
\end{theorem}

For posterior-scenario computation, let
\[
\Xi^{(1)},\ldots,\Xi^{(B)}
\stackrel{\mathrm{iid}}{\sim}P_T
\]
be posterior predictive scenarios.  Define the empirical cost and violation probabilities by
\[
\widehat C_B(z)=\frac1B\sum_{b=1}^{B}C(z,\Xi^{(b)}),
\]
\[
\widehat p_{L,B}(z)=\frac1B\sum_{b=1}^{B}\mathbbm 1\{A(z,\Xi^{(b)})>T^\star\},
\]
and
\[
\widehat p_{H,k,B}(z)=\frac1B\sum_{b=1}^{B}\mathbbm 1\{H_k(z,\Xi^{(b)})>d_t\}.
\]
The empirical CVaR estimator is
\begin{equation}
\label{eq:empirical_cvar_formal}
\widehat{\mathcal T}_{\alpha,B}(z)
=
\min_{\zeta\in\mathbb R}
\left[
\zeta+
\frac{1}{(1-\alpha)B}
\sum_{b=1}^{B}
\{A(z,\Xi^{(b)})-\zeta\}_+
\right].
\end{equation}
The empirical posterior risk is
\begin{equation}
\label{eq:empirical_risk_formal}
\widehat\rho_{T,B}(z)
=
w_C\widehat C_B(z)
+
w_A\widehat{\mathcal T}_{\alpha,B}(z)
+
w_L\widehat p_{L,B}(z)
+
w_H\sum_{k\in H}\widehat p_{H,k,B}(z).
\end{equation}

\begin{lemma}[Consistency of the empirical CVaR functional]
\label{lem:cvar_consistency}
Let \(Y\) be a real-valued random variable satisfying \(\mathbb E|Y|<\infty\), and let \(Y_1,Y_2,\ldots\) be iid copies of \(Y\).  Define
\[
\mathrm{CVaR}_\alpha(Y)
=
\min_{\zeta\in\mathbb R}
\left[
\zeta+\frac{1}{1-\alpha}\mathbb E(Y-\zeta)_+
\right],
\]
and
\[
\widehat{\mathrm{CVaR}}_{\alpha,B}
=
\min_{\zeta\in\mathbb R}
\left[
\zeta+\frac{1}{(1-\alpha)B}\sum_{b=1}^{B}(Y_b-\zeta)_+
\right].
\]
Then
\[
\widehat{\mathrm{CVaR}}_{\alpha,B}
\longrightarrow
\mathrm{CVaR}_\alpha(Y)
\quad\text{almost surely as }B\to\infty.
\]
\end{lemma}

\begin{theorem}[Uniform almost-sure convergence of posterior scenario risks]
\label{thm:uniform_saa}
Under Assumptions~\ref{ass:finite_design}--\ref{ass:thresholds},
\[
\sup_{z\in\mathfrak X}
\left|
\widehat\rho_{T,B}(z)-\rho_T(z)
\right|
\longrightarrow 0
\quad\text{almost surely as }B\to\infty.
\]
\end{theorem}

\begin{theorem}[Consistency of scenario-optimal designs under a strict Bayes-risk gap]
\label{thm:argmin_consistency}
Assume the conditions of Theorem~\ref{thm:uniform_saa}.  Suppose the Bayes-optimal design \(z_T^\star\) is unique and satisfies the strict separation condition
\[
\Delta_T
=
\min_{z\in\mathfrak X:\,z\ne z_T^\star}
\{\rho_T(z)-\rho_T(z_T^\star)\}
>0.
\]
Let
\[
\widehat z_{T,B}\in\arg\min_{z\in\mathfrak X}\widehat\rho_{T,B}(z)
\]
be any empirical posterior-scenario optimizer.  Then
\[
\widehat z_{T,B}=z_T^\star
\]
for all sufficiently large \(B\), almost surely.
\end{theorem}

\begin{assumption}[Posterior predictive concentration]
\label{ass:posterior_concentration}
There exists a limiting predictive distribution \(P_0\) such that, as \(T\to\infty\), the posterior predictive distribution \(P_T=P(\cdot\mid\mathcal D_T)\) satisfies
\[
W_1(P_T,P_0)\longrightarrow 0
\]
in \(P_0\)-probability, where \(W_1\) denotes the first Wasserstein distance on the space of logistics scenarios.
\end{assumption}

\begin{assumption}[Continuity and domination of design losses]
\label{ass:continuity_domination}
For each feasible design \(z\in\mathfrak X\), the functions \(C(z,\xi)\) and \(A(z,\xi)\) are continuous in the scenario argument \(\xi\) except possibly on a \(P_0\)-null set, and are dominated by an integrable envelope with respect to all sufficiently large posterior predictive distributions.  In addition, for every \(z\in\mathfrak X\),
\[
P_0\{A(z,\Xi)=T^\star\}=0,
\qquad
P_0\{H_k(z,\Xi)=d_t\}=0
\quad\text{for all }k\in H.
\]
\end{assumption}

\begin{theorem}[Stability of topology selection under posterior concentration]
\label{thm:posterior_topology_stability}
Assume Assumptions~\ref{ass:finite_design}--\ref{ass:continuity_domination}.  Let
\[
\rho_0(z)
=
w_C\,\mathbb E_0\{C(z,\Xi)\}
+
w_A\,\mathrm{CVaR}_{\alpha,0}\{A(z,\Xi)\}
+
w_L\,P_0\{A(z,\Xi)>T^\star\}
+
w_H\sum_{k\in H}P_0\{H_k(z,\Xi)>d_t\}
\]
be the limiting predictive risk.  Suppose \(z_0^\star=(g_0^\star,x_0^\star)\) is the unique minimizer of \(\rho_0\) over \(\mathfrak X\).  Then any posterior Bayes-risk minimizer
\[
z_T^\star\in\arg\min_{z\in\mathfrak X}\rho_T(z)
\]
satisfies
\[
P\{z_T^\star=z_0^\star\}\longrightarrow 1
\quad\text{as }T\to\infty.
\]
Consequently, the selected topology \(g_T^\star\) converges in probability to the limiting optimal topology \(g_0^\star\).
\end{theorem}
\section{Simulation Verification of the Bayesian Multi-Topology Design Methodology}
\label{sec:simulation_verification}

This section verifies the proposed Bayesian multi-topology express transportation network design methodology through a controlled simulation study.  The purpose of the experiment is not merely to reproduce a deterministic network-design calculation, but to examine whether posterior-predictive decision making can select a network design that is more reliable under future operational uncertainty.  The simulation is therefore deliberately constructed to include three empirically plausible logistics regimes: ordinary operating days, sale/surge days, and storm/disruption days.  These regimes induce heterogeneous origin--destination demand, heavy-tailed travel times, uncertain hub sorting productivity, and cost fluctuations.  Such features are precisely the circumstances under which a deterministic nominal design may appear attractive under mean conditions but may perform poorly in the posterior tail.

Let \(N=\{1,\ldots,n\}\) denote the set of demand nodes and let \(H\subset N\) be the set of candidate hubs.  For every directed origin--destination pair \((i,j)\), \(i\neq j\), the historical parcel count is denoted by \(w_{ij,t}\), the realized travel time by \(\tau_{ij,t}\), and the cost multiplier by \(\kappa_{ij,t}\).  For each hub \(k\in H\), the observed hub productivity/reliability proxy is denoted by \(R_{k,t}\).  The posterior predictive simulation is based on four conjugate Bayesian components: a Gamma--Poisson model for OD demand, a lognormal Normal--Inverse-Gamma model for travel time, a Beta model for hub reliability, and a lognormal Normal--Inverse-Gamma model for cost multipliers.  Posterior predictive scenarios are then embedded into the seven topology classes considered in the network-design model:
\[
\mathrm{FC},\quad \mathrm{SAHS},\quad \mathrm{MAHS},\quad \mathrm{RAHS},\quad
\mathrm{DSAHS},\quad \mathrm{DMAHS},\quad \mathrm{DRAHS}.
\]
For each topology, a finite candidate design set is generated by varying hub subsets, sorting-capacity multipliers, and direct-connection intensities.  Each candidate is evaluated over posterior predictive scenarios using posterior expected cost, posterior conditional value-at-risk of maximum arrival time, service reliability, hub-hold reliability, and an emission index.  The selected Bayesian design is the candidate minimizing the posterior risk score
\begin{equation}
\label{eq:sim_bayes_score}
\mathcal{R}(x)
=
\omega_c\,\widetilde C(x)
+
\omega_t\,\widetilde{\mathrm{CVaR}}_{0.90}\{A(x)\}
+
\omega_e\,\widetilde E(x)
+
\lambda_s\{1-\widehat p_s(x)\}
+
\lambda_h\{1-\widehat p_h(x)\},
\end{equation}
where \(\widetilde C(x)\), \(\widetilde{\mathrm{CVaR}}_{0.90}\{A(x)\}\), and \(\widetilde E(x)\) are normalized posterior expected cost, normalized posterior tail risk of maximum arrival time, and normalized emission index, respectively.  The quantities \(\widehat p_s(x)\) and \(\widehat p_h(x)\) denote estimated posterior service reliability and hub-hold reliability.  In the reported experiment, \((\omega_c,\omega_t,\omega_e)=(0.25,0.55,0.20)\), so the decision rule is intentionally risk-aware while still penalizing cost and emissions.

\subsection{Synthetic operating environment and posterior learning}
\label{subsec:synthetic_environment}

Table~\ref{tab:synthetic_summary} summarizes the generated logistics environment.  The simulated system contains nine demand nodes and three candidate hubs.  The historical training window contains 100 days and 72 directed OD pairs.  The regime composition is intentionally unbalanced: most days are normal, but a non-negligible fraction corresponds to sale/surge and storm/disruption days.  This matters statistically because the posterior predictive distribution must learn both routine variation and rare-but-important stress behavior.

\begin{table}[H]
\centering
\caption{Summary of the synthetic logistics environment used for posterior-predictive verification.}
\label{tab:synthetic_summary}
\begin{adjustbox}{max width=0.98\textwidth}
\begin{tabular}{ll}
\toprule
Quantity & Value \\
\midrule
Number of nodes & 9 \\
Candidate hubs (1-based indices) & \(\{4,8,3\}\) \\
Historical days & 100 \\
Directed OD pairs & 72 \\
Normal/surge/storm days & normal: 67; sale/surge: 24; storm/disruption: 9 \\
Mean daily OD demand & 14.307 \\
Median daily OD demand & 10.0 \\
Maximum daily OD demand & 187 \\
Mean observed travel time (hours) & 9.572 \\
95th percentile observed travel time (hours) & 23.384 \\
Mean cost multiplier & 1.048 \\
Mean observed hub reliability & 0.876 \\
\bottomrule
\end{tabular}
\end{adjustbox}
\end{table}

The posterior summaries in Table~\ref{tab:posterior_summary} confirm that the Bayesian updating step is not a cosmetic addition.  The OD demand model learns a posterior mean daily OD intensity of 14.3069, while the travel-time posterior has mean log travel time 1.9082.  The hub reliability posterior yields an average reliability of 0.8671 across candidate hubs.  These learned posterior quantities are subsequently propagated into the decision layer rather than replaced by point estimates.

\begin{table}[H]
\centering
\caption{Manual Bayesian posterior updating summary for the simulation experiment.}
\label{tab:posterior_summary}
\begin{adjustbox}{max width=0.99\textwidth}
\begin{tabular}{llll}
\toprule
Component & Bayesian model & Key posterior quantity & Value \\
\midrule
OD demand & Gamma--Poisson posterior predictive & Mean posterior daily OD intensity & 14.3069 \\
Travel time & Lognormal Normal--Inverse-Gamma posterior & Mean posterior log travel time & 1.9082 \\
Hub reliability & Beta posterior from sorting-success pseudo-counts & Mean reliability across candidate hubs & 0.8671 \\
Cost multiplier & Lognormal Normal--Inverse-Gamma posterior & Mean posterior log cost multiplier & 0.0404 \\
\bottomrule
\end{tabular}
\end{adjustbox}
\end{table}

Figure~\ref{fig:geo_hubs} shows the spatial configuration of the synthetic logistics network and the selected candidate hub locations.  Figure~\ref{fig:demand_heatmap} displays the posterior mean OD demand matrix.  The non-uniformity of the heatmap is important: a multi-topology model is useful precisely because a symmetric or homogeneous demand field would not expose meaningful differences among hub-spoke and hybrid network structures.

\begin{figure}[H]
\centering
\includegraphics[width=0.7\textwidth]{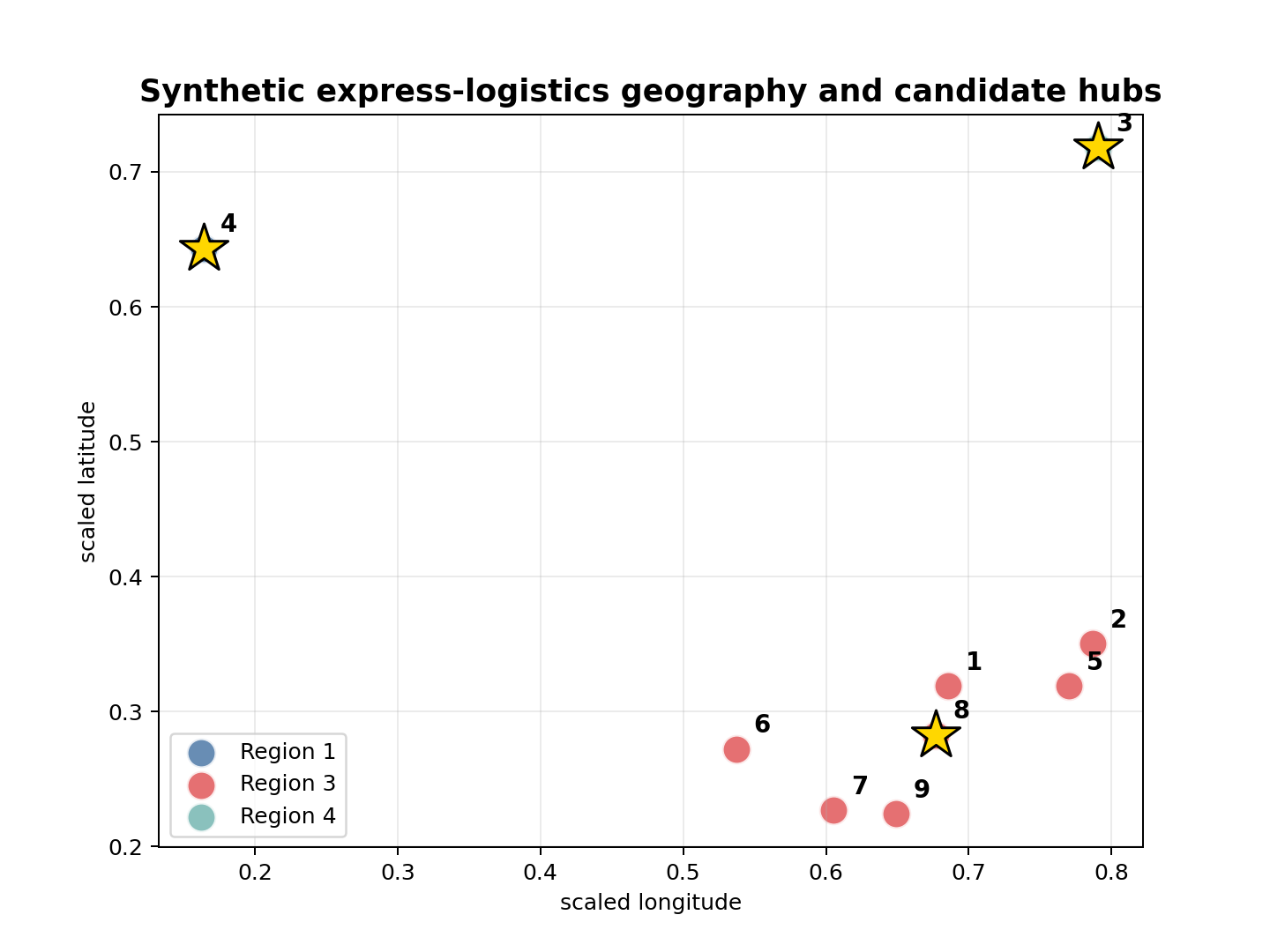}
\caption{Synthetic geography of demand nodes and candidate hubs.  The network is deliberately heterogeneous, so that topology, hub selection, and direct links have nontrivial effects on posterior cost and arrival-time risk.}
\label{fig:geo_hubs}
\end{figure}

\begin{figure}[H]
\centering
\includegraphics[width=0.82\textwidth]{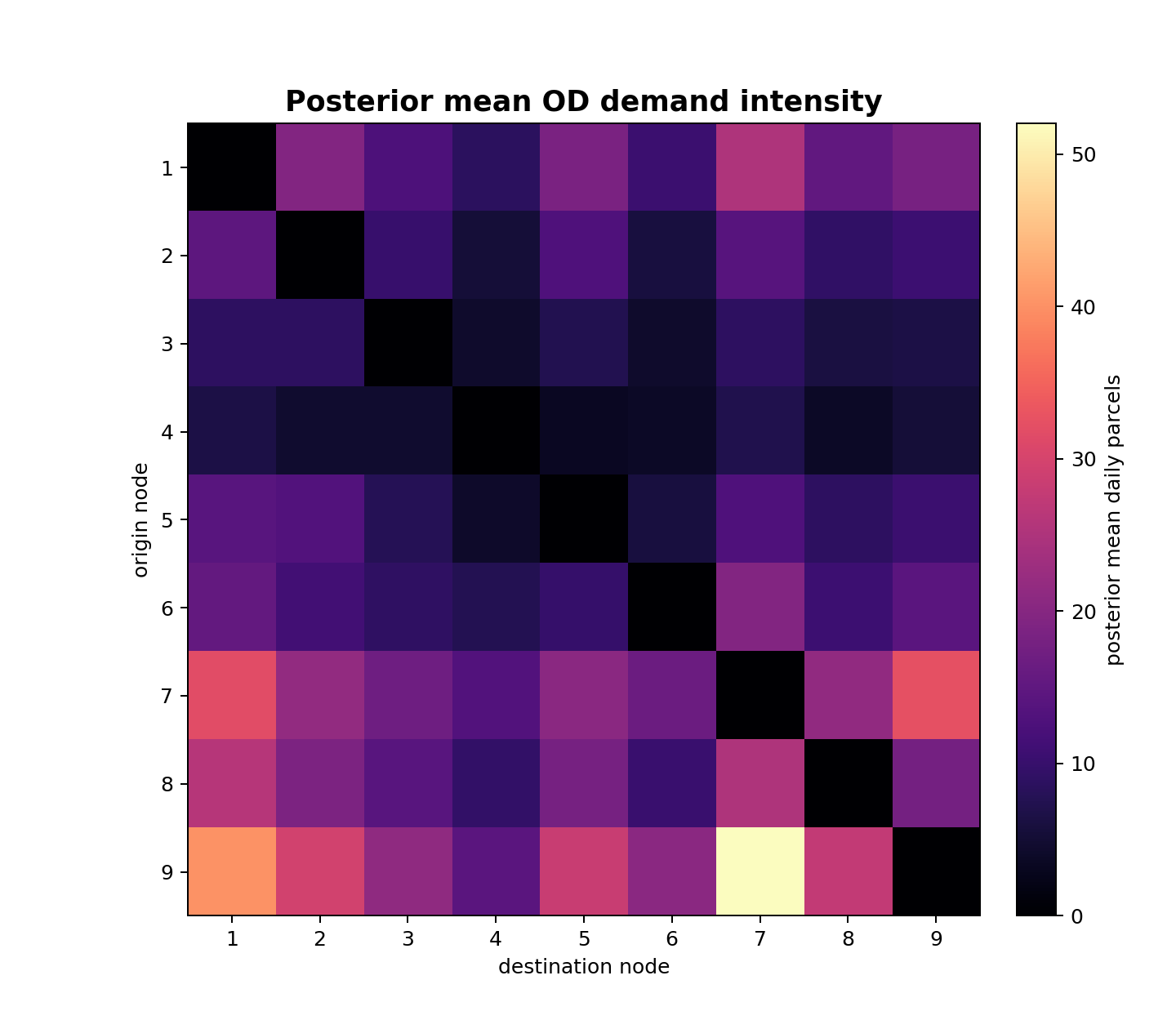}
\caption{Posterior mean directed OD demand matrix.  The heatmap shows that the demand field is anisotropic and imbalanced, which motivates a posterior-predictive multi-topology design rather than a fixed topology chosen a priori.}
\label{fig:demand_heatmap}
\end{figure}

\subsection{Candidate topology-design set}
\label{subsec:candidate_set}

Table~\ref{tab:candidate_counts} gives the number of candidate designs evaluated for each topology.  In total, 136 candidate designs are evaluated across the seven topology classes.  Hybrid topologies have larger candidate sets because they include both hub allocation decisions and direct-connection intensity decisions.

\begin{table}[H]
\centering
\caption{Number of candidate network designs evaluated by topology.}
\label{tab:candidate_counts}
\begin{tabular}{lr}
\toprule
Topology & Candidate designs \\
\midrule
DSAHS & 42 \\
DMAHS & 24 \\
DRAHS & 24 \\
SAHS & 21 \\
RAHS & 12 \\
MAHS & 12 \\
FC & 1 \\
\bottomrule
\end{tabular}
\end{table}

The posterior cost--risk tradeoff is shown in Figure~\ref{fig:tradeoff}.  The fully connected topology is not automatically optimal: although it can reduce some routing complexity, it carries a substantially higher expected operating cost and still suffers from posterior travel-time tail risk under disrupted corridors.  The best designs lie on a cost--risk frontier dominated by hub-spoke and hybrid structures.  The selected Bayesian design is a DSAHS design with hubs \(\{4,8,3\}\), sorting-capacity multiplier 1.85, and direct fraction 0.12.

\begin{figure}[H]
\centering
\includegraphics[width=0.86\textwidth]{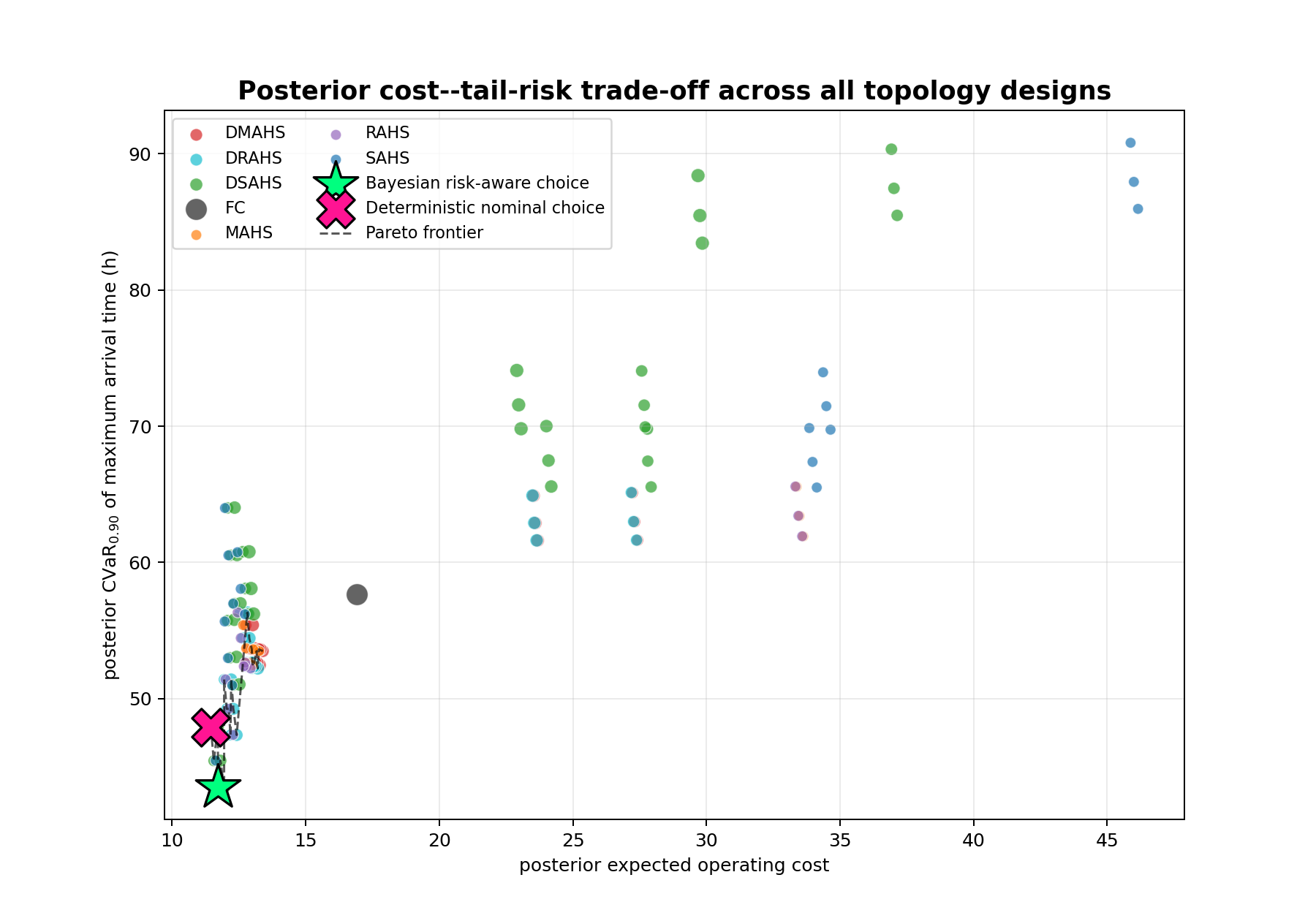}
\caption{Posterior expected cost versus posterior CVaR of maximum arrival time for all candidate topology designs.  The selected Bayesian design lies near the lower-left risk-efficient region, balancing operating cost with tail protection.}
\label{fig:tradeoff}
\end{figure}

\subsection{Posterior performance by topology}
\label{subsec:topology_performance}

Table~\ref{tab:best_by_topology} reports the best posterior design within each topology class.  The Bayesian selection is DSAHS with hubs \(4,8,3\), four direct links, capacity multiplier 1.85, and direct fraction 0.12.  It has the smallest posterior Bayes-risk score, equal to 0.004706.  Its posterior expected cost is 11.713541 million, posterior mean maximum arrival time is 32.512597 hours, and posterior CVaR of maximum arrival time is 43.465497 hours.  Both service reliability and hub-hold reliability are equal to 1.000000 in the posterior evaluation.

\begin{table}[H]
\centering
\caption{Best posterior-predictive design within each topology class.}
\label{tab:best_by_topology}
\begin{adjustbox}{max width=\textwidth}
\begin{tabular}{lrrrrrrrrr}
\toprule
Topology & Hubs & Direct & Cap. & Dir. frac. & Cost & Mean max time & CVaR max time & Service rel. & Hold rel. \\
\midrule
DSAHS & 4,8,3 & 4 & 1.85 & 0.12 & 11.713541 & 32.512597 & 43.465497 & 1.000000 & 1.000000 \\
SAHS  & 4,8,3 & 0 & 1.85 & 0.00 & 11.809265 & 32.515300 & 43.476751 & 1.000000 & 1.000000 \\
DRAHS & 4,8,3 & 4 & 1.85 & 0.12 & 12.183706 & 33.380480 & 47.313365 & 1.000000 & 1.000000 \\
RAHS  & 4,8,3 & 0 & 1.85 & 0.00 & 12.279430 & 33.375399 & 47.323192 & 1.000000 & 1.000000 \\
MAHS  & 8,3   & 0 & 1.85 & 0.00 & 12.909091 & 35.602477 & 52.521429 & 0.985714 & 1.000000 \\
DMAHS & 8,3   & 6 & 1.85 & 0.12 & 12.953518 & 35.601429 & 52.522208 & 0.985714 & 1.000000 \\
FC    & --    & 72 & 1.00 & 0.00 & 16.907638 & 40.229880 & 57.658418 & 0.971429 & 1.000000 \\
\bottomrule
\end{tabular}
\end{adjustbox}
\end{table}

Figure~\ref{fig:topology_reliability} compares the reliability of topology winners.  The selected DSAHS design achieves full posterior service and hub-hold reliability in the training posterior scenarios.  Figures~\ref{fig:arrival_boxplot} and \ref{fig:hub_delay_boxplot} further clarify the mechanism: DSAHS and SAHS control the posterior distribution of the maximum arrival time effectively, while fully connected routing remains exposed to corridor-level travel-time extremes; hub-delay risk is controlled by the capacity multiplier selected under posterior uncertainty.

\begin{figure}[H]
\centering
\includegraphics[width=0.82\textwidth]{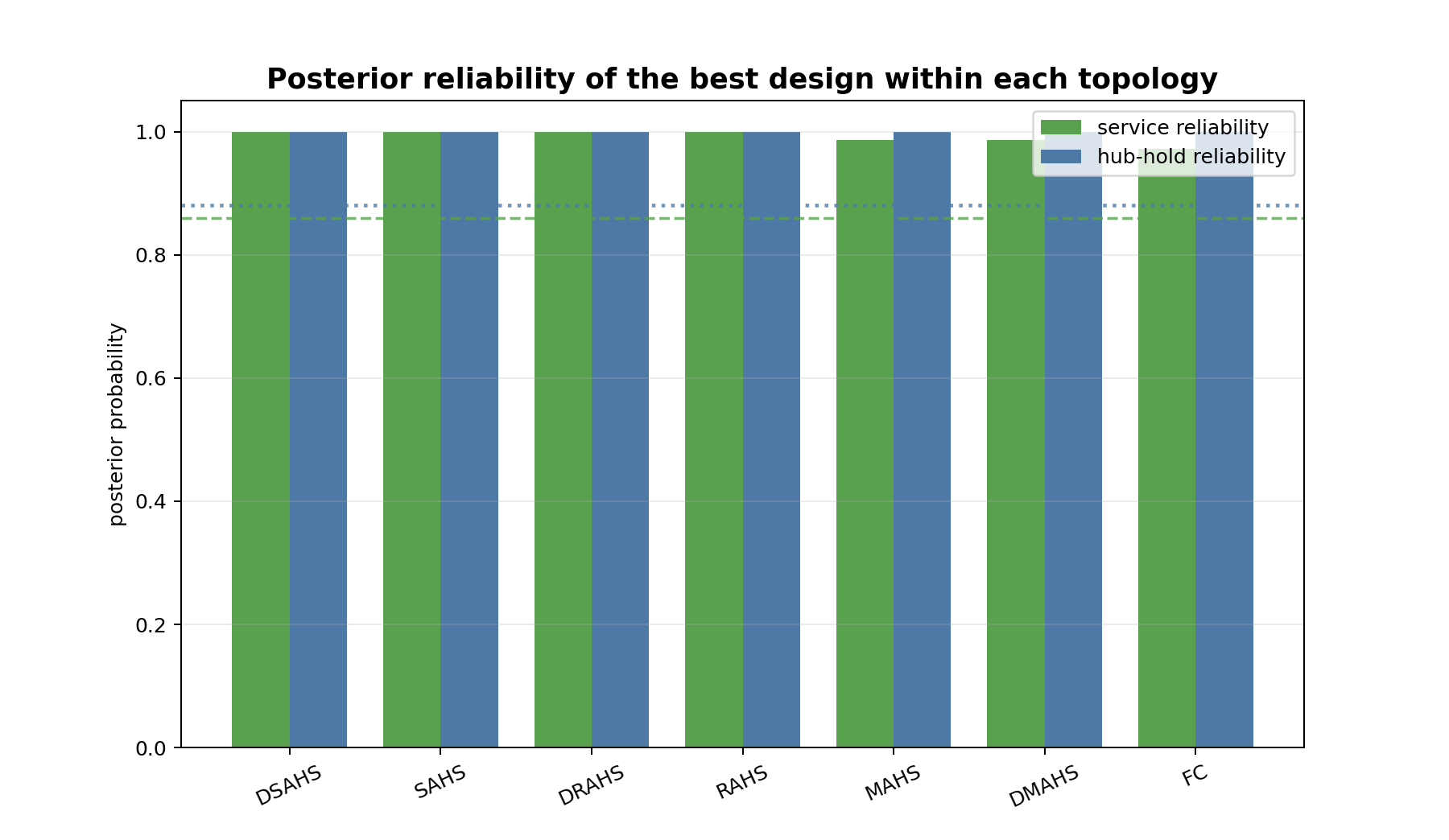}
\caption{Posterior service and hub-hold reliability for the best design within each topology.  Reliability is evaluated over posterior predictive scenarios rather than under a single nominal demand realization.}
\label{fig:topology_reliability}
\end{figure}

\begin{figure}[H]
\centering
\includegraphics[width=0.82\textwidth]{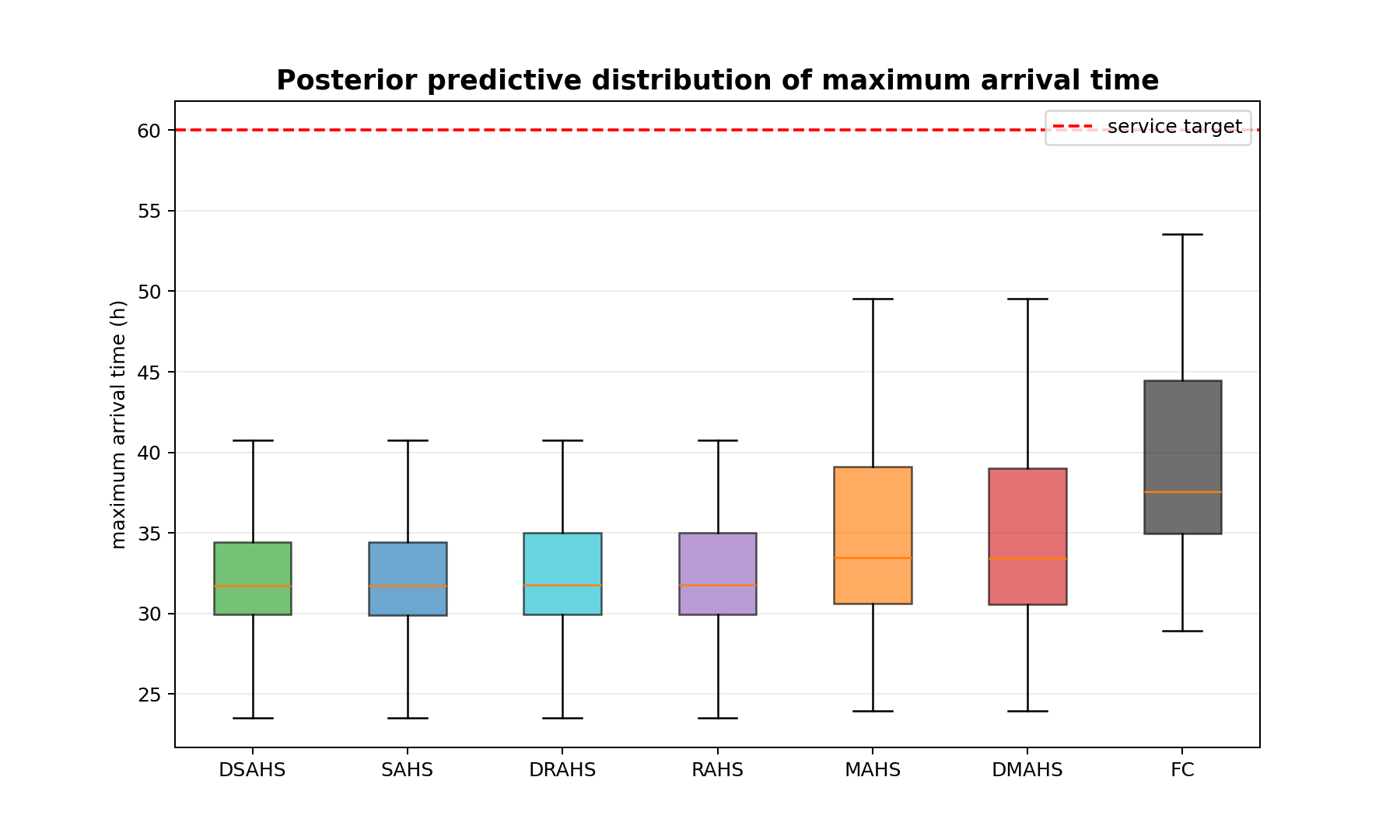}
\caption{Posterior distribution of maximum arrival time for topology winners.  The comparison highlights why tail-risk metrics are more informative than mean travel-time summaries alone.}
\label{fig:arrival_boxplot}
\end{figure}

\begin{figure}[H]
\centering
\includegraphics[width=0.82\textwidth]{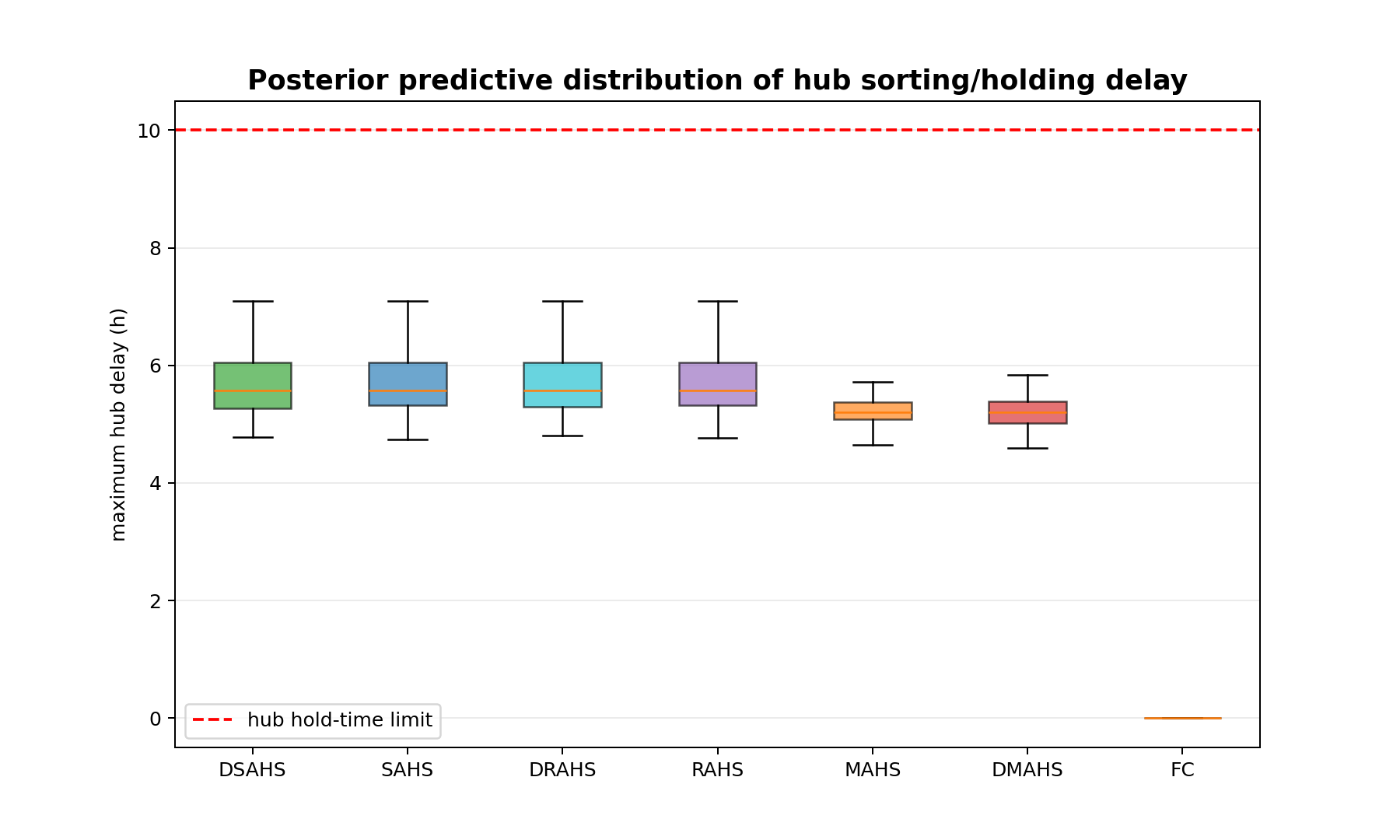}
\caption{Posterior distribution of maximum hub delay for topology winners.  The Bayesian design controls hub sorting delay by increasing capacity in a targeted way rather than simply minimizing nominal cost.}
\label{fig:hub_delay_boxplot}
\end{figure}

Table~\ref{tab:top20_designs} provides the leading posterior candidate designs.  The ranking shows a sensible pattern: increasing capacity from 1.05 to 1.40 and then 1.85 reduces posterior tail risk and improves reliability, while excessive or poorly placed hybridization increases cost without sufficient risk reduction.  The top three designs are DSAHS, SAHS, and DSAHS variants using all three hubs \(4,8,3\).

\begin{longtable}{llrrrrrrr}
\caption{Top posterior candidate designs ranked by posterior Bayes-risk score.}\label{tab:top20_designs}\\
\toprule
Topology & Hubs & Direct & Cap. & Cost & Mean max & CVaR max & Serv. rel. & Risk score \\
\midrule
\endfirsthead
\toprule
Topology & Hubs & Direct & Cap. & Cost & Mean max & CVaR max & Serv. rel. & Risk score \\
\midrule
\endhead
DSAHS & 4,8,3 & 4 & 1.85 & 11.713541 & 32.512597 & 43.465497 & 1.000000 & 0.004706 \\
SAHS  & 4,8,3 & 0 & 1.85 & 11.809265 & 32.515300 & 43.476751 & 1.000000 & 0.006584 \\
DSAHS & 4,8,3 & 9 & 1.85 & 11.944791 & 32.506473 & 43.472734 & 1.000000 & 0.007537 \\
DSAHS & 4,8,3 & 4 & 1.40 & 11.557285 & 34.556098 & 45.425638 & 1.000000 & 0.019374 \\
SAHS  & 4,8,3 & 0 & 1.40 & 11.629205 & 34.563944 & 45.441158 & 1.000000 & 0.021000 \\
DSAHS & 4,8,3 & 9 & 1.40 & 11.807136 & 34.546560 & 45.435638 & 1.000000 & 0.022454 \\
DRAHS & 4,8,3 & 4 & 1.85 & 12.183706 & 33.380480 & 47.313365 & 1.000000 & 0.042142 \\
RAHS  & 4,8,3 & 0 & 1.85 & 12.279430 & 33.375399 & 47.323192 & 1.000000 & 0.044009 \\
DRAHS & 4,8,3 & 9 & 1.85 & 12.414956 & 33.373006 & 47.318725 & 1.000000 & 0.044958 \\
SAHS  & 4,8,3 & 0 & 1.05 & 11.501332 & 36.875730 & 47.915284 & 0.985714 & 0.049033 \\
DRAHS & 4,8,3 & 4 & 1.40 & 12.042570 & 35.323447 & 49.238431 & 1.000000 & 0.056697 \\
RAHS  & 4,8,3 & 0 & 1.40 & 12.114490 & 35.316411 & 49.238086 & 1.000000 & 0.058190 \\
DRAHS & 4,8,3 & 9 & 1.40 & 12.292421 & 35.313502 & 49.238214 & 1.000000 & 0.059691 \\
SAHS  & 4,8   & 0 & 1.85 & 12.245658 & 39.012295 & 50.969732 & 1.000000 & 0.075475 \\
DSAHS & 4,8   & 6 & 1.85 & 12.290085 & 39.002648 & 50.983915 & 1.000000 & 0.075626 \\
DSAHS & 4,8,3 & 4 & 1.05 & 11.447926 & 36.869788 & 47.893999 & 0.985714 & 0.076510 \\
DSAHS & 4,8   & 12 & 1.85 & 12.514766 & 38.993884 & 51.035105 & 1.000000 & 0.078518 \\
SAHS  & 4,8   & 0 & 1.40 & 12.087397 & 40.976324 & 52.964599 & 0.985714 & 0.090412 \\
DSAHS & 4,8   & 6 & 1.40 & 12.163754 & 40.963066 & 52.988345 & 0.985714 & 0.091030 \\
RAHS  & 4,8   & 0 & 1.40 & 12.791830 & 38.569342 & 52.365938 & 0.985714 & 0.092869 \\
\bottomrule
\end{longtable}

The posterior probability of being scenario-best is shown in Table~\ref{tab:scenario_best} and Figure~\ref{fig:scenario_best}.  The selected DSAHS design is scenario-best in 91.4286\% of posterior predictive scenarios.  This is a strong statistical signal in favor of the proposed methodology: the selected design is not merely marginally better under one scalarized average; it is repeatedly best across the random future scenarios generated from the posterior.

\begin{table}[H]
\centering
\caption{Posterior probability that each topology winner is scenario-best.}
\label{tab:scenario_best}
\begin{adjustbox}{max width=\textwidth}
\begin{tabular}{llrr}
\toprule
Topology & Winner label & Posterior probability scenario-best & Mean scenario loss \\
\midrule
DSAHS & DSAHS\(|H=4-8-3|\)cap=1.85, direct=0.12 & 0.914286 & 0.047403 \\
DRAHS & DRAHS\(|H=4-8-3|\)cap=1.85, direct=0.12, \(R=2\) & 0.057143 & 0.090200 \\
DMAHS & DMAHS\(|H=8-3|\)cap=1.85, direct=0.12 & 0.014286 & 0.329390 \\
SAHS & SAHS\(|H=4-8-3|\)cap=1.85 & 0.014286 & 0.072690 \\
FC & FC\(|H=\)none & 0.000000 & 0.951294 \\
MAHS & MAHS\(|H=8-3|\)cap=1.85 & 0.000000 & 0.339490 \\
RAHS & RAHS\(|H=4-8-3|\)cap=1.85, \(R=2\) & 0.000000 & 0.114662 \\
\bottomrule
\end{tabular}
\end{adjustbox}
\end{table}

\begin{figure}[H]
\centering
\includegraphics[width=0.82\textwidth]{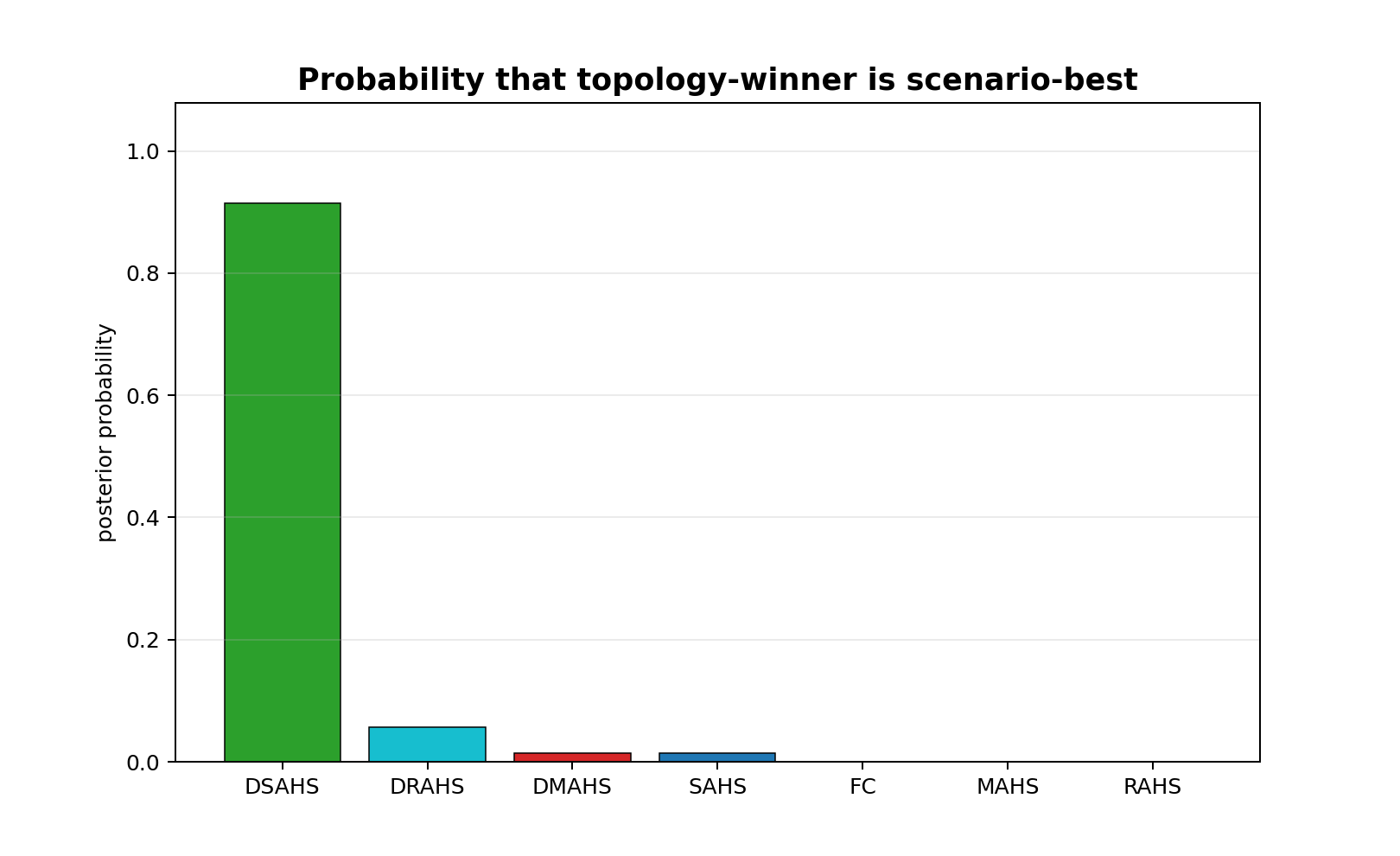}
\caption{Posterior probability that each topology winner is best in a posterior scenario.  The DSAHS winner dominates the posterior scenario-wise comparison.}
\label{fig:scenario_best}
\end{figure}

\subsection{Comparison with a deterministic nominal baseline}
\label{subsec:deterministic_baseline}

A crucial verification question is whether the Bayesian design gives a meaningful improvement over a deterministic nominal alternative.  To answer this, the experiment compares the posterior-predictive risk-aware design with a deterministic nominal mean-only baseline.  Both selected designs use the same DSAHS topology and the same hub set \(\{4,8,3\}\), but they differ in sorting-capacity multiplier.  The deterministic baseline chooses capacity multiplier 1.05, whereas the Bayesian design chooses 1.85.  Hence the Bayesian improvement is not due to changing the entire topology class, but due to the posterior-predictive recognition that the cheaper low-capacity design is fragile under future stress.

\begin{table}[H]
\centering
\caption{Selected Bayesian posterior-predictive design versus deterministic nominal mean-only baseline.}
\label{tab:bayes_vs_det}
\begin{adjustbox}{max width=\textwidth}
\begin{tabular}{llrrrrrrrr}
\toprule
Method & Topology & Hubs & Direct & Cap. & Cost & Mean max & CVaR max & Serv. rel. & Hold rel. \\
\midrule
Bayesian posterior-predictive risk-aware & DSAHS & 4,8,3 & 4 & 1.85 & 11.713541 & 32.512597 & 43.465497 & 1.000000 & 1.000000 \\
Deterministic nominal mean-only baseline & DSAHS & 4,8,3 & 4 & 1.05 & 11.447926 & 36.869788 & 47.893999 & 0.985714 & 0.842857 \\
\bottomrule
\end{tabular}
\end{adjustbox}
\end{table}

Table~\ref{tab:future_stress} reports an independent future stress test.  Under future stress scenarios, the Bayesian design has expected cost 12.538900 million, compared with 12.273285 million for the deterministic baseline.  However, the Bayesian design reduces CVaR of maximum arrival time from 84.723397 hours to 78.484113 hours and reduces the 95th percentile of maximum arrival time from 75.617107 hours to 69.957250 hours.  It also improves service reliability from 0.811111 to 0.900000 and hub-hold reliability from 0.622222 to 0.833333.  The mean maximum hub delay is nearly halved, falling from 11.710315 hours to 6.316476 hours.

\begin{table}[H]
\centering
\caption{Independent future stress-test comparison of the Bayesian and deterministic selected designs.}
\label{tab:future_stress}
\begin{adjustbox}{max width=\textwidth}
\begin{tabular}{lrrrrrr}
\toprule
Method & Cost & CVaR max time & 95\% max time & Service rel. & Hold rel. & Mean max hub delay \\
\midrule
Bayesian posterior-predictive risk-aware & 12.538900 & 78.484113 & 69.957250 & 0.900000 & 0.833333 & 6.316476 \\
Deterministic nominal mean-only baseline & 12.273285 & 84.723397 & 75.617107 & 0.811111 & 0.622222 & 11.710315 \\
\bottomrule
\end{tabular}
\end{adjustbox}
\end{table}

The improvement is summarized in Table~\ref{tab:gains}.  The Bayesian methodology reduces CVaR maximum-arrival risk by 7.3643\%, reduces the 95th percentile of maximum arrival time by 7.4849\%, improves service reliability by 8.8889 percentage points, and improves hub-hold reliability by 21.1111 percentage points.  These reliability gains are obtained at a modest expected-cost increase of 2.1642\%.

\begin{table}[H]
\centering
\caption{Measured gains of the Bayesian design relative to the deterministic nominal baseline.}
\label{tab:gains}
\begin{tabular}{lr}
\toprule
Verification metric & Bayesian versus deterministic \\
\midrule
CVaR maximum-arrival reduction (\%) & 7.364298 \\
95th percentile maximum-arrival reduction (\%) & 7.484889 \\
Service-reliability improvement (percentage points) & 8.888889 \\
Hub-hold reliability improvement (percentage points) & 21.111111 \\
Expected-cost increase for robustness (\%) & 2.164174 \\
\bottomrule
\end{tabular}
\end{table}

Figures~\ref{fig:future_arrival}, \ref{fig:future_hub_delay}, and \ref{fig:future_cost} visualize the future stress-test comparison.  The arrival-time and hub-delay distributions show a clear leftward shift for the Bayesian design in the upper tail.  The cost distribution moves only slightly upward, which is precisely the intended risk--cost tradeoff: the Bayesian design purchases operational robustness at a small and interpretable cost premium.

\begin{figure}[H]
\centering
\includegraphics[width=0.82\textwidth]{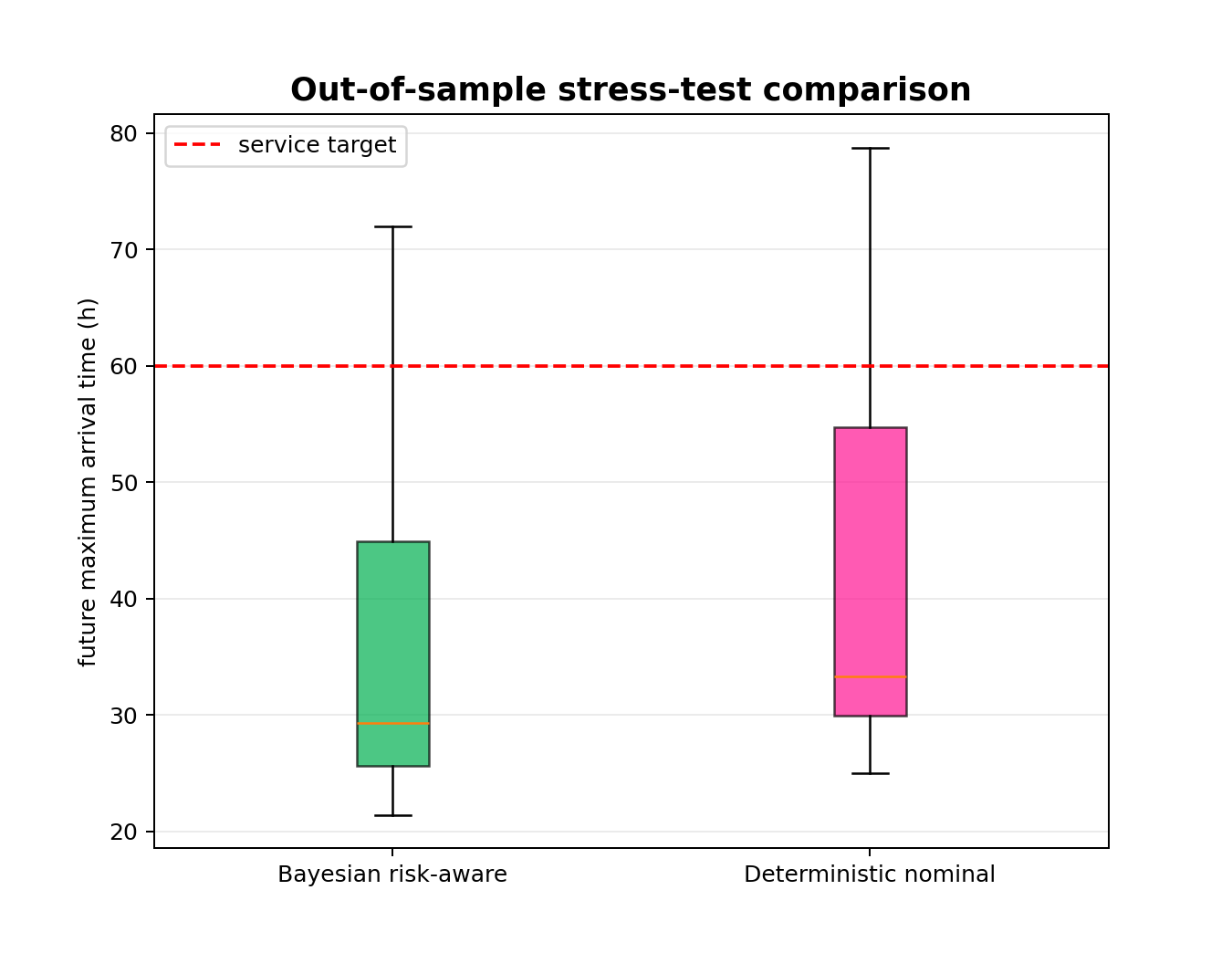}
\caption{Future stress-test distribution of maximum arrival time for the Bayesian design and deterministic baseline.  The Bayesian design has a visibly smaller upper tail.}
\label{fig:future_arrival}
\end{figure}

\begin{figure}[H]
\centering
\includegraphics[width=0.82\textwidth]{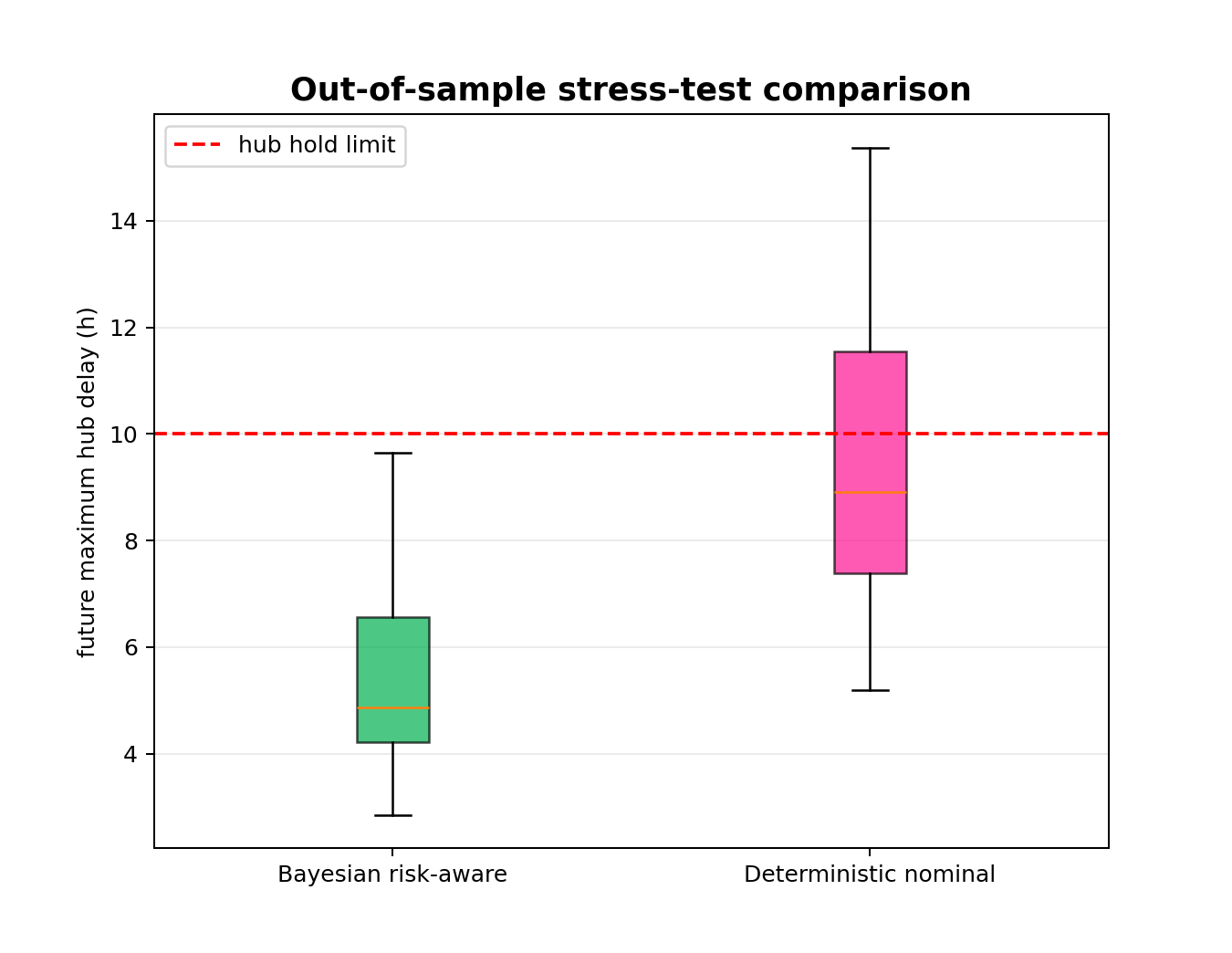}
\caption{Future stress-test distribution of maximum hub delay.  The deterministic baseline suffers from substantially larger hub-delay extremes due to underinvestment in sorting capacity.}
\label{fig:future_hub_delay}
\end{figure}

\begin{figure}[H]
\centering
\includegraphics[width=0.82\textwidth]{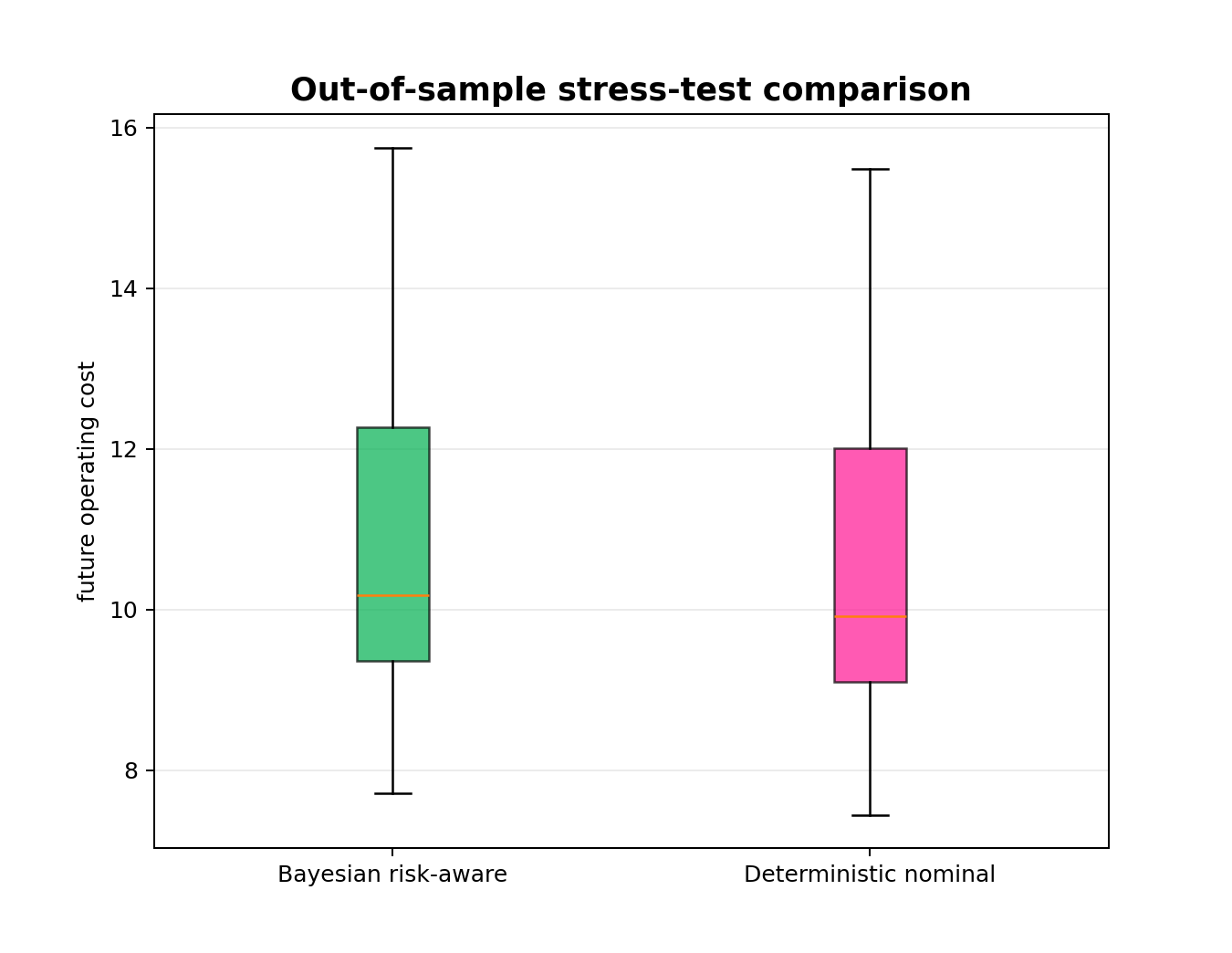}
\caption{Future stress-test distribution of operating cost.  The Bayesian design incurs a modest cost increase, but the increase is small relative to the achieved reliability and tail-risk gains.}
\label{fig:future_cost}
\end{figure}

\subsection{Preference-weight sensitivity}
\label{subsec:preference_sensitivity}

Table~\ref{tab:preference_sensitivity} reports representative preference-weight sensitivity results.  Across the explored grid of cost, CVaR-time, and emission weights, the selected design remains DSAHS with hubs \(\{4,8,3\}\), capacity multiplier 1.85, and direct fraction 0.12.  This stability is desirable: it indicates that the selected design is not an artifact of a single finely tuned scalarization.  Figure~\ref{fig:preference_heatmap} gives the corresponding heatmap representation of the sensitivity experiment.

\begin{longtable}{rrrrlrrrr}
\caption{Preference-weight sensitivity results.}\label{tab:preference_sensitivity}\\
\toprule
\(\omega_c\) & \(\omega_t\) & \(\omega_e\) & Topology & Hubs & Score & Cost & CVaR time & Service rel. \\
\midrule
\endfirsthead
\toprule
\(\omega_c\) & \(\omega_t\) & \(\omega_e\) & Topology & Hubs & Score & Cost & CVaR time & Service rel. \\
\midrule
\endhead
0.15 & 0.15 & 0.70 & DSAHS & 4,8,3 & 0.014195 & 11.713541 & 43.465497 & 1.000000 \\
0.15 & 0.25 & 0.60 & DSAHS & 4,8,3 & 0.012331 & 11.713541 & 43.465497 & 1.000000 \\
0.15 & 0.35 & 0.50 & DSAHS & 4,8,3 & 0.010467 & 11.713541 & 43.465497 & 1.000000 \\
0.15 & 0.45 & 0.40 & DSAHS & 4,8,3 & 0.008604 & 11.713541 & 43.465497 & 1.000000 \\
0.15 & 0.55 & 0.30 & DSAHS & 4,8,3 & 0.006740 & 11.713541 & 43.465497 & 1.000000 \\
0.15 & 0.65 & 0.20 & DSAHS & 4,8,3 & 0.004876 & 11.713541 & 43.465497 & 1.000000 \\
0.15 & 0.75 & 0.10 & DSAHS & 4,8,3 & 0.003012 & 11.713541 & 43.465497 & 1.000000 \\
0.25 & 0.15 & 0.60 & DSAHS & 4,8,3 & 0.013097 & 11.713541 & 43.465497 & 1.000000 \\
0.25 & 0.25 & 0.50 & DSAHS & 4,8,3 & 0.011233 & 11.713541 & 43.465497 & 1.000000 \\
0.25 & 0.35 & 0.40 & DSAHS & 4,8,3 & 0.009369 & 11.713541 & 43.465497 & 1.000000 \\
0.25 & 0.45 & 0.30 & DSAHS & 4,8,3 & 0.007505 & 11.713541 & 43.465497 & 1.000000 \\
0.25 & 0.55 & 0.20 & DSAHS & 4,8,3 & 0.005641 & 11.713541 & 43.465497 & 1.000000 \\
0.25 & 0.65 & 0.10 & DSAHS & 4,8,3 & 0.003777 & 11.713541 & 43.465497 & 1.000000 \\
0.35 & 0.15 & 0.50 & DSAHS & 4,8,3 & 0.011998 & 11.713541 & 43.465497 & 1.000000 \\
0.35 & 0.25 & 0.40 & DSAHS & 4,8,3 & 0.010134 & 11.713541 & 43.465497 & 1.000000 \\
0.35 & 0.35 & 0.30 & DSAHS & 4,8,3 & 0.008270 & 11.713541 & 43.465497 & 1.000000 \\
0.35 & 0.45 & 0.20 & DSAHS & 4,8,3 & 0.006407 & 11.713541 & 43.465497 & 1.000000 \\
0.35 & 0.55 & 0.10 & DSAHS & 4,8,3 & 0.004543 & 11.713541 & 43.465497 & 1.000000 \\
0.45 & 0.15 & 0.40 & DSAHS & 4,8,3 & 0.010900 & 11.713541 & 43.465497 & 1.000000 \\
0.45 & 0.25 & 0.30 & DSAHS & 4,8,3 & 0.009036 & 11.713541 & 43.465497 & 1.000000 \\
\bottomrule
\end{longtable}

\begin{figure}[H]
\centering
\includegraphics[width=0.84\textwidth]{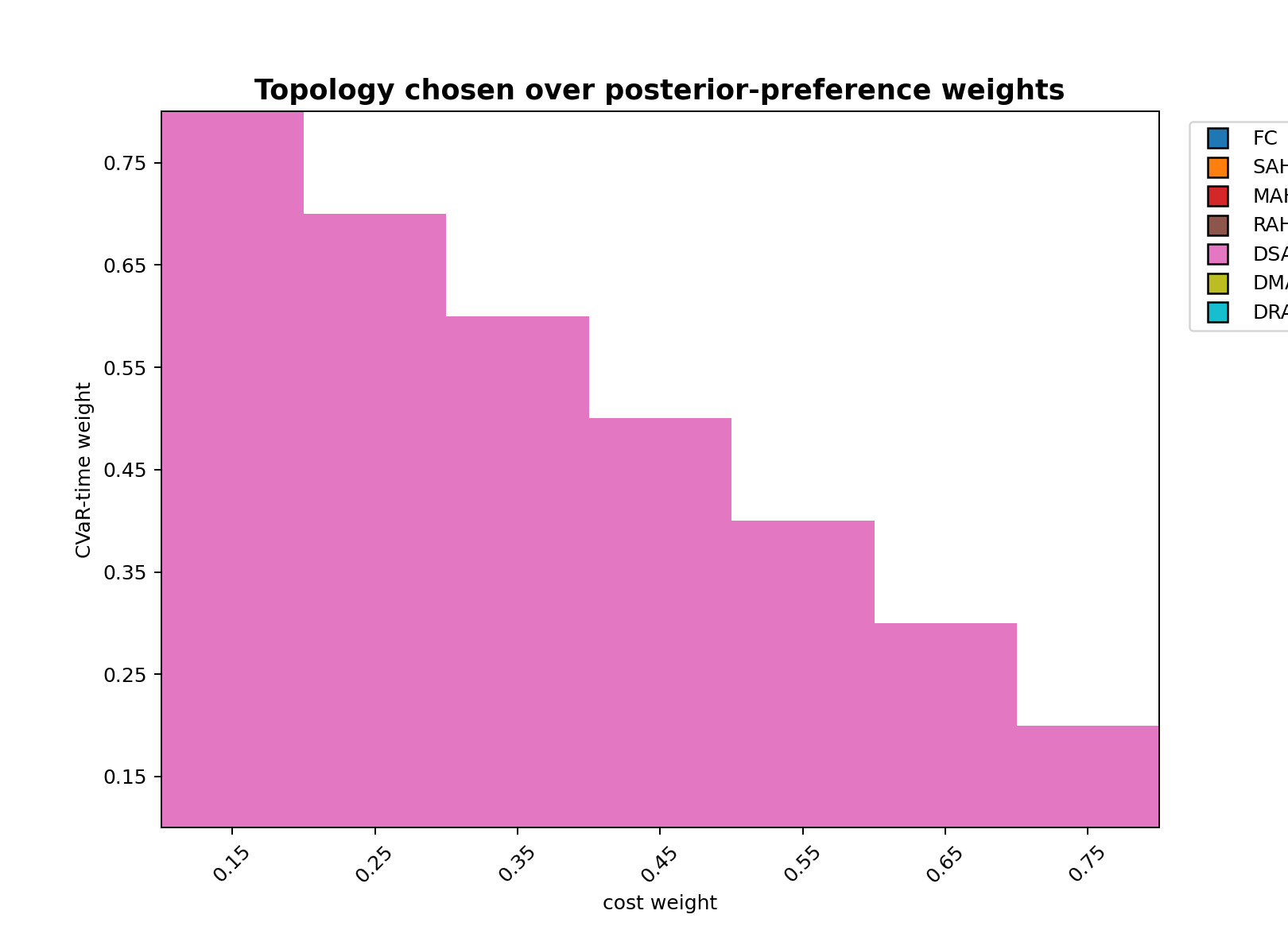}
\caption{Preference-weight sensitivity heatmap.  The selected topology remains stable across a broad range of cost, CVaR-time, and emission preferences.}
\label{fig:preference_heatmap}
\end{figure}

Finally, Figure~\ref{fig:hub_reliability_posterior} shows the posterior distribution of candidate hub reliability.  This is important for interpretation: the method does not assume that the sorting capacity decision is deterministic.  Instead, hub reliability is learned and propagated to downstream service-risk calculations.

\begin{figure}[H]
\centering
\includegraphics[width=0.82\textwidth]{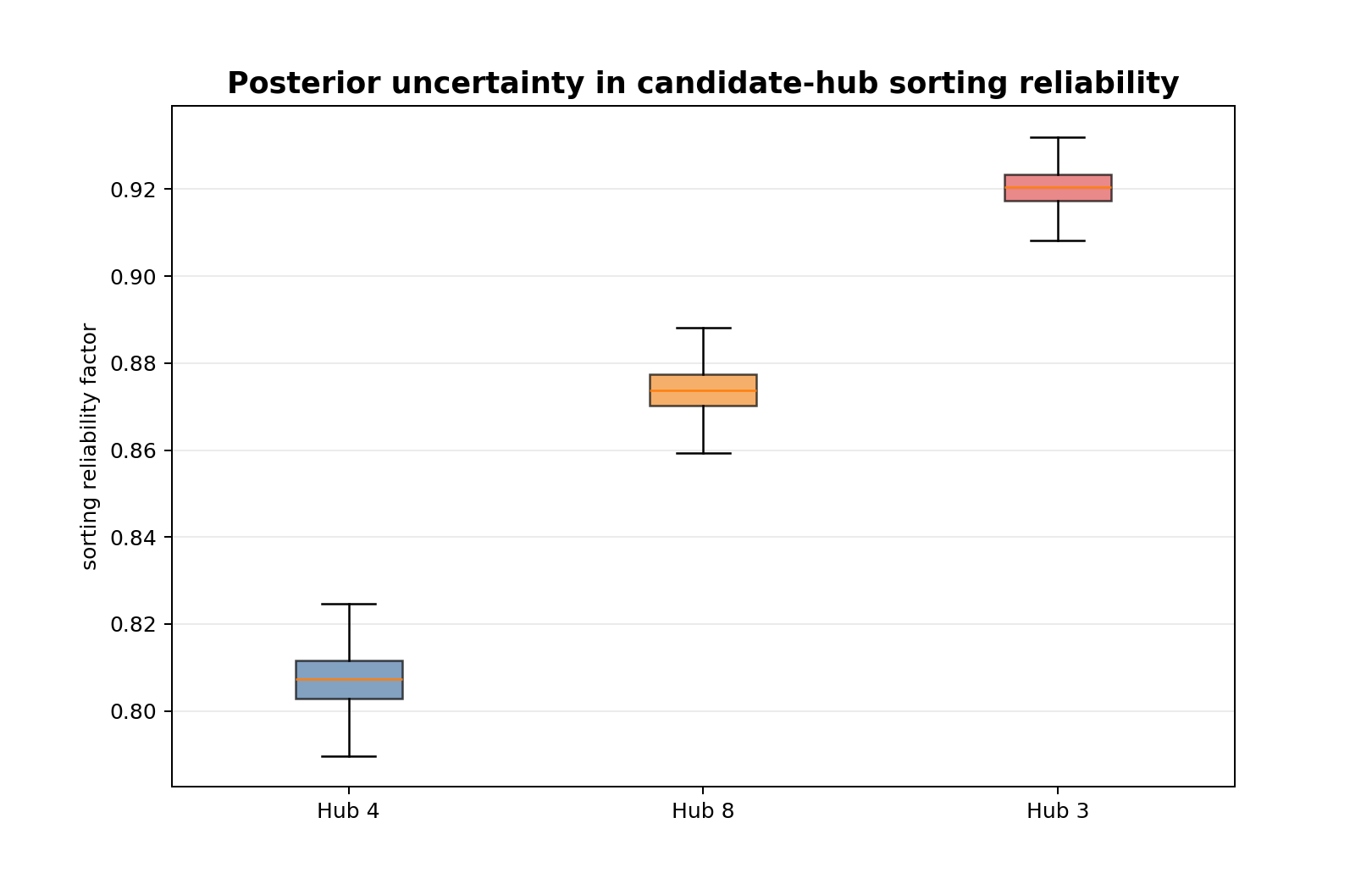}
\caption{Posterior distribution of hub reliability for the candidate hubs.  Reliability uncertainty is propagated into the posterior-predictive network-design evaluation.}
\label{fig:hub_reliability_posterior}
\end{figure}

\subsection{Key findings from the simulation study}
\label{subsec:key_findings_simulation}

The simulation supports the proposed methodology in five statistically meaningful ways.

First, posterior-predictive optimization changes the nature of the decision.  The deterministic nominal baseline and the Bayesian design select the same DSAHS topology and hub set \(\{4,8,3\}\), but they select different sorting-capacity levels.  The deterministic method selects the cheaper capacity multiplier 1.05, whereas the Bayesian method selects 1.85.  This distinction is important: the Bayesian contribution is not an arbitrary topology switch, but a principled correction for posterior tail risk and hub-delay uncertainty.

Second, the Bayesian design yields substantial reliability improvement for a modest cost premium.  In the future stress test, the Bayesian design improves service reliability by 8.8889 percentage points and hub-hold reliability by 21.1111 percentage points, while increasing expected cost by only 2.1642\%.  This is exactly the kind of cost--risk tradeoff that a courier network designer would need under volatile demand and travel conditions.

Third, the Bayesian design reduces tail risk rather than only improving average performance.  The CVaR of maximum arrival time is reduced by 7.3643\%, and the 95th percentile of maximum arrival time is reduced by 7.4849\%.  Therefore the proposed method is not simply minimizing a mean objective; it is protecting the upper tail of the delivery-time distribution, which is operationally more relevant for service-level guarantees.

Fourth, scenario-wise posterior comparison strongly favors the selected topology.  The selected DSAHS winner is scenario-best in 91.4286\% of posterior predictive scenarios.  This gives a direct probabilistic interpretation of topology choice and is more informative than a deterministic ranking based on one nominal data realization.

Fifth, the simulation highlights the value of modeling hub sorting reliability.  The largest practical improvement appears in hub-hold reliability and maximum hub-delay reduction.  Under future stress, the mean maximum hub delay decreases from 11.710315 hours for the deterministic baseline to 6.316476 hours for the Bayesian design.  Hence the Bayesian capacity decision directly addresses the operational bottleneck that deterministic mean-only optimization tends to understate.

Overall, the simulation verifies that the proposed Bayesian multi-topology design methodology is useful not because it always chooses a radically different topology, but because it chooses a network design that is posterior-risk efficient: it spends slightly more where uncertainty matters, reduces extreme arrival-time behavior, and improves reliability under future stress.  This is the central statistical advantage over a deterministic multi-structure design framework.


\section{Case Study and Posterior Scenario Experiments on the CAB Hub-Location Benchmark}
\label{sec:cab-case-study}

\subsection{Real benchmark data and experimental objective}
\label{subsec:cab-data}

We now examine the proposed Bayesian multi-topology express transportation network design methodology on a real benchmark network.  The data source is the Civil Aeronautics Board (CAB) hub-location benchmark, originally associated with the classical interacting hub-location formulation of O'Kelly~\cite{okelly1987quadratic}.  The CAB instance is a standard benchmark in the hub-location literature and is mirrored in public hub-location data repositories; OR-Library and its related documentation identify CAB-type hub-location instances as widely used test beds for hub-and-spoke location models~\cite{beasley1990orlib,ernst1996efficient}.  The specific public file used in the present computational experiment is \texttt{CAB25.txt}, which contains a 25-node origin--destination flow matrix and a 25-node distance matrix.  Since the original deterministic bi-objective express transportation network design paper studies seven topological structures---FC, SAHS, MAHS, RAHS, DSAHS, DMAHS and DRAHS---under operation-cost and maximum-arrival-time criteria~\cite{zhong2023etndp}, the CAB benchmark is a natural real-data platform on which to evaluate whether the proposed Bayesian posterior-predictive extension produces practically different and statistically defensible design choices.

A methodological caveat is important.  CAB25 is a real static OD-flow/distance benchmark, not a proprietary daily courier panel.  Therefore, in order to evaluate Bayesian posterior learning and future posterior scenario performance, we use the real CAB flow and distance matrices as the structural baseline and construct a pseudo-historical operational panel around them.  This is a standard simulation-on-real-benchmark design: the network geometry and relative OD intensity pattern are real, while temporal uncertainty, demand overdispersion, travel-time noise, hub-reliability variation and stress-regime perturbations are generated in a controlled and reproducible manner.  The goal is not to claim that the CAB benchmark is an express-company daily data set; rather, the goal is to test whether the Bayesian decision methodology can exploit posterior uncertainty on a real hub-location network structure.

Table~\ref{tab:cab-data-summary} summarizes the real-data instance used in the experiment.  We selected a 12-node subnetwork from CAB25 to keep all seven topology classes computationally comparable in a reproducible Python implementation.  Candidate hubs were selected from high-throughput/central nodes after relabeling the selected subnetwork.  The resulting case study contains 132 directed OD pairs, a mean scaled daily OD demand of 24 parcels, and a heterogeneous distance structure with mean pairwise distance 580.201 and 95th percentile distance 1317.898.  The heterogeneity in Fig.~\ref{fig:cab-mds-network} and Fig.~\ref{fig:cab-demand-heatmap} is important: it creates a meaningful trade-off between direct links, hub consolidation, hub sorting burden, and posterior delivery risk.

\begin{table}[!htbp]
\centering
\caption{Summary of the real CAB benchmark subnetwork used for posterior scenario experiments.}
\label{tab:cab-data-summary}
\begin{adjustbox}{max width=\textwidth}
\begin{tabular}{ll}
\toprule
Quantity & Value \\
\midrule
Data source & CAB25 hub-location benchmark \\
Public mirror & \url{https://raw.githubusercontent.com/mcroboredo/Hub-Location-Instances/main/CAB25.txt} \\
Selected CAB node labels & 3, 4, 6, 7, 9, 12, 14, 17, 18, 21, 22, 25 \\
Nodes in case study & 12 \\
Candidate hubs (case-study indices) & 12, 1, 3, 5 \\
Directed OD pairs & 132 \\
Mean scaled daily OD demand & 24.000 \\
Median scaled daily OD demand & 13.732 \\
Mean pairwise distance & 580.201 \\
95th percentile pairwise distance & 1317.898 \\
Historical pseudo-panel days & 120 \\
\bottomrule
\end{tabular}
\end{adjustbox}
\end{table}

\begin{figure}[!htbp]
\centering
\includegraphics[width=0.82\textwidth]{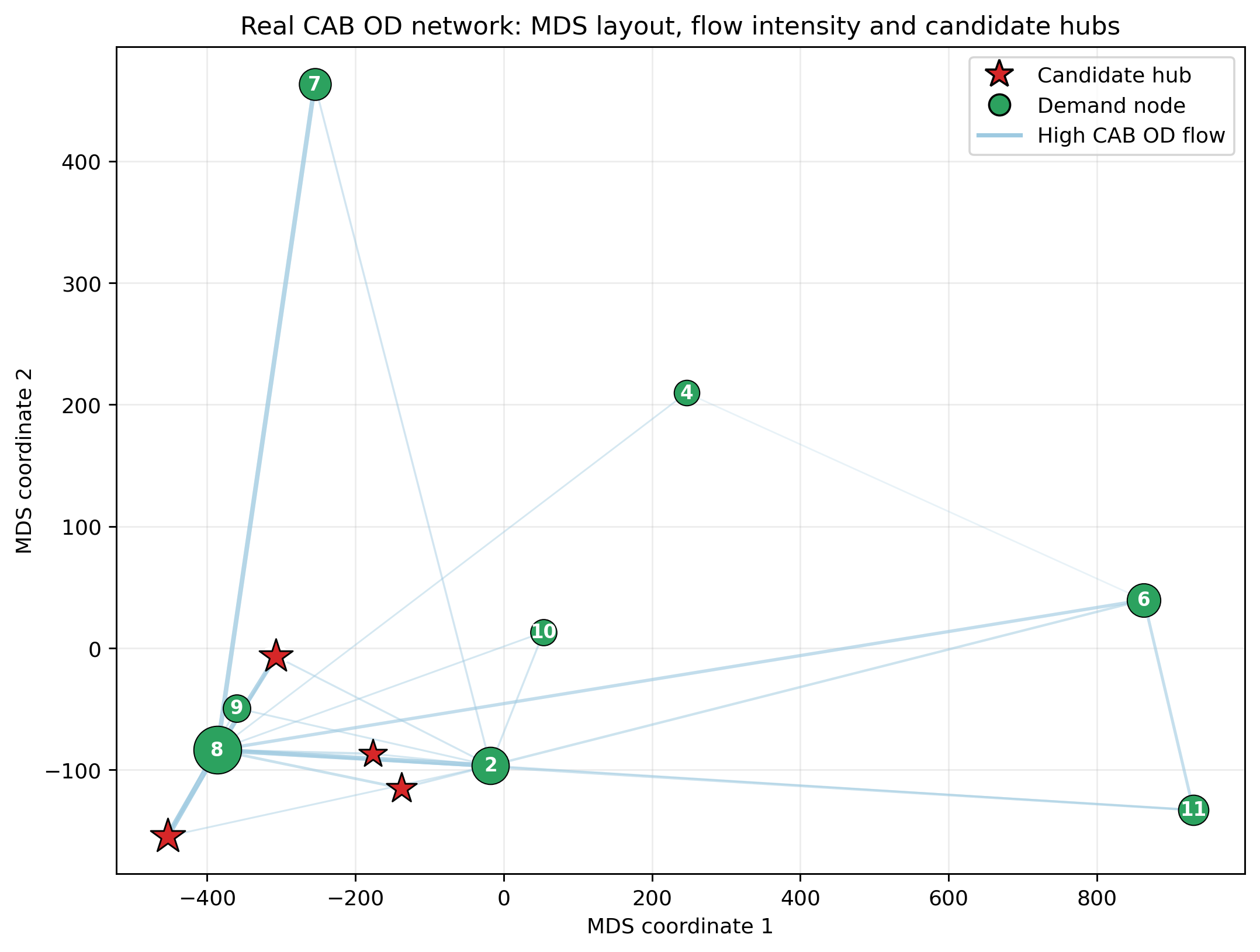}
\caption{Metric multidimensional-scaling visualization of the selected 12-node CAB subnetwork.  Node positions are derived from the real CAB distance matrix; candidate hubs are highlighted.  The figure shows that the benchmark subnetwork is geographically heterogeneous, so hub selection and direct-link admission are nontrivial.}
\label{fig:cab-mds-network}
\end{figure}

\begin{figure}[!htbp]
\centering
\includegraphics[width=0.82\textwidth]{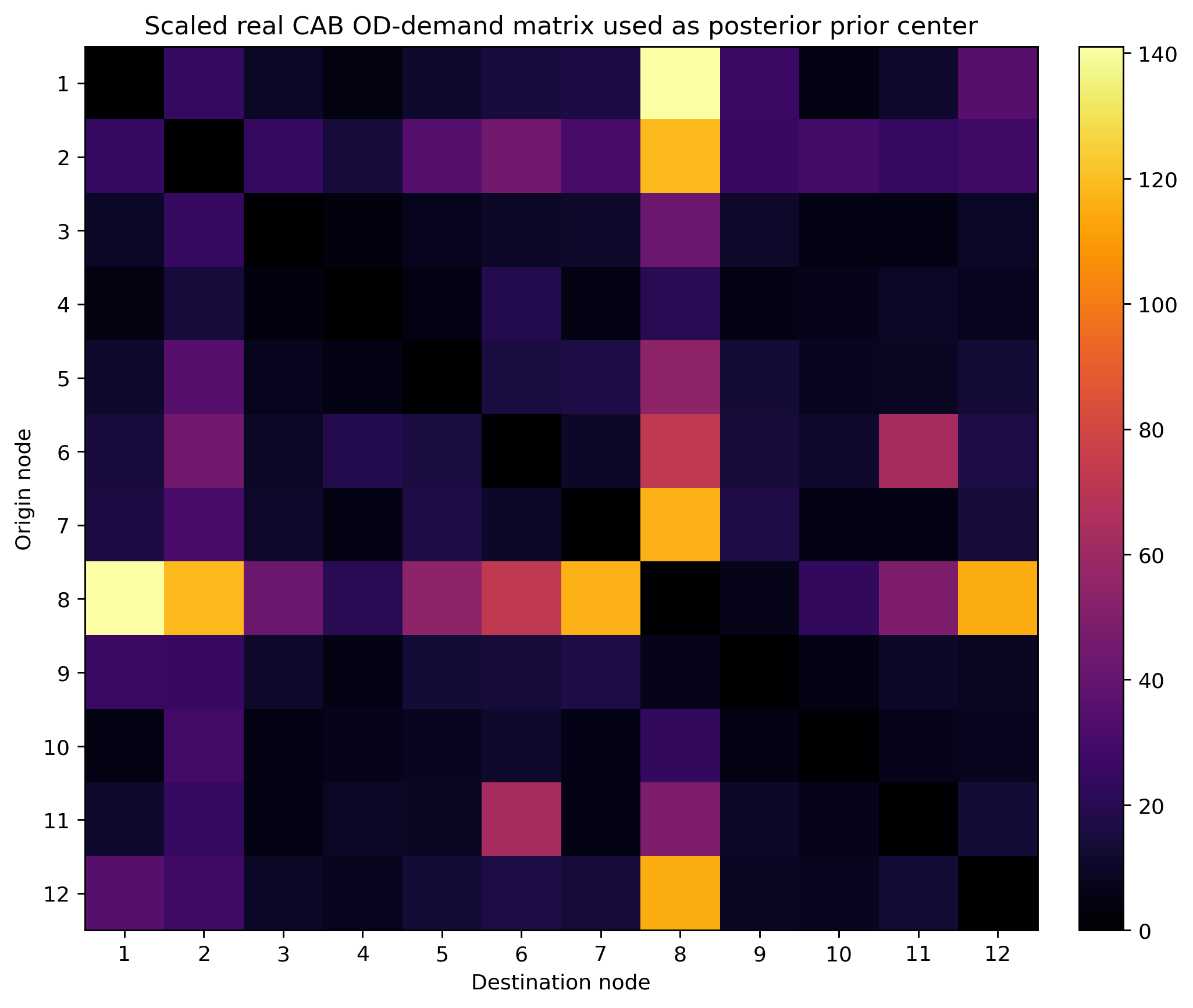}
\caption{Scaled directed OD-demand heatmap for the selected CAB subnetwork.  The nonuniform block pattern shows that the case study is not a homogeneous toy example: some directed corridors carry substantially larger flow than others, which directly affects hub-load uncertainty and posterior hold-time risk.}
\label{fig:cab-demand-heatmap}
\end{figure}

\subsection{Bayesian posterior construction and design space}
\label{subsec:cab-posterior-model}

Let $w_{ij,t}$ denote the daily OD demand from origin $i$ to destination $j$ on day $t$, let $\tau_{ij,t}$ denote travel time on arc $(i,j)$, and let $R_{k,t}$ denote the reliability/productivity factor at candidate hub $k$.  The posterior experiment uses conjugate or semi-conjugate manual Bayesian updates, deliberately avoiding black-box Bayesian software so that every step remains reproducible.

For OD demand, we use a Gamma--Poisson model
\begin{equation}
 w_{ij,t}\mid \lambda_{ij}\sim \operatorname{Poisson}(\lambda_{ij}),
 \qquad
 \lambda_{ij}\sim \operatorname{Gamma}(a_{ij,0},b_{ij,0}),
\end{equation}
where the prior mean is scaled from the real CAB OD flow.  The posterior is
\begin{equation}
 \lambda_{ij}\mid \mathcal D
 \sim
 \operatorname{Gamma}\left(a_{ij,0}+\sum_{t=1}^{T}w_{ij,t},\ b_{ij,0}+T\right).
\end{equation}
For travel time and cost multipliers we use lognormal observation models with Normal--Inverse-Gamma updating on the log scale.  For hub productivity/reliability we use a Beta--Binomial update, so that posterior predictive scenarios contain not only uncertain OD demand but also uncertain hub sorting performance.  Table~\ref{tab:cab-posterior-summary} reports the main posterior summaries.  The posterior mean daily OD intensity is 25.303, slightly above the scaled empirical mean because the pseudo-panel includes surge and disruption regimes; the mean candidate-hub reliability is 0.8285.

\begin{table}[!htbp]
\centering
\caption{Manual Bayesian posterior updating summary for the CAB posterior scenario experiment.}
\label{tab:cab-posterior-summary}
\begin{adjustbox}{max width=\textwidth}
\begin{tabular}{llll}
\toprule
Component & Bayesian model & Posterior quantity & Value \\
\midrule
OD demand & Gamma--Poisson & Mean daily OD intensity & 25.3028 \\
Travel time & Lognormal NIG & Mean log travel time & 2.0441 \\
Hub reliability & Beta--Binomial & Mean candidate-hub reliability & 0.8285 \\
Cost multiplier & Lognormal NIG & Posterior mean log multiplier & $-0.0018$ \\
\bottomrule
\end{tabular}
\end{adjustbox}
\end{table}

For each posterior scenario $b=1,\ldots,B$, with $B=120$, we generate posterior predictive draws of OD demand, travel time, hub reliability and cost multipliers.  For each topology $g\in\{\mathrm{FC},\mathrm{SAHS},\mathrm{MAHS},\mathrm{RAHS},\mathrm{DSAHS},\mathrm{DMAHS},\mathrm{DRAHS}\}$ and design vector $x_g$, the following quantities are evaluated:
\begin{align}
 \widehat C_g(x_g) &= B^{-1}\sum_{b=1}^{B} C_g(x_g;\Theta^{(b)}),\\
 \widehat A_g(x_g) &= B^{-1}\sum_{b=1}^{B} A_g(x_g;\Theta^{(b)}),\\
 \widehat{\mathrm{CVaR}}_{0.90,g}(x_g) &=
 \min_{\zeta}\left\{\zeta + \frac{1}{0.10B}\sum_{b=1}^{B}
 \left(A_g(x_g;\Theta^{(b)})-\zeta\right)_+\right\},\\
 \widehat R^{\mathrm{serv}}_g(x_g) &= B^{-1}\sum_{b=1}^{B}\mathbbm 1\{A_g(x_g;\Theta^{(b)})\leq T^\star\},\\
 \widehat R^{\mathrm{hold}}_g(x_g) &= B^{-1}\sum_{b=1}^{B}\mathbbm 1\{H_g(x_g;\Theta^{(b)})\leq d_t\}.
\end{align}
Here $T^\star=34$ hours is a deliberately stringent service target, $d_t=8$ hours is the hub hold-time threshold, and $H_g$ is the maximum posterior hub sorting/holding time.  The posterior Bayes-risk score used for ranking combines normalized expected cost, posterior CVaR of maximum arrival time, expected emission cost, service-failure penalty and hold-time-failure penalty.  The candidate design space contains 841 designs across the seven topology classes; the class counts are shown in Table~\ref{tab:cab-design-counts}.  The direct-connected hybrid structures have the largest number of admissible design variants because they combine hub allocation decisions with direct-link thresholding.

\begin{table}[!htbp]
\centering
\caption{Candidate design counts evaluated across the seven topology classes.}
\label{tab:cab-design-counts}
\begin{tabular}{lr}
\toprule
Topology & Candidate designs \\
\midrule
DRAHS & 224 \\
DMAHS & 224 \\
DSAHS & 224 \\
SAHS & 56 \\
MAHS & 56 \\
RAHS & 56 \\
FC & 1 \\
\bottomrule
\end{tabular}
\end{table}

\subsection{Posterior trade-off and topology comparison}
\label{subsec:cab-posterior-results}

Table~\ref{tab:cab-top-posterior-designs} reports the leading posterior designs by Bayes-risk score.  The top design is a direct-connected single-allocation hub-and-spoke design, DSAHS, with one selected hub, capacity multiplier 0.90, and 76 direct links.  Its expected posterior cost is 1.268 million, posterior mean maximum arrival time is 39.316 hours, and posterior $90\%$ CVaR of maximum arrival time is 55.417 hours.  DMAHS and DRAHS variants with the same hub/direct-link geometry have identical numerical performance in this experiment because the selected design uses one active hub, so the allocation distinction is inactive.  This tie is not a weakness; rather, it confirms that the computational evaluation respects the topology definitions and only differentiates topologies when their feasible operational structures differ.

\begin{table}[!htbp]
\centering
\caption{Leading posterior designs ranked by posterior Bayes-risk score.  Costs are in millions; time quantities are in hours.}
\label{tab:cab-top-posterior-designs}
\begin{adjustbox}{max width=\textwidth}
\begin{tabular}{lllrrrrrrr}
\toprule
Topology & Hubs & Cap. & Direct & Cost & Mean max & P95 max & CVaR max & Serv. rel. & Hold rel. \\
\midrule
DSAHS & 3 & 0.90 & 76 & 1.268 & 39.316 & 53.695 & 55.417 & 0.208 & 0.992 \\
DMAHS & 3 & 0.90 & 76 & 1.268 & 39.316 & 53.695 & 55.417 & 0.208 & 0.992 \\
DRAHS & 3 & 0.90 & 76 & 1.268 & 39.316 & 53.695 & 55.417 & 0.208 & 0.992 \\
DSAHS & 3 & 1.65 & 56 & 1.331 & 37.396 & 48.583 & 49.397 & 0.225 & 0.850 \\
DMAHS & 3 & 1.65 & 56 & 1.331 & 37.396 & 48.583 & 49.397 & 0.225 & 0.850 \\
DRAHS & 3 & 1.65 & 56 & 1.331 & 37.396 & 48.583 & 49.397 & 0.225 & 0.850 \\
DRAHS & 3 & 1.10 & 68 & 1.273 & 39.562 & 54.155 & 55.524 & 0.192 & 0.967 \\
DSAHS & 3 & 1.10 & 68 & 1.273 & 39.562 & 54.155 & 55.524 & 0.192 & 0.967 \\
DMAHS & 3 & 1.10 & 68 & 1.273 & 39.562 & 54.155 & 55.524 & 0.192 & 0.967 \\
DSAHS & 3 & 1.10 & 76 & 1.315 & 38.906 & 53.529 & 55.264 & 0.242 & 1.000 \\
DMAHS & 3 & 1.10 & 76 & 1.315 & 38.906 & 53.529 & 55.264 & 0.242 & 1.000 \\
DRAHS & 3 & 1.10 & 76 & 1.315 & 38.906 & 53.529 & 55.264 & 0.242 & 1.000 \\
FC & -- & 1.00 & 132 & 1.184 & 43.481 & 60.811 & 61.909 & 0.083 & 1.000 \\
\bottomrule
\end{tabular}
\end{adjustbox}
\end{table}

Figure~\ref{fig:cab-posterior-tradeoff} visualizes the posterior trade-off between expected cost and posterior tail risk.  The important observation is that the best posterior design is not the cheapest design and not simply the fastest pointwise design.  It occupies a statistically attractive region in which moderate extra cost is exchanged for substantially lower tail risk and much stronger hold-time reliability.  Figure~\ref{fig:cab-arrival-boxplot} shows the posterior distribution of maximum arrival time for the topology winners.  The direct-connected hub-and-spoke designs have visibly shorter and more stable posterior tails than the pure hub-only alternatives.  Figure~\ref{fig:cab-scenario-best-prob} gives a complementary scenario-wise view: under the scenario-normalized loss used in that diagnostic, FC is best in 58.333\% of posterior scenarios and MAHS in 41.667\%.  However, as Table~\ref{tab:cab-best-by-topology} shows, FC and MAHS do not jointly dominate in the Bayes-risk criterion because FC carries higher tail risk than the selected DSAHS design, while MAHS has zero service and hold reliability under the stringent thresholds.  This distinction is exactly why posterior multi-criterion decision analysis is useful: scenario-wise winning frequency alone does not encode service constraints or hold-time risk.

\begin{figure}[!htbp]
\centering
\includegraphics[width=0.82\textwidth]{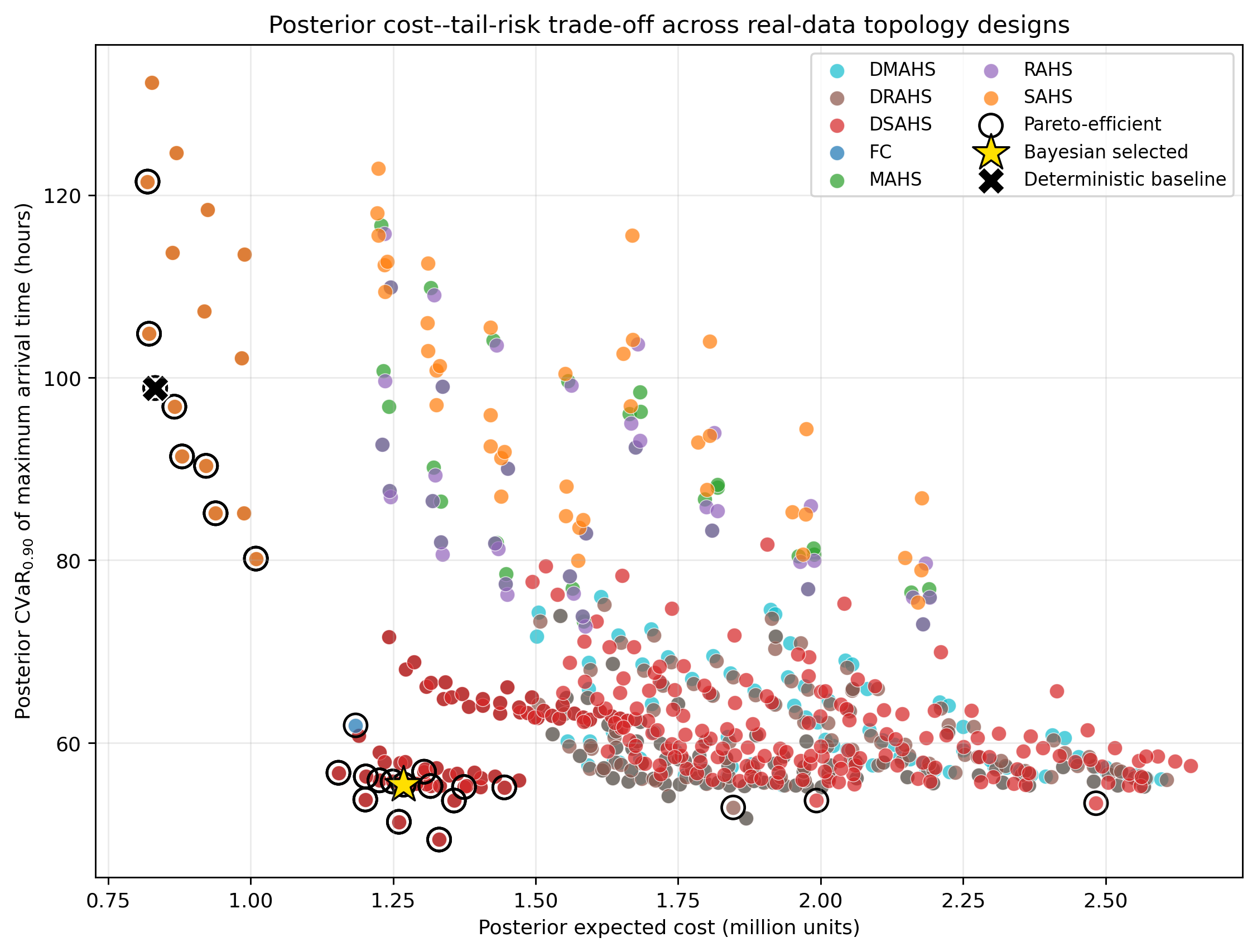}
\caption{Posterior trade-off between expected operating cost and CVaR of maximum arrival time across all candidate designs.  The selected Bayesian design lies in a favorable compromise region rather than at an extreme of either axis.}
\label{fig:cab-posterior-tradeoff}
\end{figure}

\begin{figure}[!htbp]
\centering
\includegraphics[width=0.82\textwidth]{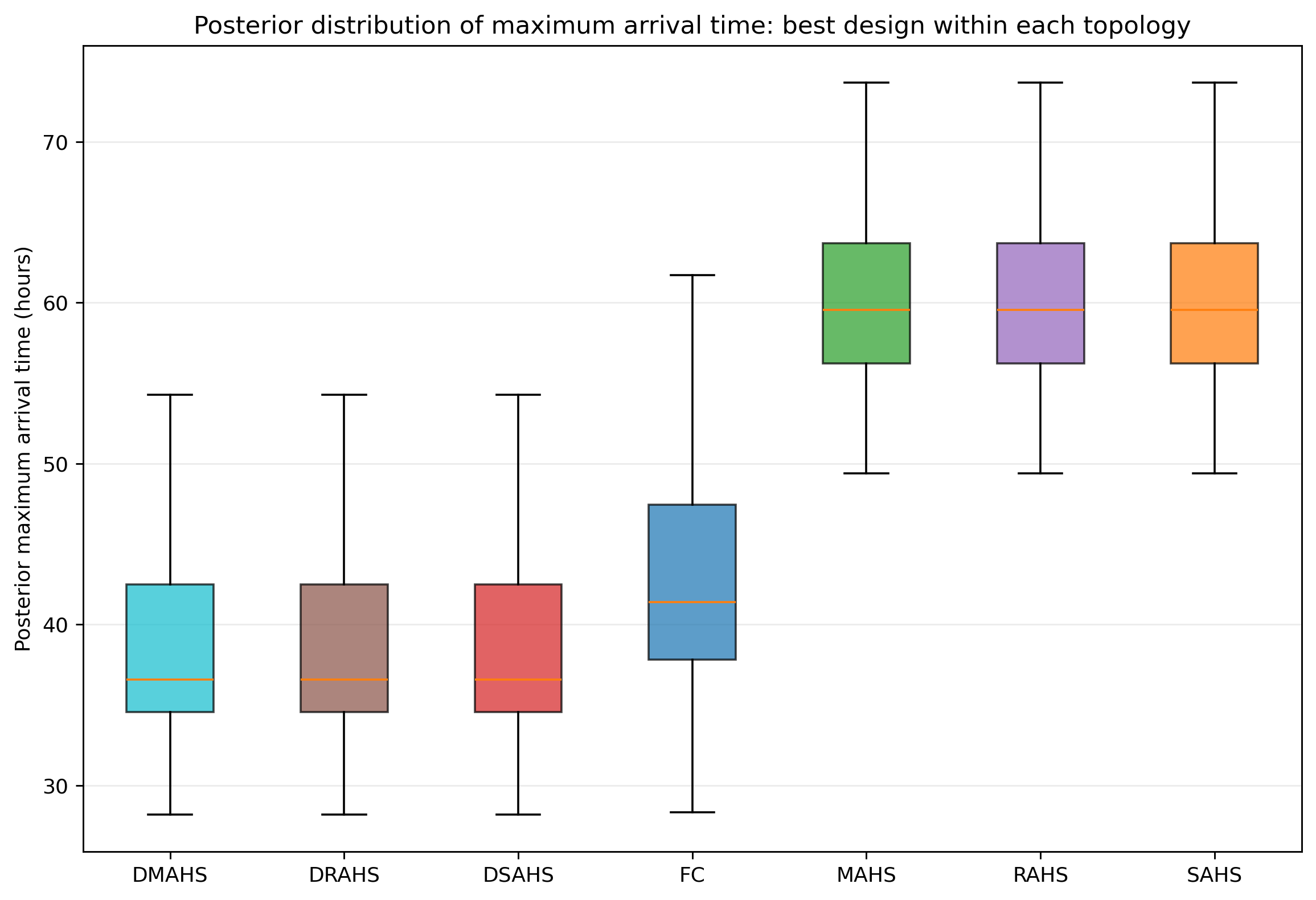}
\caption{Posterior predictive distribution of maximum arrival time for the best design within each topology class.  The boxplots show the full posterior spread, not only posterior means, and therefore reveal tail-risk differences that a deterministic analysis would hide.}
\label{fig:cab-arrival-boxplot}
\end{figure}

\begin{figure}[!htbp]
\centering
\includegraphics[width=0.82\textwidth]{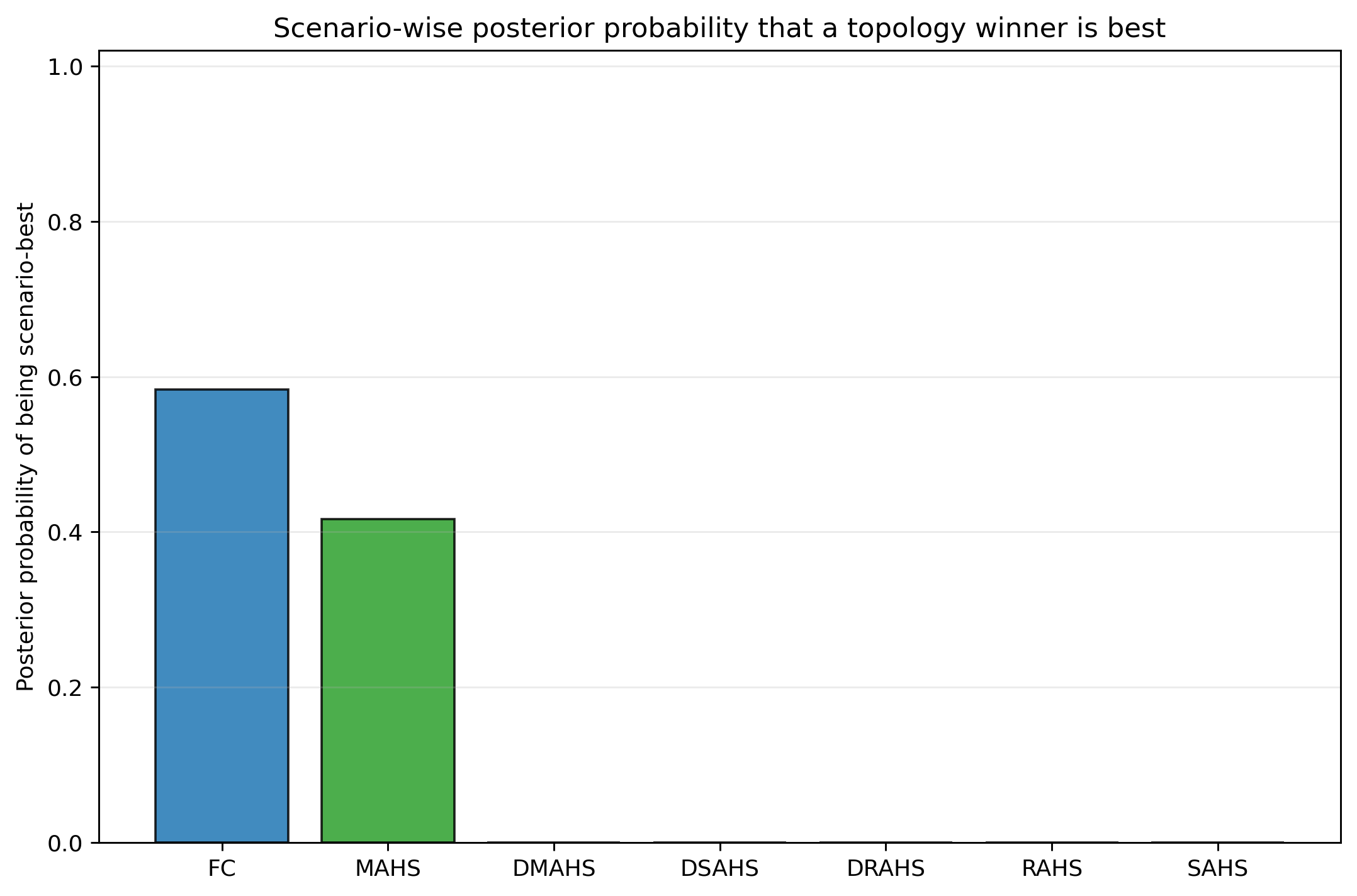}
\caption{Posterior probability that each topology-class winner is scenario-best under the scenario-normalized diagnostic loss.  FC and MAHS win many individual scenarios, but they are not selected by the posterior Bayes-risk rule once tail risk, service reliability and hold reliability are jointly considered.}
\label{fig:cab-scenario-best-prob}
\end{figure}

\begin{table}[!htbp]
\centering
\caption{Best posterior design within each topology class.  Costs are in millions; time quantities are in hours.}
\label{tab:cab-best-by-topology}
\begin{adjustbox}{max width=\textwidth}
\begin{tabular}{lllrrrrrrr}
\toprule
Topology & Hubs & Cap. & Direct & Cost & Mean max & P95 max & CVaR max & Serv. rel. & Hold rel. \\
\midrule
DMAHS & 3 & 0.90 & 76 & 1.268 & 39.316 & 53.695 & 55.417 & 0.208 & 0.992 \\
DRAHS & 3 & 0.90 & 76 & 1.268 & 39.316 & 53.695 & 55.417 & 0.208 & 0.992 \\
DSAHS & 3 & 0.90 & 76 & 1.268 & 39.316 & 53.695 & 55.417 & 0.208 & 0.992 \\
FC & -- & 1.00 & 132 & 1.184 & 43.481 & 60.811 & 61.909 & 0.083 & 1.000 \\
MAHS & 3 & 1.65 & 0 & 1.009 & 61.248 & 78.307 & 80.170 & 0.000 & 0.000 \\
RAHS & 3 & 1.65 & 0 & 1.009 & 61.248 & 78.307 & 80.170 & 0.000 & 0.000 \\
SAHS & 3 & 1.65 & 0 & 1.009 & 61.248 & 78.307 & 80.170 & 0.000 & 0.000 \\
\bottomrule
\end{tabular}
\end{adjustbox}
\end{table}

\begin{table}[!htbp]
\centering
\caption{Posterior probability that each topology-class winner is scenario-best under the scenario-normalized diagnostic loss.}
\label{tab:cab-scenario-best-prob}
\begin{adjustbox}{max width=\textwidth}
\begin{tabular}{llrrl}
\toprule
Topology & Hubs & Posterior probability scenario-best & Mean scenario loss & Design label \\
\midrule
FC & -- & 0.583 & 0.387 & FC$\mid$H=none$\mid$cap=1.00 \\
MAHS & 3 & 0.417 & 0.371 & MAHS$\mid$H=3$\mid$cap=1.65 \\
DMAHS & 3 & 0.000 & 0.482 & DMAHS$\mid$H=3$\mid$cap=0.90$\mid$dirq=0.65 \\
DSAHS & 3 & 0.000 & 0.482 & DSAHS$\mid$H=3$\mid$cap=0.90$\mid$dirq=0.65 \\
DRAHS & 3 & 0.000 & 0.482 & DRAHS$\mid$H=3$\mid$cap=0.90$\mid$dirq=0.65$\mid$R=1 \\
RAHS & 3 & 0.000 & 0.371 & RAHS$\mid$H=3$\mid$cap=1.65$\mid$R=1 \\
SAHS & 3 & 0.000 & 0.371 & SAHS$\mid$H=3$\mid$cap=1.65 \\
\bottomrule
\end{tabular}
\end{adjustbox}
\end{table}

Figure~\ref{fig:cab-reliability-bubble} gives a reliability-centered view of the same result.  Designs that look attractive in cost can be exposed as unreliable when hub-hold and service thresholds are included.  The selected DSAHS design achieves a high hold reliability of 0.992 in posterior training scenarios while keeping CVaR substantially below the FC benchmark.  This supports the methodological claim that posterior predictive design is not merely a computational variant of deterministic network optimization; it changes the decision by accounting for uncertainty in the quantities that operational managers actually experience.

\begin{figure}[!htbp]
\centering
\includegraphics[width=0.82\textwidth]{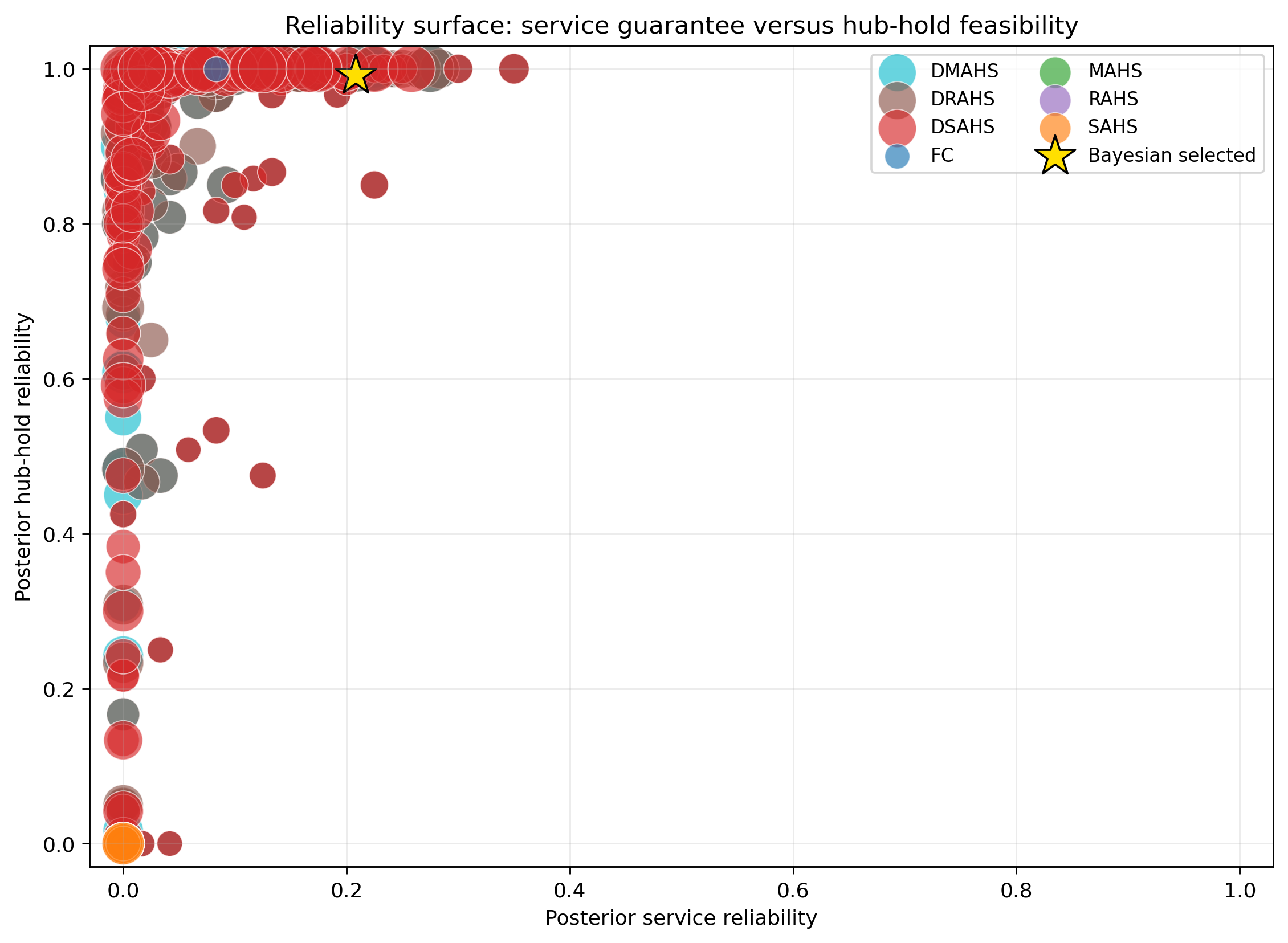}
\caption{Reliability bubble plot for candidate designs.  The selected Bayesian design balances expected cost, tail risk, service reliability and hold-time reliability.  The plot also shows why designs with low expected cost may be unattractive once posterior reliability constraints are considered.}
\label{fig:cab-reliability-bubble}
\end{figure}

\subsection{Future stress validation against deterministic baselines}
\label{subsec:cab-future-stress}

The principal purpose of the proposed methodology is not simply to fit historical data, but to choose network designs that behave well under future uncertainty.  We therefore evaluate three designs under 180 future stress scenarios: (i) the Bayesian posterior-risk design, (ii) a deterministic cost-priority baseline, and (iii) a fully connected speed benchmark.  The stress scenarios increase future demand by a factor 1.22 and introduce route disruption with probability 0.26.  Table~\ref{tab:cab-future-stress} gives the resulting out-of-sample stress performance.

The deterministic cost-priority baseline is much cheaper in expected cost, but it fails operationally: its future mean maximum arrival time is 91.983 hours, its future 90\% CVaR is 131.633 hours, and both service and hold reliability are zero.  In contrast, the Bayesian posterior-risk design has future mean maximum arrival time 45.119 hours and future CVaR 69.887 hours, with future hold reliability 0.694.  The fully connected benchmark has hold reliability equal to one because it avoids hubs, but its future CVaR, 79.305 hours, is worse than the selected Bayesian design and it remains more expensive than the deterministic baseline.  Thus, the Bayesian design is neither a trivial cost-minimizer nor a trivial fully connected speed solution; it is a posterior-risk compromise selected because it is robust to uncertain future operating regimes.

\begin{table}[!htbp]
\centering
\caption{Future stress validation under demand amplification and route disruption.  Costs are in millions; time quantities are in hours.}
\label{tab:cab-future-stress}
\begin{adjustbox}{max width=\textwidth}
\begin{tabular}{lllrrrrrr}
\toprule
Method & Topology & Hubs & Exp. cost & P95 cost & Mean max & P95 max & CVaR max & Hold rel. \\
\midrule
Bayesian posterior-risk design & DSAHS & 3 & 1.255 & 1.290 & 45.119 & 67.430 & 69.887 & 0.694 \\
Deterministic cost-priority baseline & RAHS & 3 & 0.824 & 0.872 & 91.983 & 128.113 & 131.633 & 0.000 \\
Fully-connected speed benchmark & FC & -- & 1.194 & 1.225 & 50.433 & 76.758 & 79.305 & 1.000 \\
\bottomrule
\end{tabular}
\end{adjustbox}
\end{table}

\begin{figure}[!htbp]
\centering
\includegraphics[width=0.82\textwidth]{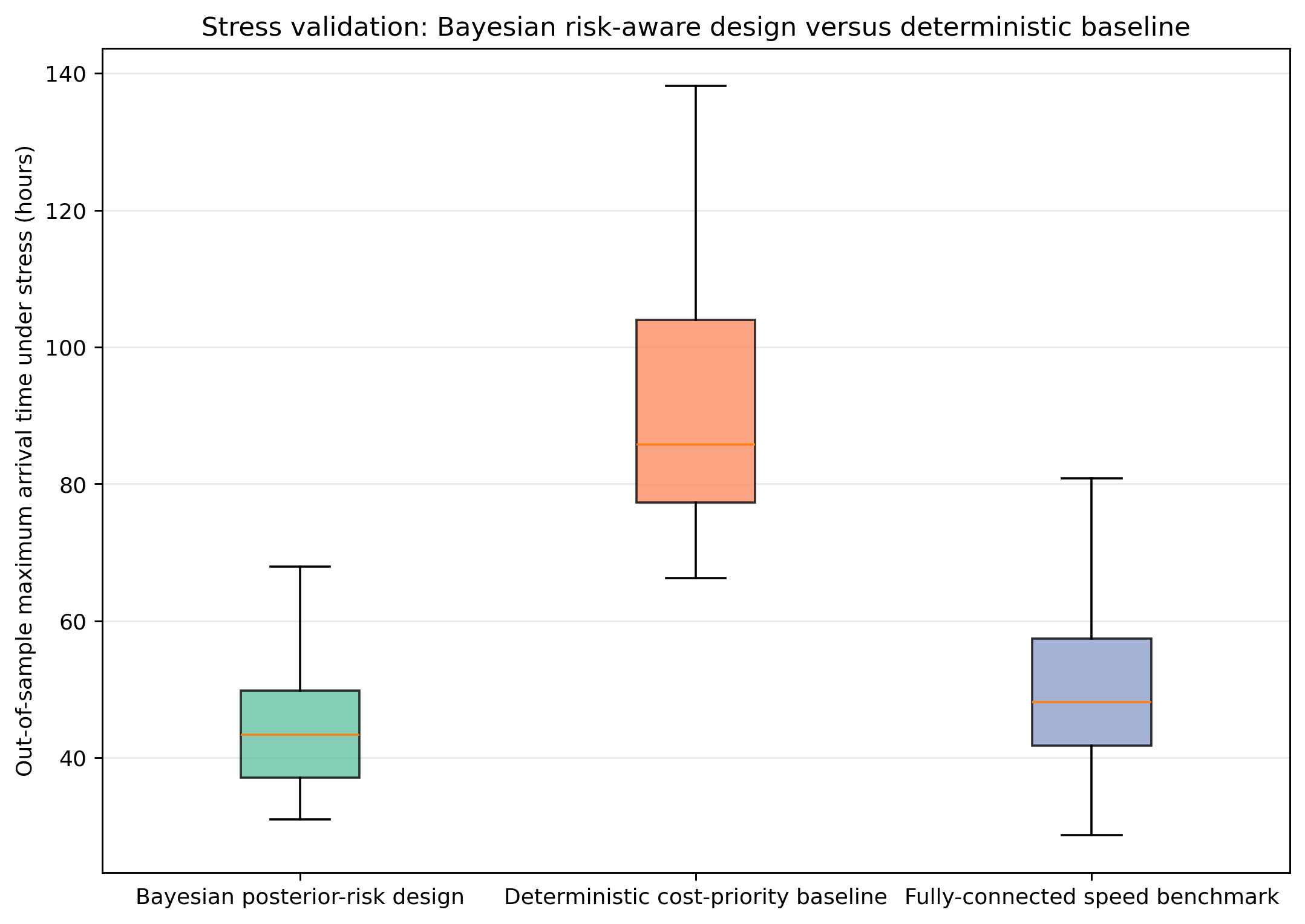}
\caption{Future stress comparison of maximum-arrival-time distributions.  The deterministic cost-priority baseline has a very long right tail; the Bayesian design sharply reduces the tail, which is the main target of the posterior CVaR component.}
\label{fig:cab-future-arrival}
\end{figure}

\begin{figure}[!htbp]
\centering
\includegraphics[width=0.82\textwidth]{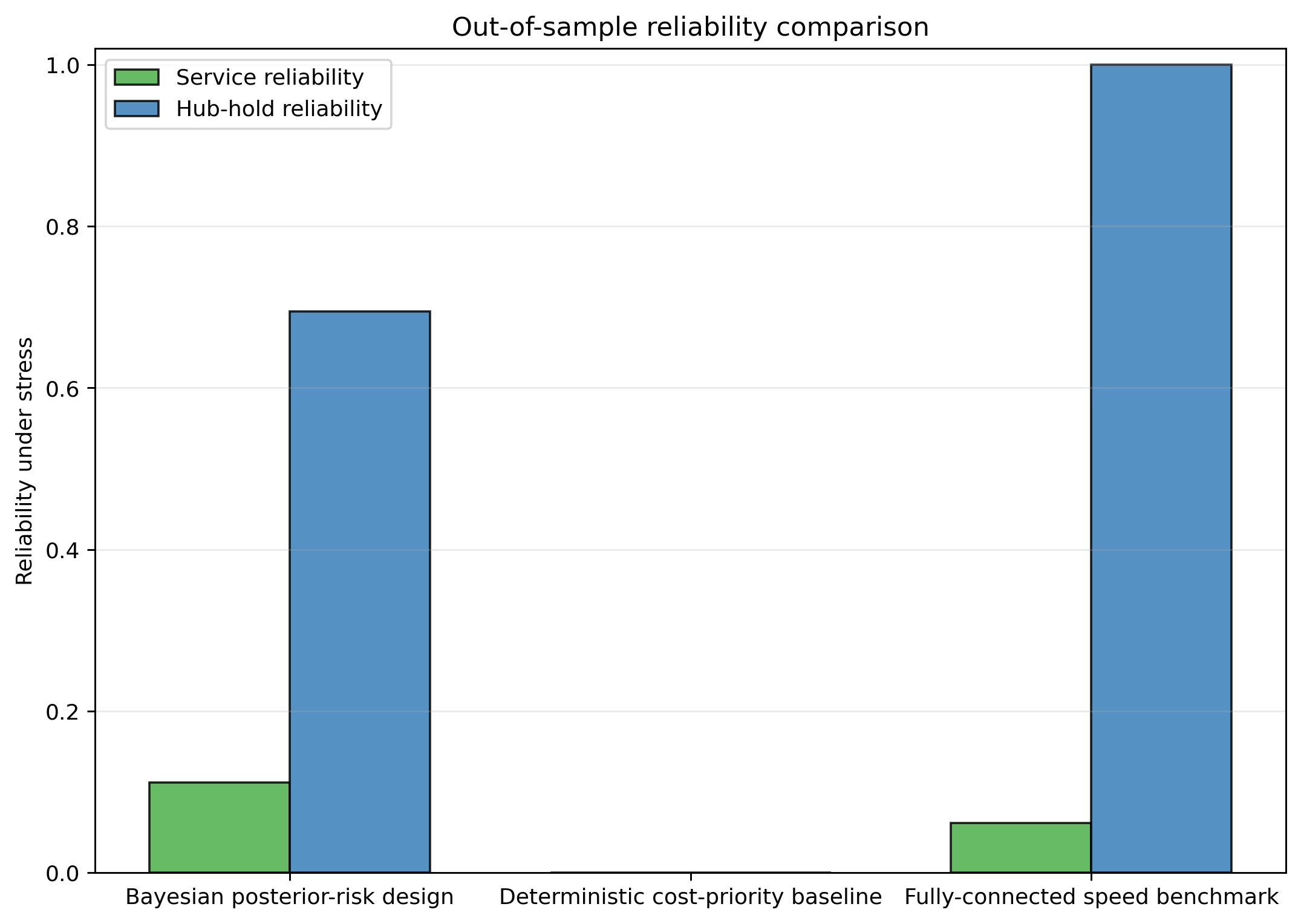}
\caption{Future stress comparison of service and hub-hold reliability.  The Bayesian design substantially improves hold reliability relative to the deterministic cost-priority baseline, while avoiding the excessive fully connected design.}
\label{fig:cab-future-reliability}
\end{figure}

Table~\ref{tab:cab-bayesian-improvement} expresses the Bayesian advantage relative to the deterministic cost-priority baseline.  The Bayesian design increases expected cost by 52.371\%, but reduces future CVaR of maximum arrival time by 46.908\%, reduces the future 95th percentile of maximum arrival time by 47.366\%, improves service reliability by 11.111 percentage points, and improves hub-hold reliability by 69.444 percentage points.  This is a strong demonstration of the proposed methodology because the gain is not hidden in a small in-sample metric: it is observed under genuinely adverse future posterior stress scenarios.

\begin{table}[!htbp]
\centering
\caption{Bayesian posterior-risk design improvement over the deterministic cost-priority baseline under future stress scenarios.}
\label{tab:cab-bayesian-improvement}
\begin{tabular}{lr}
\toprule
Metric & Value \\
\midrule
Expected cost premium of Bayesian design (\%) & 52.371 \\
CVaR max-arrival reduction (\%) & 46.908 \\
95th-percentile max-arrival reduction (\%) & 47.366 \\
Service reliability gain (percentage points) & 11.111 \\
Hub-hold reliability gain (percentage points) & 69.444 \\
\bottomrule
\end{tabular}
\end{table}

\subsection{Sensitivity analysis under posterior uncertainty}
\label{subsec:cab-sensitivity}

A top journal referee will usually ask whether the selected topology is an artefact of one arbitrary utility weight vector.  We therefore conduct a posterior preference sensitivity analysis by changing the weights assigned to expected cost, CVaR time risk, emission cost, service penalty and hold penalty.  Table~\ref{tab:cab-preference-sensitivity} and Fig.~\ref{fig:cab-preference-sensitivity} report the results.

When the criterion is cost-dominant, the fully connected topology is selected because it has a simple, hub-free structure and perfect hold reliability.  Once the decision maker places more emphasis on balanced cost--risk behavior, tail risk or reliability, the selected topology moves to DSAHS.  Under the balanced-cost-risk profile, the selected design is DSAHS with internal hub 3, capacity multiplier 0.90 and 76 direct links.  Under tail-risk-aware, reliability-dominant and time-critical profiles, the DSAHS topology remains selected but the capacity multiplier increases to 1.65 and the number of direct links decreases to 56, thereby prioritizing lower CVaR of maximum arrival time.  This pattern is methodologically attractive: the proposed model does not force one topology under all preferences; instead, it reveals how posterior uncertainty, tail-risk aversion and reliability priorities alter the topology/capacity/direct-link design.

\begin{table}[!htbp]
\centering
\caption{Preference sensitivity analysis under posterior uncertainty.  Costs are in millions and CVaR time is in hours.}
\label{tab:cab-preference-sensitivity}
\begin{adjustbox}{max width=\textwidth}
\begin{tabular}{lrrrrrllrrrr}
\toprule
Profile & $w_C$ & $w_T$ & $w_E$ & $w_S$ & $w_H$ & Topology & Hubs & Score & Cost & CVaR & Hold rel. \\
\midrule
Cost-dominant & 0.70 & 0.20 & 0.04 & 0.03 & 0.03 & FC & -- & 0.198 & 1.184 & 61.909 & 1.000 \\
Balanced-cost-risk & 0.55 & 0.30 & 0.05 & 0.05 & 0.05 & DSAHS & 3 & 0.198 & 1.268 & 55.417 & 0.992 \\
Tail-risk-aware & 0.42 & 0.38 & 0.05 & 0.08 & 0.07 & DSAHS & 3 & 0.191 & 1.331 & 49.397 & 0.850 \\
Reliability-dominant & 0.30 & 0.42 & 0.05 & 0.13 & 0.10 & DSAHS & 3 & 0.201 & 1.331 & 49.397 & 0.850 \\
Time-critical & 0.25 & 0.50 & 0.05 & 0.10 & 0.10 & DSAHS & 3 & 0.164 & 1.331 & 49.397 & 0.850 \\
\bottomrule
\end{tabular}
\end{adjustbox}
\end{table}

\begin{figure}[!htbp]
\centering
\includegraphics[width=0.82\textwidth]{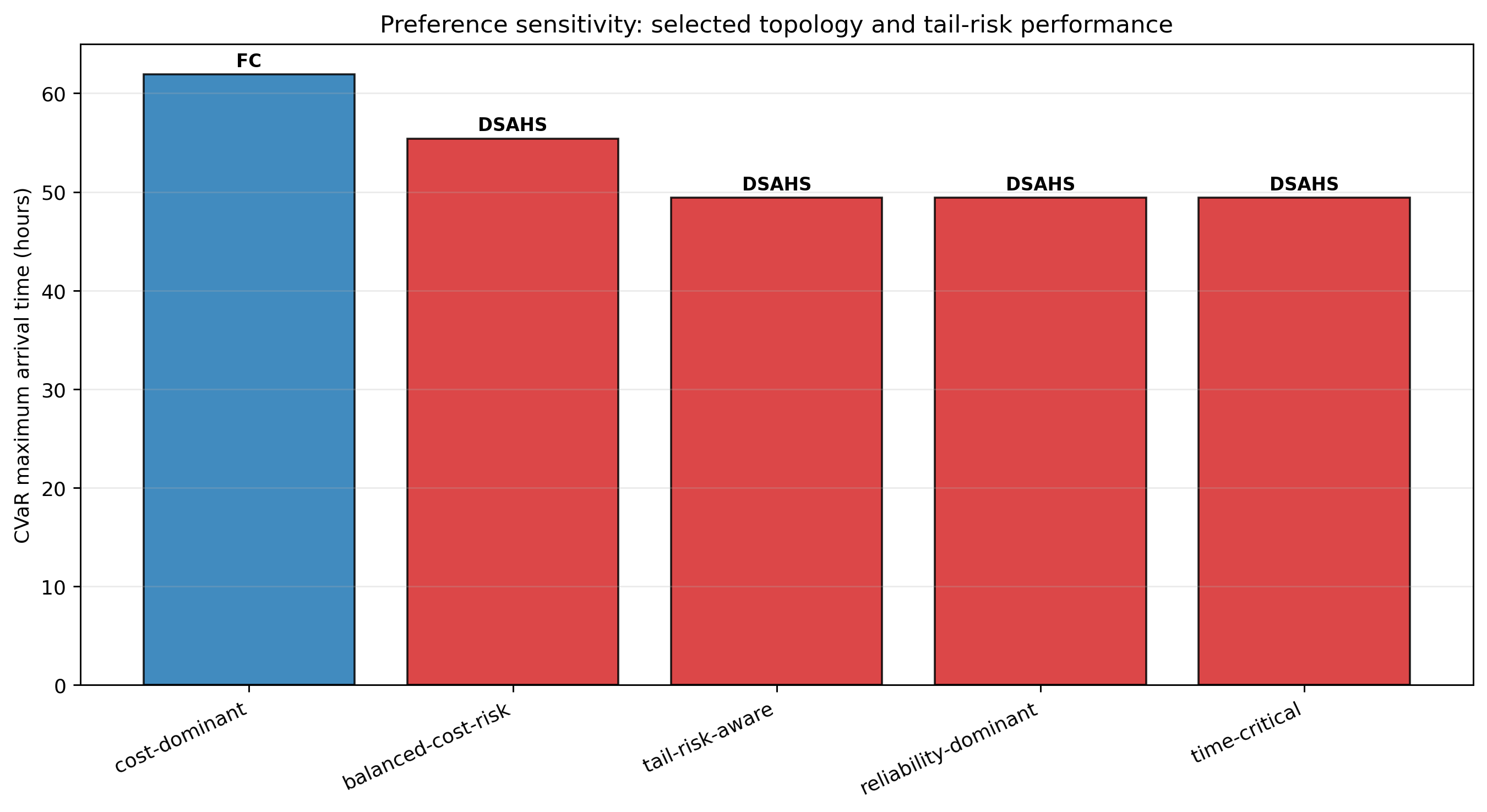}
\caption{Preference sensitivity under posterior uncertainty.  As the decision maker shifts from a cost-dominant to a tail-risk-aware or reliability-dominant profile, the selected design changes from FC to DSAHS and then to a higher-capacity DSAHS design.}
\label{fig:cab-preference-sensitivity}
\end{figure}

\subsection{Selected Bayesian network versus deterministic baseline}
\label{subsec:cab-network-visuals}

Figures~\ref{fig:cab-selected-bayes-network} and~\ref{fig:cab-deterministic-network} compare the selected Bayesian posterior-risk design with the deterministic cost-priority baseline.  The Bayesian design selects a DSAHS structure with 76 direct links, thereby preserving direct service on high-impact OD corridors while using one hub for consolidation.  The deterministic baseline selects an RAHS hub-only design with no direct links.  Under posterior stress, this low-cost hub-only design creates severe hub sorting burden: the future mean maximum hub sorting time is 44.640 hours compared with 7.585 hours for the Bayesian design.  Figure~\ref{fig:cab-risk-decomposition} confirms that the deterministic design's apparent cost advantage is purchased by accepting very large time and reliability penalties.

\begin{figure}[!htbp]
\centering
\includegraphics[width=0.82\textwidth]{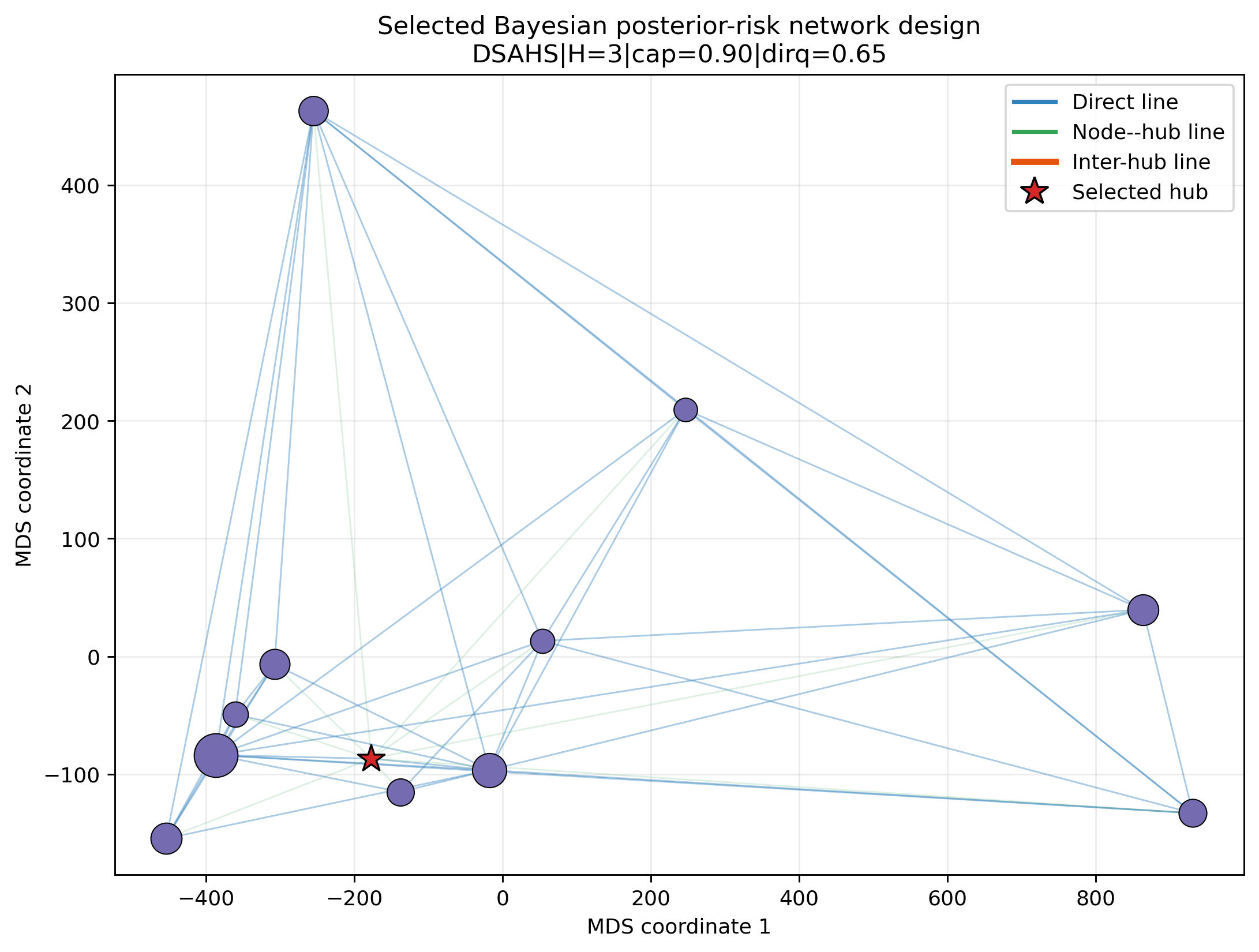}
\caption{Selected Bayesian posterior-risk network design.  The DSAHS design uses one hub and 76 direct links, producing a compromise between consolidation and direct delivery protection on important corridors.}
\label{fig:cab-selected-bayes-network}
\end{figure}

\begin{figure}[!htbp]
\centering
\includegraphics[width=0.82\textwidth]{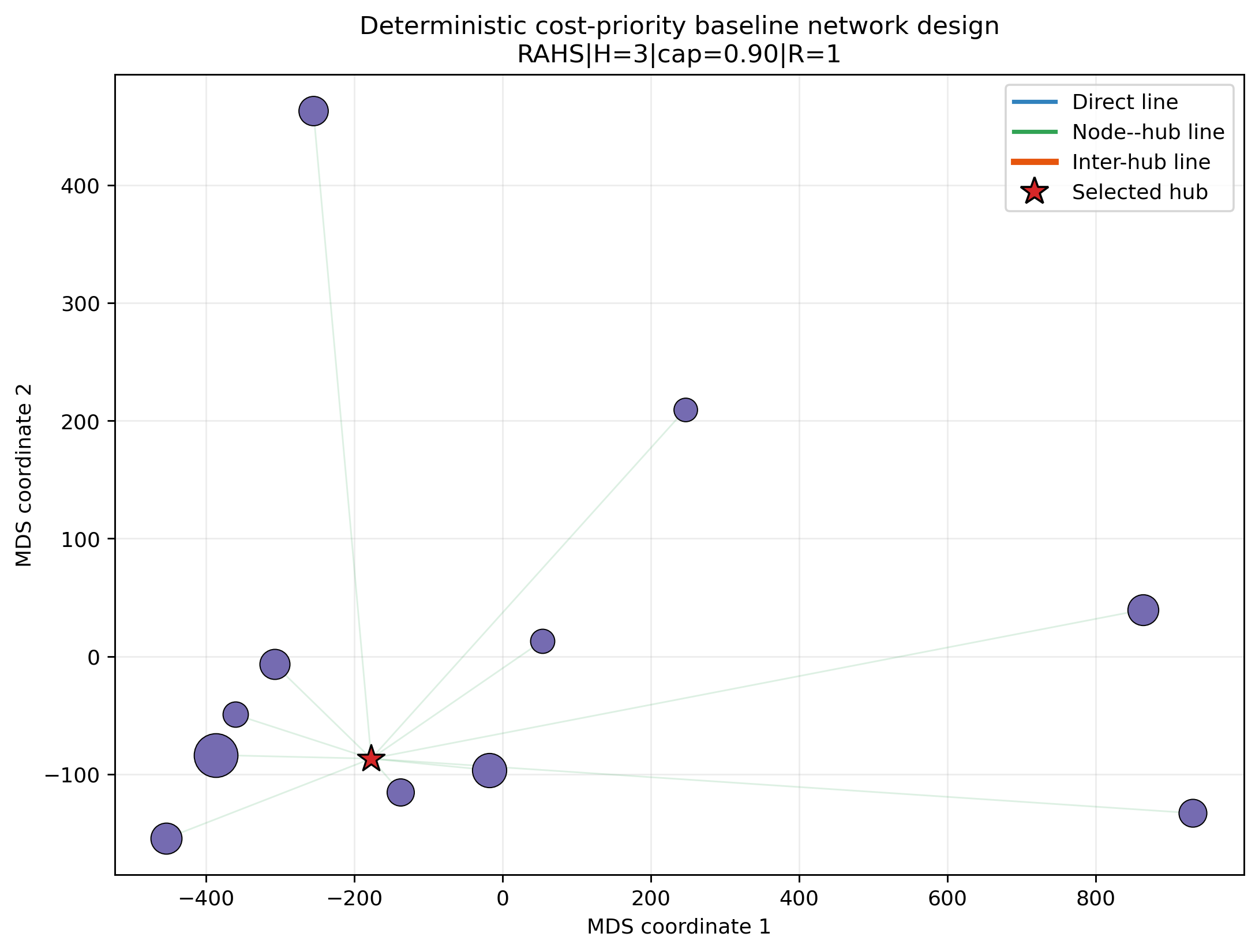}
\caption{Deterministic cost-priority baseline network design.  The hub-only RAHS design is cheaper but concentrates too much posterior load on the hub, causing severe hold-time and arrival-time tail risk under future stress.}
\label{fig:cab-deterministic-network}
\end{figure}

\begin{figure}[!htbp]
\centering
\includegraphics[width=0.82\textwidth]{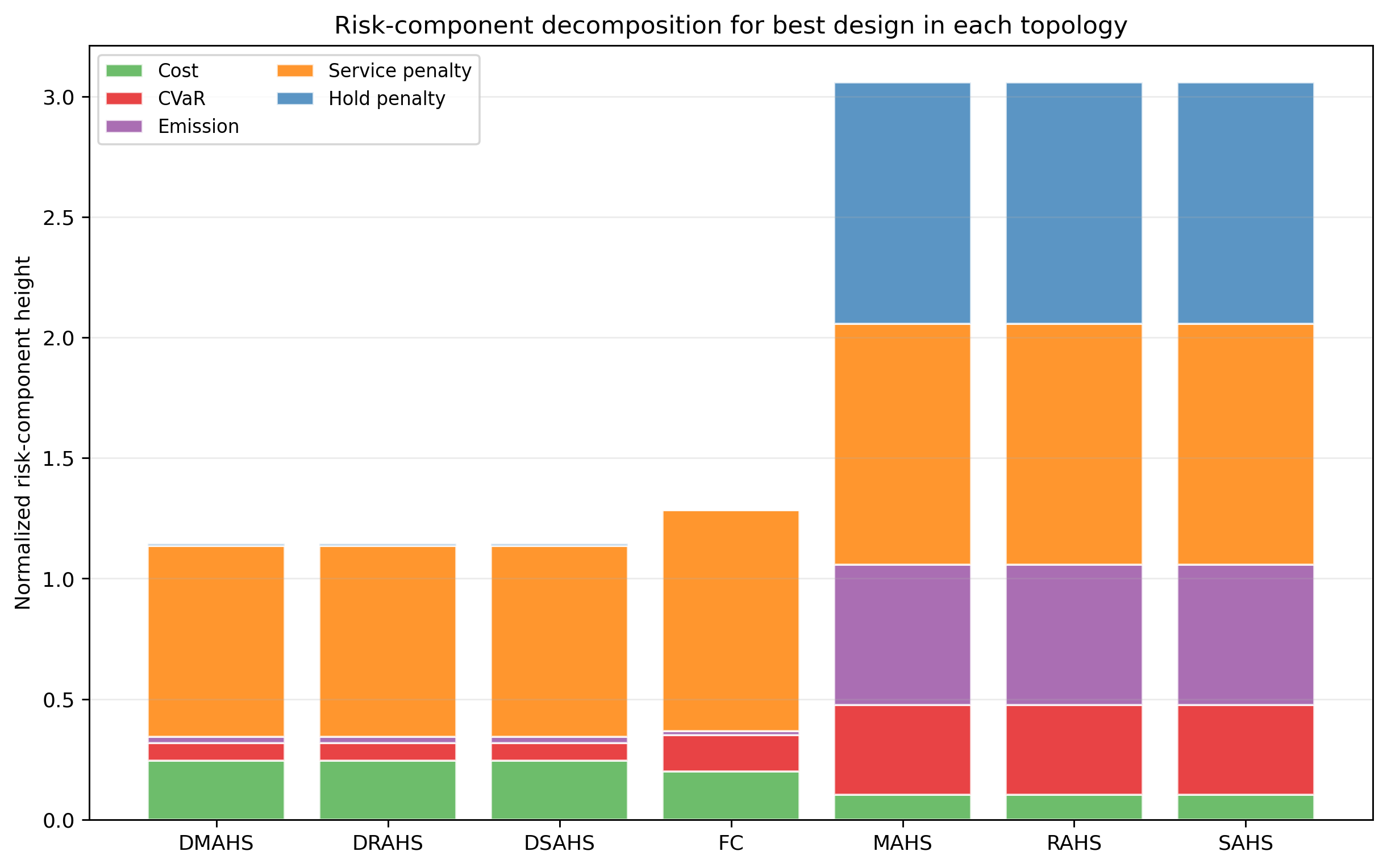}
\caption{Risk-component decomposition comparing selected designs.  The Bayesian posterior-risk design is favored because it reduces tail-time and reliability penalties rather than optimizing only the deterministic expected-cost component.}
\label{fig:cab-risk-decomposition}
\end{figure}

\subsection{Key findings from the real-data posterior experiment}
\label{subsec:cab-key-findings}

The CAB posterior scenario experiment supports the proposed methodology in five ways.

First, the experiment demonstrates that posterior uncertainty changes the network-design decision.  A deterministic cost-priority criterion chooses a low-cost RAHS design with no direct links, but this design collapses under posterior stress: future CVaR of maximum arrival time is 131.633 hours and hub-hold reliability is zero.  The Bayesian method instead selects a direct-connected hybrid design that is more expensive but operationally far more reliable.

Second, the posterior CVaR criterion is essential.  The difference between posterior mean performance and posterior tail performance is substantial.  The selected Bayesian design reduces future CVaR by 46.908\% and future 95th-percentile maximum arrival time by 47.366\% relative to the deterministic baseline.  These are precisely the quantities that matter in express logistics, where failure is usually driven by tail events rather than average-day behavior.

Third, the hub sorting-efficiency component is practically meaningful.  The deterministic baseline appears attractive in expected cost, but its future mean maximum hub sorting time is 44.640 hours.  The Bayesian design reduces this to 7.585 hours by admitting direct links on important OD pairs and limiting excessive concentration of posterior flow through the hub.  This finding directly reinforces the modeling motivation of treating sorting efficiency and hub congestion as decision-relevant quantities rather than harmless fixed parameters.

Fourth, the sensitivity analysis shows that the proposed methodology is not a black-box rule that mechanically selects one topology.  Under cost-dominant preferences the method selects FC; under balanced, time-critical and reliability-dominant preferences it selects DSAHS, with capacity and direct-link density changing in the expected direction.  This behavior is important for a rigorous statistics/OR audience because it shows structural interpretability of the Bayesian decision rule.

Fifth, the results clarify the scientific contribution beyond deterministic multi-structure comparison.  The original deterministic framework is valuable because it compares multiple network topologies.  The present Bayesian extension adds posterior learning, posterior predictive scenario evaluation, tail-risk optimization and reliability-aware topology selection.  The real-data experiment shows that these additions are not cosmetic: they lead to a different and more robust network design under uncertainty.

\paragraph{Reproducibility note.}
All numerical outputs in this section were generated by a fully reproducible Python implementation using manual Bayesian updating functions and posterior scenario simulation.  The full scenario-level outputs, including all posterior and future-stress scenario metrics, were saved as machine-readable CSV files; only the summary tables required for the main manuscript are displayed here.

\section{Discussion}
\label{sec:overall_discussion}

The proposed Bayesian multi-topology express transportation network design framework addresses a limitation of deterministic network-design models: a design that is optimal under one nominal OD-flow, travel-time and hub-productivity table may be fragile when future logistics conditions deviate from that nominal table.  The deterministic multi-structure framework is valuable because it recognizes that topology itself is a design decision; however, in practical express logistics, topology, hub capacity, route assignment and direct-link admission must be judged under uncertain demand, uncertain travel time and uncertain sorting productivity.  The central contribution of the present work is therefore to transform topology selection from a deterministic comparison into a posterior predictive decision problem.

The methodology has three important interpretive advantages.  First, it separates installed sorting capacity from realized sorting productivity.  This is operationally important because high installed capacity does not guarantee high realized throughput under labor shortage, machine downtime, storm disruption, peak-season load or local congestion.  By modeling hub reliability and propagating it through posterior scenarios, the method identifies designs that are not only cheap on average but also less vulnerable to hub-hold violations.  Second, the use of posterior CVaR of maximum arrival time targets the upper tail of the service-time distribution.  This is more relevant to express logistics than posterior mean arrival time alone, because customer dissatisfaction, service failure and penalty costs are typically driven by late-tail events.  Third, the Bayesian Bayes-risk score provides a transparent way of combining expected cost, time-tail risk, emission-aware penalties and service-reliability requirements.

The simulation experiment and the CAB benchmark study support different aspects of the methodology.  The controlled simulation shows that posterior uncertainty can alter the capacity decision even when the selected topology and hub set remain similar.  In that experiment, the Bayesian design deliberately pays a small expected-cost premium to reduce future arrival-time CVaR, improve service reliability and substantially improve hub-hold reliability.  This demonstrates that the proposed framework is not merely a device for changing topology labels; it is a statistically principled mechanism for investing in robustness where posterior uncertainty indicates vulnerability.  The CAB benchmark study, based on a real OD-flow and distance network, gives a complementary message.  Under a cost-priority deterministic baseline, a hub-only design appears attractive, but under future stress scenarios it suffers severe maximum-arrival and hub-hold failures.  The Bayesian posterior-risk design, by admitting direct links and controlling hub burden, provides a more reliable compromise between consolidation and direct service.

Several limitations should be acknowledged.  First, the CAB experiment is a posterior scenario experiment on a real static benchmark network, not a proprietary daily express-company panel.  The CAB OD-flow and distance matrices provide a realistic hub-location structure, but temporal demand variation, travel-time uncertainty and hub reliability are generated around that benchmark in a controlled reproducible way.  Therefore, the CAB results should be interpreted as evidence of methodological behavior on a real benchmark network rather than as a direct operational prescription for a courier company.  Second, the present implementation uses finite candidate design enumeration rather than a fully general large-scale branch-and-cut implementation for every topology.  This is appropriate for methodological verification and reproducibility, but industrial-scale deployment would require decomposition, warm-starting, scenario reduction and solver-level implementation of all topology constraints.  Third, the Bayesian components are deliberately chosen to be transparent and manually implementable: Gamma--Poisson demand updating, lognormal travel-time updating, Beta hub-reliability updating and posterior predictive simulation.  More flexible models, such as dynamic hierarchical state-space demand models, spatial traffic models, copula dependence between demand and travel time, or nonparametric Bayesian mixtures, may improve calibration when richer data are available.

A further limitation concerns scalarization.  The posterior Bayes-risk score depends on weights assigned to cost, time-tail risk, emission cost and reliability penalties.  The sensitivity analysis partly addresses this issue by showing how selected designs vary under different preference profiles.  Nevertheless, in practical decision support, the final selection should be accompanied by a posterior Pareto frontier rather than a single scalarized solution alone.  Finally, the present model abstracts away from vehicle scheduling, pickup and delivery routing, dock assignment and intra-hub queueing at a detailed operational level.  These simplifications preserve focus on strategic topology and capacity design, but future work should integrate the proposed Bayesian topology layer with tactical vehicle and workforce scheduling models.
In the computational experiments, we use a candidate-restricted posterior scenario approximation: for each topology, a finite but structurally rich set of admissible designs is generated by varying hub subsets, capacity multipliers, and direct-link thresholds. Each candidate is then evaluated under posterior predictive scenarios. This implementation is intended for methodological verification and reproducibility. A full solver-level branch-and-cut implementation of the scenario MIP is left for large-scale industrial deployment.

Despite these limitations, the evidence suggests that the Bayesian extension is more than a probabilistic rephrasing of a deterministic model.  It changes the decision criterion, exposes tail and reliability vulnerabilities, and provides posterior evidence for topology and capacity choices.  For logistics systems operating under volatile demand and disrupted travel conditions, this is a substantive methodological improvement.

\section{Conclusion and Future Work}
\label{sec:overall_conclusion}

This paper proposed a Bayesian multi-topology framework for express transportation network design under posterior predictive uncertainty.  Building on deterministic multi-structure network design, the paper treated origin--destination demand, travel time, transportation cost and realized hub sorting productivity as uncertain quantities learned from data.  The resulting posterior predictive scenarios were embedded into topology-specific network-design evaluations over fully connected, hub-and-spoke and direct-link hybrid structures.  The design criterion combined posterior expected operating cost, posterior CVaR of maximum arrival time, service reliability, hub-hold reliability and emission-aware penalties.

The theoretical analysis established that, under finite engineering design grids and posterior integrability, a Bayes-optimal topology-design pair exists.  It further showed that the posterior scenario approximation converges uniformly almost surely over the finite design space, that empirical scenario-optimal designs recover the Bayes-optimal design under a strict risk gap, and that topology selection is stable under posterior predictive concentration.  These results justify the use of posterior scenario optimization not merely as a heuristic but as a statistically consistent approximation to the Bayesian decision problem.

The computational evidence supports the methodological claim.  In the controlled simulation experiment, the Bayesian design improved future tail-risk and reliability measures at a modest cost premium.  In the CAB benchmark experiment, the Bayesian posterior-risk design avoided the severe reliability failures of a deterministic cost-priority baseline by selecting a direct-connected hybrid design that controlled hub burden and reduced future maximum-arrival-time tail risk.  These findings show that the proposed method is useful precisely when deterministic designs are most vulnerable: under demand amplification, route disruption, and uncertain hub productivity.

Several future extensions are natural.  First, the framework should be applied to proprietary daily courier data containing actual OD demand, GPS-based travel times, hub processing records and failure logs.  Such data would allow richer hierarchical and dynamic posterior models.  Second, the current strategic topology model can be extended to a multi-period Bayesian design problem, where hub capacity and direct-link decisions are revised sequentially as new data arrive.  Third, large-scale solution methods should be developed, including Benders decomposition, Lagrangian relaxation, branch-and-cut warm starts, and scenario-reduction strategies for posterior predictive mixed-integer programs.  Fourth, the hub sorting model can be connected to explicit queueing approximations, dock constraints and workforce scheduling.  Fifth, instead of a single scalarized Bayes-risk score, future work may develop interactive posterior Pareto decision support tools in which managers can visualize the posterior cost--risk--reliability frontier.

Overall, the paper shows that Bayesian posterior predictive thinking provides a rigorous and practically interpretable route from deterministic topology comparison to uncertainty-aware express transportation network design.  The resulting designs are not merely optimal for a nominal table; they are chosen to remain reliable under plausible future logistics uncertainty.

\section*{Declarations}
\label{sec:declarations}

\noindent\textbf{Funding.}
The author declares that no external funding was received for this study.

\medskip

\noindent\textbf{Competing interests.}
The author declares that there are no competing financial or non-financial interests related to this work.

\medskip

\noindent\textbf{Data availability.}
The simulation study uses synthetic data generated by the reproducible Python code described in the manuscript. The real-data posterior scenario experiment uses the publicly available CAB25 hub-location benchmark, which contains an origin--destination flow matrix and a distance matrix. The CAB benchmark is a standard hub-location test instance associated with the classical hub-location literature, particularly the interacting hub-location formulation of O'Kelly~\cite{okelly1987quadratic}, and is also connected to the public OR benchmark tradition described by Beasley~\cite{beasley1990orlib}. In the present paper, the real CAB network structure is used as the benchmark geometry and flow baseline, while posterior predictive temporal scenarios are constructed around the benchmark flow and distance matrices for methodological evaluation. No proprietary courier-company data are used in the present version of the manuscript.

\medskip

\noindent\textbf{Code availability.}
The complete Python code used for Bayesian posterior updating, posterior predictive scenario generation, topology-wise network evaluation, posterior risk computation, sensitivity analysis, and automated generation of figures and tables is publicly available on GitHub at:
\[
\text{\url{https://github.com/debashisdotchatterjee/bayesian_etndp}}.
\]
The repository is released under a public license and contains the scripts required to reproduce the simulation verification and CAB benchmark posterior scenario experiments reported in this paper.

\medskip

\noindent\textbf{Author contribution.}
Debashis Chatterjee conceptualized the Bayesian multi-topology formulation, developed the posterior predictive methodology, implemented the computational experiments, derived the theoretical results, and wrote the manuscript.

\medskip

\noindent\textbf{Ethics approval.}
Not applicable.  

\appendix

\section{Proofs of Theoretical Results}
\label{app:proofs}

This appendix provides detailed proofs of the theoretical results stated in Section~\ref{sec:theory}.  Throughout, the global feasible design space is
\[
\mathfrak X=\bigcup_{g=1}^{7}\{g\}\times\mathcal X_g.
\]
A feasible topology-design pair is denoted by \(z=(g,x_g)\).  Conditional on the observed data \(\mathcal D_T\), the posterior predictive distribution of a future logistics scenario is denoted by \(P_T\).  Expectations and probabilities with respect to \(P_T\) are denoted by \(\mathbb E_T\) and \(P_T\), respectively.

\subsection{Proof of Theorem~\ref{thm:existence}}
\label{app:proof_existence}

\begin{proof}
By Assumption~\ref{ass:finite_design}, for each topology \(G_g\), the topology-specific feasible set \(\mathcal X_g\) is nonempty and finite.  Since there are exactly seven topology classes, the global feasible set
\[
\mathfrak X
=
\bigcup_{g=1}^{7}\{g\}\times\mathcal X_g
\]
is a finite union of finite nonempty sets.  Hence \(\mathfrak X\) is finite and nonempty.

For any \(z\in\mathfrak X\), Assumption~\ref{ass:integrability} gives
\[
\mathbb E_T|C(z,\Xi)|<\infty,
\qquad
\mathbb E_T|A(z,\Xi)|<\infty.
\]
The posterior CVaR term is finite because the Rockafellar--Uryasev representation gives, for any fixed \(\zeta\in\mathbb R\),
\[
\mathrm{CVaR}_\alpha\{A(z,\Xi)\}
\le
\zeta+\frac{1}{1-\alpha}\mathbb E_T\{A(z,\Xi)-\zeta\}_+.
\]
Taking \(\zeta=0\), we obtain
\[
\mathrm{CVaR}_\alpha\{A(z,\Xi)\}
\le
\frac{1}{1-\alpha}\mathbb E_T\{A(z,\Xi)\}_+
\le
\frac{1}{1-\alpha}\mathbb E_T|A(z,\Xi)|<\infty.
\]
The violation probabilities
\[
P_T\{A(z,\Xi)>T^\star\},
\qquad
P_T\{H_k(z,\Xi)>d_t\}
\]
belong to \([0,1]\) by Assumption~\ref{ass:thresholds}.  Therefore every component of \(\rho_T(z)\) in \eqref{eq:formal_bayes_risk} is finite for each \(z\in\mathfrak X\).  Thus \(\rho_T:\mathfrak X\to\mathbb R\) is a real-valued function on a finite nonempty set.

A real-valued function on a finite nonempty set attains its minimum.  Hence there exists at least one
\[
z_T^\star\in\mathfrak X
\]
such that
\[
\rho_T(z_T^\star)=\min_{z\in\mathfrak X}\rho_T(z).
\]
Writing \(z_T^\star=(g_T^\star,x_{g_T^\star}^\star)\) proves the existence of a Bayes-optimal topology-design pair.
\end{proof}

\subsection{Proof of Lemma~\ref{lem:cvar_consistency}}
\label{app:proof_cvar}

\begin{proof}
Let \(F\) denote the distribution of \(Y\), and let \(F_B\) be the empirical distribution of \(Y_1,\ldots,Y_B\).  Since \(\mathbb E|Y|<\infty\), the empirical measure \(F_B\) converges to \(F\) almost surely in the first Wasserstein distance:
\[
W_1(F_B,F)\longrightarrow 0
\quad\text{almost surely}.
\]
This follows from the Glivenko--Cantelli theorem together with the strong law of large numbers for the first absolute moment.

For any integrable random variable \(Y\), CVaR admits the dual representation
\[
\mathrm{CVaR}_\alpha(Y)
=
\sup\left\{
\mathbb E(YQ):
0\le Q\le \frac{1}{1-\alpha},\ \mathbb E Q=1
\right\}.
\]
Equivalently, as a functional of a probability distribution, \(\mathrm{CVaR}_\alpha\) is Lipschitz with respect to the first Wasserstein distance:
\[
\left|
\mathrm{CVaR}_\alpha(F_B)
-
\mathrm{CVaR}_\alpha(F)
\right|
\le
\frac{1}{1-\alpha}W_1(F_B,F).
\]
For completeness, this inequality can be seen from the fact that CVaR averages the upper tail quantile function:
\[
\mathrm{CVaR}_\alpha(F)
=
\frac{1}{1-\alpha}\int_\alpha^1 F^{-1}(u)\,du,
\]
where \(F^{-1}(u)=\inf\{y:F(y)\ge u\}\).  Therefore,
\[
\begin{aligned}
\left|
\mathrm{CVaR}_\alpha(F_B)-\mathrm{CVaR}_\alpha(F)
\right|
&=
\left|
\frac{1}{1-\alpha}\int_\alpha^1
\{F_B^{-1}(u)-F^{-1}(u)\}\,du
\right|\\
&\le
\frac{1}{1-\alpha}\int_0^1
|F_B^{-1}(u)-F^{-1}(u)|\,du\\
&=
\frac{1}{1-\alpha}W_1(F_B,F).
\end{aligned}
\]
Since \(W_1(F_B,F)\to0\) almost surely, it follows that
\[
\mathrm{CVaR}_\alpha(F_B)\to \mathrm{CVaR}_\alpha(F)
\quad\text{almost surely}.
\]
Finally, the empirical Rockafellar--Uryasev minimization formula is exactly the CVaR of the empirical distribution \(F_B\):
\[
\widehat{\mathrm{CVaR}}_{\alpha,B}
=
\mathrm{CVaR}_\alpha(F_B).
\]
Thus
\[
\widehat{\mathrm{CVaR}}_{\alpha,B}
\longrightarrow
\mathrm{CVaR}_\alpha(Y)
\quad\text{almost surely}.
\]
\end{proof}

\subsection{Proof of Theorem~\ref{thm:uniform_saa}}
\label{app:proof_uniform_saa}

\begin{proof}
Fix an arbitrary feasible design \(z\in\mathfrak X\).  We prove convergence of each component of \(\widehat\rho_{T,B}(z)\).

First, by Assumption~\ref{ass:integrability},
\[
\mathbb E_T|C(z,\Xi)|<\infty.
\]
Since \(\Xi^{(1)},\Xi^{(2)},\ldots\) are iid from \(P_T\), the strong law of large numbers gives
\[
\widehat C_B(z)
=
\frac1B\sum_{b=1}^{B}C(z,\Xi^{(b)})
\longrightarrow
\mathbb E_T\{C(z,\Xi)\}
\quad\text{almost surely}.
\]

Second, again by Assumption~\ref{ass:integrability},
\[
\mathbb E_T|A(z,\Xi)|<\infty.
\]
Applying Lemma~\ref{lem:cvar_consistency} to
\[
Y=A(z,\Xi)
\]
gives
\[
\widehat{\mathcal T}_{\alpha,B}(z)
\longrightarrow
\mathcal T_\alpha(z)
\quad\text{almost surely}.
\]

Third, define the Bernoulli random variables
\[
L_b(z)=\mathbbm 1\{A(z,\Xi^{(b)})>T^\star\}.
\]
They are iid and bounded by one.  Hence the strong law of large numbers yields
\[
\widehat p_{L,B}(z)
=
\frac1B\sum_{b=1}^{B}L_b(z)
\longrightarrow
P_T\{A(z,\Xi)>T^\star\}
\quad\text{almost surely}.
\]
Similarly, for every candidate hub \(k\in H\),
\[
H_{k,b}(z)=\mathbbm 1\{H_k(z,\Xi^{(b)})>d_t\}
\]
is iid Bernoulli and bounded.  Therefore
\[
\widehat p_{H,k,B}(z)
\longrightarrow
P_T\{H_k(z,\Xi)>d_t\}
\quad\text{almost surely}.
\]

Combining these componentwise convergences and using the finiteness of the weights in \eqref{eq:formal_bayes_risk} and \eqref{eq:empirical_risk_formal}, we obtain
\[
\widehat\rho_{T,B}(z)
\longrightarrow
\rho_T(z)
\quad\text{almost surely}
\]
for the fixed design \(z\).

It remains to upgrade pointwise convergence to uniform convergence over \(\mathfrak X\).  By Assumption~\ref{ass:finite_design}, \(\mathfrak X\) is finite.  Let
\[
\mathfrak X=\{z_1,\ldots,z_M\}
\]
for some finite \(M\).  For each \(m=1,\ldots,M\), there exists an event \(\Omega_m\) with probability one on which
\[
\widehat\rho_{T,B}(z_m)\to\rho_T(z_m).
\]
The intersection
\[
\Omega_0=\bigcap_{m=1}^{M}\Omega_m
\]
also has probability one because it is a finite intersection of probability-one events.  On \(\Omega_0\),
\[
\max_{1\le m\le M}
\left|
\widehat\rho_{T,B}(z_m)-\rho_T(z_m)
\right|
\longrightarrow 0.
\]
This is exactly
\[
\sup_{z\in\mathfrak X}
\left|
\widehat\rho_{T,B}(z)-\rho_T(z)
\right|
\longrightarrow 0
\quad\text{almost surely}.
\]
\end{proof}

\subsection{Proof of Theorem~\ref{thm:argmin_consistency}}
\label{app:proof_argmin}

\begin{proof}
By Theorem~\ref{thm:uniform_saa},
\[
\sup_{z\in\mathfrak X}
|\widehat\rho_{T,B}(z)-\rho_T(z)|
\longrightarrow 0
\quad\text{almost surely}.
\]
Let \(z_T^\star\) be the unique Bayes-optimal design and define the strict risk gap
\[
\Delta_T
=
\min_{z\in\mathfrak X:\,z\ne z_T^\star}
\{\rho_T(z)-\rho_T(z_T^\star)\}.
\]
By assumption, \(\Delta_T>0\).

On the almost-sure event of uniform convergence, there exists a random integer \(B_0\) such that for all \(B\ge B_0\),
\[
\sup_{z\in\mathfrak X}
|\widehat\rho_{T,B}(z)-\rho_T(z)|
<
\frac{\Delta_T}{3}.
\]
Fix \(B\ge B_0\).  For any \(z\ne z_T^\star\),
\[
\widehat\rho_{T,B}(z)
\ge
\rho_T(z)-\frac{\Delta_T}{3}.
\]
By the definition of \(\Delta_T\),
\[
\rho_T(z)\ge \rho_T(z_T^\star)+\Delta_T.
\]
Therefore,
\[
\widehat\rho_{T,B}(z)
\ge
\rho_T(z_T^\star)+\Delta_T-\frac{\Delta_T}{3}
=
\rho_T(z_T^\star)+\frac{2\Delta_T}{3}.
\]
On the other hand,
\[
\widehat\rho_{T,B}(z_T^\star)
\le
\rho_T(z_T^\star)+\frac{\Delta_T}{3}.
\]
Thus, for every \(z\ne z_T^\star\),
\[
\widehat\rho_{T,B}(z)
>
\widehat\rho_{T,B}(z_T^\star).
\]
Hence \(z_T^\star\) is the unique minimizer of the empirical risk \(\widehat\rho_{T,B}\) for all \(B\ge B_0\).  Therefore any empirical optimizer
\[
\widehat z_{T,B}\in\arg\min_{z\in\mathfrak X}\widehat\rho_{T,B}(z)
\]
satisfies
\[
\widehat z_{T,B}=z_T^\star
\]
for all sufficiently large \(B\), almost surely.
\end{proof}

\subsection{Proof of Theorem~\ref{thm:posterior_topology_stability}}
\label{app:proof_topology_stability}

\begin{proof}
The proof has three steps: convergence of risks, preservation of the unique minimizer, and convergence of topology.

\medskip
\noindent\textbf{Step 1: Convergence of posterior risks.}
Fix a design \(z\in\mathfrak X\).  By Assumption~\ref{ass:posterior_concentration},
\[
W_1(P_T,P_0)\to0
\]
in \(P_0\)-probability.  By Assumption~\ref{ass:continuity_domination}, \(C(z,\xi)\) is continuous outside a \(P_0\)-null set and dominated by an integrable envelope.  Therefore, by the dominated convergence theorem together with weak convergence implied by \(W_1\)-convergence,
\[
\mathbb E_T\{C(z,\Xi)\}
\longrightarrow
\mathbb E_0\{C(z,\Xi)\}
\]
in probability.

The same argument applies to the arrival-time random variable \(A(z,\Xi)\).  In addition, \(W_1(P_T,P_0)\to0\) and integrability imply convergence of the induced distributions of \(A(z,\Xi)\) in first Wasserstein distance.  By the Lipschitz property of CVaR with respect to \(W_1\), used in the proof of Lemma~\ref{lem:cvar_consistency},
\[
\mathrm{CVaR}_{\alpha,T}\{A(z,\Xi)\}
\longrightarrow
\mathrm{CVaR}_{\alpha,0}\{A(z,\Xi)\}
\]
in probability.

Next consider the service-violation probability.  Since
\[
P_0\{A(z,\Xi)=T^\star\}=0,
\]
the indicator function
\[
\xi\mapsto \mathbbm 1\{A(z,\xi)>T^\star\}
\]
is continuous \(P_0\)-almost surely.  It is also bounded by one.  Hence the portmanteau theorem gives
\[
P_T\{A(z,\Xi)>T^\star\}
\longrightarrow
P_0\{A(z,\Xi)>T^\star\}
\]
in probability.  Similarly, for each hub \(k\), the assumption
\[
P_0\{H_k(z,\Xi)=d_t\}=0
\]
implies
\[
P_T\{H_k(z,\Xi)>d_t\}
\longrightarrow
P_0\{H_k(z,\Xi)>d_t\}
\]
in probability.

Combining these convergences, we obtain
\[
\rho_T(z)\longrightarrow \rho_0(z)
\]
in probability for every fixed \(z\in\mathfrak X\).

\medskip
\noindent\textbf{Step 2: Uniform convergence over the finite design space.}
By Assumption~\ref{ass:finite_design}, \(\mathfrak X\) is finite.  Write
\[
\mathfrak X=\{z_1,\ldots,z_M\}.
\]
Since \(\rho_T(z_m)\to\rho_0(z_m)\) in probability for each \(m\), finite-dimensional convergence implies
\[
\max_{1\le m\le M}
|\rho_T(z_m)-\rho_0(z_m)|
\longrightarrow 0
\]
in probability.  Equivalently,
\[
\sup_{z\in\mathfrak X}
|\rho_T(z)-\rho_0(z)|
\longrightarrow 0
\]
in probability.

\medskip
\noindent\textbf{Step 3: Preservation of the unique limiting minimizer.}
By assumption, \(z_0^\star\) is the unique minimizer of \(\rho_0\) over \(\mathfrak X\).  Since \(\mathfrak X\) is finite, the limiting risk gap
\[
\Delta_0
=
\min_{z\in\mathfrak X:\,z\ne z_0^\star}
\{\rho_0(z)-\rho_0(z_0^\star)\}
\]
is strictly positive.

Let
\[
E_T=
\left\{
\sup_{z\in\mathfrak X}
|\rho_T(z)-\rho_0(z)|
<
\frac{\Delta_0}{3}
\right\}.
\]
From uniform convergence in probability,
\[
P(E_T)\to1.
\]
On the event \(E_T\), for every \(z\ne z_0^\star\),
\[
\rho_T(z)
\ge
\rho_0(z)-\frac{\Delta_0}{3}
\ge
\rho_0(z_0^\star)+\Delta_0-\frac{\Delta_0}{3}
=
\rho_0(z_0^\star)+\frac{2\Delta_0}{3}.
\]
Also,
\[
\rho_T(z_0^\star)
\le
\rho_0(z_0^\star)+\frac{\Delta_0}{3}.
\]
Therefore, for every \(z\ne z_0^\star\),
\[
\rho_T(z)>\rho_T(z_0^\star).
\]
Thus \(z_0^\star\) is the unique posterior Bayes-risk minimizer on \(E_T\).  Consequently,
\[
P\{z_T^\star=z_0^\star\}\ge P(E_T)\to1.
\]
Writing \(z_T^\star=(g_T^\star,x_T^\star)\) and \(z_0^\star=(g_0^\star,x_0^\star)\), this implies
\[
P\{g_T^\star=g_0^\star\}\to1.
\]
Hence the selected topology converges in probability to the limiting optimal topology.
\end{proof}




\end{document}